\documentclass[11pt,a4paper]{article}
\usepackage[pageanchor,hypertexnames=false]{hyperref}
\usepackage[dvipsnames]{xcolor}
\usepackage{amsmath,amssymb,amsbsy,amsthm,calc,cases,enumerate,url,framed,fancyhdr,lipsum,tikz,enumitem,natbib,etoolbox,bbm}
\usepackage[mathscr]{euscript}
\usepackage{xcolor}
\hypersetup{colorlinks,linkcolor={red!70!black},citecolor={blue!65!black},urlcolor={blue!70!black}}
\usepackage{graphicx}
\usepackage{epstopdf}
\usepackage[linesnumbered,ruled]{algorithm2e}
\usepackage[font=small,skip=0pt]{caption}
\usepackage[all]{hypcap}

\newcommand{\floor}[1]{\lfloor #1 \rfloor}

\newcommand{\Ind}{\mathbbm{1}}

\newcommand\independent{\protect
\mathpalette{\protect\independenT}{\perp}}
\def\independenT#1#2{\mathrel{\rlap{$#1#2$}\mkern2mu{#1#2}}}
\newcommand*\circled[1]{\tikz[baseline=(char.base)]{
            \node[shape=circle,draw,inner sep=1pt] (char) {#1};}}

\DeclareRobustCommand\full {\tikz[baseline=-0.5ex]\draw[blue,thick] (0,0)--(0.5,0);}
\DeclareRobustCommand\dotted{\tikz[baseline=-0.5ex]\draw[blue,thick,dotted] (0.06,0)--(0.5,0);}
\DeclareRobustCommand\dashed{\tikz[baseline=-0.5ex]\draw[black!30!green,thick,dashed] (0,0)--(0.52,0);}
\DeclareRobustCommand\sqr[1]{\tikz{\node[draw=#1,fill=#1, fill opacity = 0.5,rectangle,minimum
width=0.2cm,minimum height=0.2cm,inner sep=0pt] at (0,0) {};}}

\DeclareRobustCommand\fullgrey {\tikz[baseline=-0.5ex]\draw[black,thick, opacity = 0.5] (0,0)--(0.5,0);}
\DeclareRobustCommand\dashedgrey{\tikz[baseline=-0.5ex]\draw[black,thick,dashed, opacity = 0.5] (0,0)--(0.52,0);}

\theoremstyle{plain} 
\newtheorem{theorem}{Theorem}
\newtheorem{proposition}[theorem]{Proposition}
\newtheorem*{proposition*}{Proposition}
\newtheorem{lemma}[theorem]{Lemma}
\newtheorem*{lemma*}{Lemma}
\newtheorem{corollary}[theorem]{Corollary}

\newtheorem*{claim*}{Claim}
\newtheorem{definition}{Definition}
\newtheorem*{definition*}{Definition}
\newtheorem{assumption}{Assumption}

\newtheorem*{question*}{Question}

\newtheorem*{conjecture*}{Conjecture}

\theoremstyle{definition} 

\newtheorem*{remark*}{Remark}

\newtheorem*{aside*}{Aside}
\newtheorem{example}[theorem]{Example}
\newtheorem*{example*}{Example}
\newtheorem*{aim*}{Aim}

\newtheorem*{observation*}{Observation}

\begin{document}

\title{Semiparametric posterior corrections}

\author{Andrew Yiu$^{1}$, Edwin Fong$^{2}$, Chris Holmes$^{1,3}$, 
and Judith Rousseau$^{1}$ \\ \\
$^1$Department of Statistics, University of Oxford\\
$^2$Novo Nordisk\\
$^3$The Alan Turing Institute}

\maketitle

\begin{abstract}
We present a new approach to semiparametric inference using corrected posterior distributions. The method allows us to leverage the adaptivity, regularization and predictive power of nonparametric Bayesian procedures to estimate low-dimensional functionals of interest without being restricted by the holistic Bayesian formalism. Starting from a conventional nonparametric posterior, we target the functional of interest by transforming the entire distribution with a Bayesian bootstrap correction. We provide conditions for the resulting \textit{one-step posterior} to possess calibrated frequentist properties and specialize the results for several canonical examples: the integrated squared density, the mean of a missing-at-random outcome, and the average causal treatment effect on the treated. The procedure is computationally attractive, requiring only a simple, efficient post-processing step that can be attached onto any arbitrary posterior sampling algorithm. Using the ACIC 2016 causal data analysis competition, we illustrate that our approach can outperform the existing state-of-the-art through the propagation of Bayesian uncertainty.
\end{abstract}

\section{Introduction}

Functional estimation is emerging as a major paradigm in modern statistics \citep{vanderLaan11, Buja19a, Buja19b, Vansteelandt22}. The definition of the estimand---the target of inference---is placed at the forefront of the statistical analysis, requiring the user to carefully articulate their scientific objectives. This helps to avoid the common pitfall in which the choice of estimand is informed by the choice of statistical model, rather than the other way round. Causal inference, in particular, has a natural affinity with functional estimation through the nonparametric identification of causal effects \citep{Pearl09, Hernan20}. 

An important advantage of the functional approach is that we are free to estimate low-dimensional functionals without being restricted to low-dimensional models. Nonparametric methods provide greater flexibility in capturing complex relationships between variables, mitigating the risks of model misspecification bias. In particular, it is becoming increasingly popular to leverage highly flexible black-box algorithms (e.g. random forests, gradient boosting, deep learning) to estimate nuisance parameters \citep{vanderLaan06, Chernozhukov18, Vansteelandt22}. 

Using infinite-dimensional models to estimate low-dimensional estimands is often referred to as \textit{semiparametric inference}. By exploiting the smoothness (or \textit{differentiability}) of the target functional, it may be possible to obtain estimators with parametric properties like $\sqrt{n}$-consistency and asymptotic normality, even if the nuisance parameter estimators converge at a slower rate \citep{Murphy00, vanderLaan11, Kennedy22, Hines22}. Some authors have suggested that semiparametric methods can bridge the ``two cultures'' of statistical modelling \citep{Breiman01} by combining the interpretability and efficiency of parametric inference with the robustness of black-box algorithms \citep{Ogburn21,Kennedy21,Vansteelandt21}.

Nonparametric Bayesian approaches such as \textit{Bayesian additive regression trees} \citep{Chipman07} have performed especially well in empirical studies for both prediction and inferring low-dimensional estimands such as causal effects \citep{Hill11, Dorie19, Hahn20}. For parametric Bayesian models, the influence of the prior disappears asymptotically to first-order thanks to the renowned \textit{Bernstein-von Mises theorem} \citep{LeCam86, vanderVaart98}. For infinite-dimensional models however, the impact of the prior may remain significant in the asymptotic regime, and the prior specification cannot be completely justified by subjective beliefs \citep{diaconis:freedman86, Rousseau16}. This has led to a rich literature on the frequentist properties of Bayesian nonparametric methods, with a strong emphasis on posterior contraction rates \citep[e.g.][]{Ghosal17}. In particular, many Bayesian approaches have been shown to enjoy excellent properties in terms of adaptive posterior contraction rates due   to the flexibility enabled by hierarchical prior modelling \citep[e.g.][]{Rousseau16, Ghosal17}.

Bayesian approaches naturally deal with nuisance parameters by integrating them out, which allows for coherent inference on possibly multiple parameters of interest, whether they are estimated simultaneously or sequentially \citep[e.g.][]{berger:liseo:wolpert99}.  Unfortunately, good adaptive contraction rates do not necessarily translate to good behaviour for marginal posterior distributions of specific functionals, even to first-order.  Although there has been a growing literature on the existence of Bernstein-von Mises theorems in the semiparametric setting \citep{Castillo12,Rivoirard12,castillo:nickl14,Castillo15}, the results have been only partially positive, and it is now recognized that a given prior will perform well for some functionals of interest but not for others. This is perhaps unsurprising given that Bayesian inference is a ``plug-in'' approach, and similar phenomena have been observed for semiparametric plug-in estimators in the frequentist literature \citep{Bickel03, vanderLaan06, Robins17}. The reason is that flexible nonparametric methods must employ regularization to achieve good performance, and the ``regularization bias'' can bleed into the corresponding plug-in estimator, precluding fast convergence rates \citep{vanderLaan06, vanderVaart14, Chernozhukov18}.

For some particular examples, it has been demonstrated that the asymptotic posterior bias for a nonparametric Bayesian model can be controlled by tailoring the prior specification to the estimand \citep{Castillo15, Ray20}. This is undesirable from a practical perspective however, as it will require careful tuning (or even rewriting) of any existing software for posterior computation. Moreover, a different modification will be required for each estimand of interest. In any case, studying the existence of a Bernstein-von Mises theorem remains a formidable open question for many popular families of priors, such as nonparametric mixtures. For the implementations of nonparametric Bayesian procedures that are widely applied in practice, the existing empirical evidence suggests that the lack of a semiparametric Bernstein-von Mises theorem represents the rule, rather than the exception \citep{Dorie19, RaySzabo19, Hahn20}.

To address these issues, we introduce a simple post-processing procedure that starts from a given posterior distribution on the whole data-generating distribution and then corrects the marginal posterior for each functional of interest. This operates by adding a stochastic term based on the \textit{efficient influence function} of the functional \citep{Pfanzagl82, vanderVaart91, vanderVaart98} and the Bayesian bootstrap \citep{Rubin81}, which plays a crucial role for correcting not only the bias of the posterior but also its shape. The original user-specified Bayesian model can be left untouched, acting as a central hub from which we can target each estimand of interest individually. 

We provide general conditions for these new \textit{one-step posteriors} to satisfy a Bernstein-von Mises theorem, which ensures that central credible regions achieve approximately nominal coverage with the semiparametric efficient size. Unlike the semiparametric Bernstein-von Mises theorems derived for standard nonparametric posteriors \citep[e.g.][]{Castillo12, Castillo15}, our conditions do not depend delicately on the likelihood and prior. Instead, we require only contraction rates and complexity bounds on the posterior with empirical process theory, leading to a Bayesian counterpart to the classical theory of influence-function-based estimation \citep{Pfanzagl82, vanderLaan03, vanderLaan11, Hines22, Kennedy22}. In particular, the one-step posterior could be interpreted as a natural analogue of the \textit{one-step estimator} \citep{Pfanzagl82, Newey98} for correcting an entire posterior distribution rather than just a point estimator. 

Our methodology, which we introduce in Section \ref{sec::meth}, is  computationally efficient and attaches onto any existing posterior sampling implementation without modification of the original algorithm. The procedure takes each posterior sample and adds a randomized correction term, which only requires drawing an independent set of uniform Dirichlet weights of lengths equal to the sample size. A single set of posterior samples can be retained for simultaneous inference of multiple functionals of interest.

We apply our approach to the classic example of estimating the integrated squared density in Section \ref{sec::dens_squared}. This was previously studied in a Bayesian setting by \citet{Castillo15} but only for random histogram priors on the interval $[0,1]$. We verify the conditions for the more complex class of Dirichlet process Gaussian location mixtures \citep{Shen13}. To the best our knowledge, this is the first Bernstein-von Mises theorem associated with priors based on nonparametric mixture models.

In Section \ref{sec::miss}, we study the estimation of the mean of an outcome that is missing-at-random. This is a problem that has received much attention due to its connections with estimating the average treatment effect in causal inference \citep{Kang07, Robins17}. The conditions for the Bernstein-von Mises theorem are expressed here in terms of the propensity score and outcome regression models, showing that the one-step posterior possesses a doubly robust property, i.e. it depends on the combined contraction rates of both parameters. For this specific example in the binary outcome setting, \citet{Ray20} proposed a correction approach that has similar motivations and asymptotic considerations to ours, but it requires a modification to the prior that can be relatively complicated to implement. Working within their set-up, we study the behaviour of the one-step posterior and provide a comparison of theoretical assumptions.


The last example is the average treatment effect on the treated \citep{Rubin77, Heckman85}, which has not been studied before from a large-sample Bayesian perspective. Due to its complex structure, this estimand poses some theoretical challenges. In Section \ref{sec::att}, we introduce a slight modification of the correction algorithm, taking inspiration from estimating equation methodology. The resulting corrected posterior also satisfies a doubly robust property. Extensions of this approach to sample conditional treatment effects can be found in Section \ref{sec::samp_treat}.

We evaluate our methodology empirically in Section \ref{sec::sim}. The first simulation study validates the approach of Section \ref{sec::dens_squared}, exhibiting the improved performance of the one-step posterior relative to the uncorrected posterior, which is heavily biased and exhibits substantial undercoverage. This is followed by a large-scale comparison with state-of-the-art causal algorithms using the ACIC 2016 data anaylsis competition \citep{Dorie19}. Our method---applied to \textit{Bayesian additive regression trees} \citep{Chipman07}---is the best performer in terms of both point and interval estimation. In fact, no other method was able to obtain nominal coverage across all of the data-generating mechanisms. We suggest that this superior performance could be attributed to the propagation of Bayesian uncertainty through the posterior correction to regularize the estimation, which cannot be exploited by the classical semiparametric approaches. Thus, we argue that the one-step posterior correction should appeal not only to Bayesians, but also to classically-minded users interested in leveraging the power of Bayesian algorithms while obtaining frequentist-calibrated inference.

\section{Methodology} \label{sec::meth}

\subsection{Set-up and background} \label{sec::setup}

Suppose that we observe independent and identically distributed data $Z^{(n)} = (Z_{1}, \ldots, Z_{n})$ from a distribution $P_{0}$ known to belong to a set of probability measures $\mathcal{P}$ on a Polish sample space $(\mathcal{Z}, \mathcal{A})$. We will use the shorthand notation $P[f] = \int f(z)\,dP(z)$ for $P \in \mathcal{P}$. In particular, $\mathbb{P}_{n}[f] = n^{-1}\sum_{i=1}^{n} f(Z_{i})$, where $\mathbb{P}_{n}$ is the empirical measure, and we let $\mathbb{G}_{n} = \sqrt{n}(\mathbb{P}_{n}-P_{0})$ denote the empirical process. For $P \in \mathcal{P}$, let $L_{2}(P)$ be the Hilbert space consisting of all measurable functions $h:\mathcal{Z}\rightarrow \mathbb{R}$ with $P[h^{2}] < \infty$ equipped with the inner product $\langle h_{1}, h_{2} \rangle_{P} = P[h_{1}h_{2}]$ and norm $\Vert h \Vert_{P} = \sqrt{P[h^{2}]}$. 

In a Bayesian analysis, we equip $\mathcal{P}$ with a $\sigma$-field $\mathcal{B}$ such that $(\mathcal{P},\mathcal{B})$ is a standard Borel space, and we specify a prior probability distribution $\Pi$ on $P$. Let $\Pi(\cdot \mid Z^{(n)})$ denote a version of the posterior distribution given $Z^{(n)}$. The target estimand is $\chi(P_{0})$, where $\chi: \mathcal{P} \rightarrow \mathbb{R}$ is a measurable functional. For notational clarity, we have chosen the range space to be one-dimensional; generalizing our results to higher dimensions is straightforward (e.g. by applying the Cram\`er-Wold device: see p. 16 of \citet{vanderVaart98}). This set-up includes ``separated'' semiparametric problems in which we have $\mathcal{P} = \{P_{\theta,\eta}\}$ and $\chi(P_{\theta,\eta}) = \theta$, where $\theta \in \mathbb R^k$, for some $k\geq 1$ and $\eta$ is a possibly infinite-dimensional nuisance parameter.

An important special case is $\mathcal{P} = M(\mathcal{Z})$---the set of all probability measures on $(\mathcal{Z}, \mathcal{A})$---equipped with its Borel $\sigma$-field $\mathcal{B}_{M}$ for the weak topology \citep{Ghosh03}. Dirichlet processes \citep{Ferguson73} form the canonical family of prior distributions on $(M(\mathcal{Z}), \mathcal{B}_{M})$. If $P \sim DP(\alpha)$ for a finite positive measure $\alpha$ on $(\mathcal{Z}, \mathcal{A})$, then $DP(\alpha + n\mathbb{P}_{n})$ is a version of the posterior given $Z^{(n)}$. The process $DP(n\mathbb{P}_{n})$ is known as the \textit{Bayesian bootstrap} \citep{Rubin81}, which can be interpreted as the noninformative limiting posterior for any sequence of Dirichlet process models in which the prior information $\alpha$ goes to zero. The Bayesian bootstrap plays a central role in our methodology; we denote it by $\Pi_{BB}(\cdot \mid Z^{(n)})$. 

It is now necessary to introduce some concepts from semiparametric theory. We refer the reader to Section \ref{supp::overview} for a detailed overview of the following definitions. The functional $\chi$ is required to be \textit{differentiable} at all $P \in \mathcal{P}$ in the sense of \citet{vanderVaart91}. The \textit{efficient influence function} of $\chi$ at $P$ is a measurable function $\dot{\chi}_{P}: \mathcal{Z} \rightarrow \mathbb{R}$ satisfying $P[\dot{\chi}_{P}] = 0$ and $\Vert \dot{\chi}_{P} \Vert_{P} < \infty$. A sequence of regular estimators $\hat{\chi}_{n} = \hat{\chi}_{n}(Z^{(n)})$ is said to be \textit{asymptotically efficient} at $P_{0}$ if
\begin{equation} \label{eqn::asympeff}
    \sqrt{n}(\hat{\chi}_{n}-\chi(P_{0})) = \sqrt{n}\mathbb{P}_{n}[\dot{\chi}_{P_{0}}] + o_{P_{0}}(1).
\end{equation}
The resulting limiting distribution $\mathcal{N}(0, \Vert \dot{\chi}_{P_{0}} \Vert^{2}_{P_{0}})$ is optimal in terms of both the \textit{convolution theorem} \citep[e.g. Theorem 25.20 of][]{vanderVaart98} and the \textit{local asymptotic minimax theorem} \citep{vanderVaart92}. 

For smooth parametric models, it is well-known that---under mild regularity conditions---the marginal posterior of the parameter converges to a normal distribution centred at an efficient estimator (e.g. the maximum likelihood estimator) with covariance equal to the inverse Fisher information divided by the sample size. This is the \textit{(efficient) Bernstein-von Mises} (BvM) theorem (Chapter 10 of \citet{vanderVaart98}). Consequently, the frequentist coverage of many credible regions---such as balls around the posterior mean, equal tails intervals, highest posterior density regions etc.---will be asymptotically equal to its nominal level, and their sizes will be asymptotically optimal. 

A semiparametric BvM theorem describes the marginal posterior of the target functional on an infinite-dimensional model. Similar to the parametric case, we would like the marginal posterior to converge to a normal distribution centred at an efficient estimator with the efficient asymptotic variance. In this article,  we study weak convergence of the posterior distributions, which we control using  the \textit{bounded Lipschitz distance} $d_{BL}$ as in \cite{Castillo15} (see also Section \ref{supp::weak_conv}).
\begin{definition} \label{def::semi_bvm}
Let $\mathcal{L}_{\Pi}(\sqrt{n}(\chi(P) - \hat{\chi}_{n}) \mid Z^{(n)})$ denote the posterior law of $\sqrt{n}(\chi(P) - \hat{\chi}_{n})$, where $\hat{\chi}_{n}$ is any sequence of estimators satisfying (\ref{eqn::asympeff}). The posterior satisfies the \textit{semiparametric Bernstein-von Mises theorem} if
\begin{equation*}
    d_{BL}\left(\mathcal{L}_{\Pi}(\sqrt{n}(\chi(P) - \hat{\chi}_{n}) \mid Z^{(n)}), \mathcal{N}(0, \Vert \dot{\chi}_{P_{0}} \Vert^{2}_{P_{0}})\right)\xrightarrow[]{P_{0}} 0.
\end{equation*}
More informally, we can say that the posterior law of $\sqrt{n}(\chi(P) - \hat{\chi}_{n})$ ``converges weakly to $\mathcal{N}(0, \Vert \dot{\chi}_{P_{0}} \Vert^{2}_{P_{0}})$ in probability''. 
\end{definition}

In the non-Bayesian setting, efficient estimators have been obtained by \textit{one-step estimation} \citep{Pfanzagl82, Newey98}, which operates by performing bias corrections that specifically target the estimand. Suppose we have an initial estimate $\hat{P}$ of the data distribution, which induces a \textit{plug-in estimate} $\chi(\hat{P})$. Since the efficient influence function acts as a first-order distributional derivative of the functional \citep{Robins17,Fisher21,Kennedy22}, we might expect the ``second-order remainder''
\begin{equation*}
    r_{2}(P_{0},P) =  \chi(P_{0})-\chi(\hat{P}) - P_{0}[\dot{\chi}_{\hat{P}}]
\end{equation*}
to have some form of quadratic dependence on the estimation error of $\hat{P}$. This suggests that $\chi(\hat{P}) + P_{0}[\dot{\chi}_{\hat{P}}]$ would be an improvement over the plug-in estimator, but the true data distribution is of course unknown. Replacing $P_{0}$ with the empirical measure $\mathbb{P}_{n}$ leads to
\begin{equation*}
    \hat{\chi}_{\text{1-step}} = \chi(\hat{P})+\mathbb{P}_{n}[\dot{\chi}_{\hat{P}}].
\end{equation*}
which is called the \textit{one-step estimator}. We review sufficient conditions for the one-step estimator to be asymptotically efficient in Section \ref{sec::grad_based}; as we will discuss in Section \ref{sec::theory}, the assumptions for our approach share key conceptual similarities.

\begin{example}[The mean of an outcome that is missing-at-random] \label{exa::miss}
Suppose the data takes the form $Z = (X,A,AY)$, where $X$ is a vector of covariates, $A$ is a binary missingness indicator, and $Y$ is the outcome variable of interest that is unobserved when $A=0$. The target estimand is the outcome mean $\chi(P) = E_{P}[Y]$, which is identified by
\begin{equation*}
    \chi(P) =  E_{P}[E_{P}[Y \mid A=1, X]]
\end{equation*}
under the \textit{ignorability} (or \textit{missing-at-random}) \citep{Rubin76} assumption $A \independent Y \mid X$, and the \textit{positivity} (or \textit{overlap)} assumption:
\begin{equation*}
    P(A=1 \mid X) > 0
\end{equation*}
with $P$-probability 1. In practical terms, the missingness is deemed to be uninformative with respect to the outcome after adjusting for a set of measured covariates. This example is closely related to estimating the average causal effect of a binary treatment in causal inference (see Section \ref{sec::samp_treat} or \citet{Morgan07}).

The data distribution $P$ is naturally parameterized in terms of the ``propensity score'' $\pi(x) = P(A=1 \mid X=x)$, the conditional outcome density $p_{Y\mid X}(y \mid x)$, and the marginal covariate distribution $Q$. Let $m(x) = E_{P}[Y \mid X=x]$ denote the ``outcome regression function'', such that the target estimand can be written as
\begin{equation*}
    \chi(P) = Q[m(X)].
\end{equation*}
Thus, a plug-in estimator of $\chi$ takes the form
\begin{equation*}
    \chi(\hat{P}) = \hat{Q}[\hat{m}(X)]
\end{equation*}
for estimates $\hat{Q}$ and $\hat{m}$ of $Q$ and $m$ respectively.

The efficient influence function of $\chi$ is
\begin{equation*}
    \dot{\chi}_{P}(Z) = \frac{A}{\pi(X)}\{Y - m(X)\} + m(X) - \chi(P),
\end{equation*}
which was first derived in \citet{Robins92}. As a result, the one-step estimator can be easily constructed as
\begin{equation*}
    \hat{\chi}_{\text{1-step}} = \mathbb{P}_{n}\left[\frac{A}{\hat{\pi}(X)}\{Y-\hat{m}(X)\} + \hat{m}(X)\right],
\end{equation*}
where we have avoided the need to estimate $Q$ but now require an estimate $\hat{\pi}$ of the propensity score. In Section \ref{sec::miss}, we will study this example in the Bayesian setting.

\end{example}

\subsection{The one-step posterior correction} \label{sec::1step_cor}

If $\chi$ is well-defined for discrete distributions, its marginal posterior under a Dirichlet process prior on $P$ generally satisfies the semiparametric BvM theorem (see, for example, Theorem 12.2 of \citet{Ghosal17} and the subsequent discussion). However, the functionals in semiparametric problems are usually defined with respect to restrictions on the data-generating distribution that require some smoothness, in which case Dirichlet processes are inappropriate. For instance, we might require $\mathcal{P}$ to be dominated by the Lebesgue measure (e.g. Section \ref{sec::dens_squared}), or we might impose positivity conditions as described in Example \ref{exa::miss} (see also: Sections \ref{sec::miss} and \ref{sec::att}). 

As discussed in the introduction, nonparametric smoothing in our prior may introduce a non-negligible bias in the posterior on a low-dimensional functional that precludes a BvM theorem. In some settings, it has been shown that this bias can be avoided by undersmoothing \citep{Castillo12, Castillo15, Ray20}; that is, we must ensure that the true parameters are at least as regular as the prior, e.g. support on sufficiently smooth H\"older and Sobolev spaces \citep[Appendix C of][]{Ghosal17}. This is impractical, however, because the true regularities are unknown and the recovery of the nuisance parameters is degraded by overfitting, which can be detrimental to finite-sample performance. Moreover, undersmoothing may require non-trivial modifications to existing implementations of Bayesian computation methods.

Instead of tuning the prior, we propose an approach that adjusts the posterior directly. Our starting point is an ordinary Bayesian model that is---as usual---optimized to recover the whole data-generating distribution $P$ rather than any particular functional. This means that the prior can be specified according to conventional guidelines, and the posterior $\Pi(\cdot \mid Z^{(n)})$ can be computed with standard methods. As discussed later, our approach will also allow us to circumvent some of the limitations of likelihood-based methods.

\begin{definition}
The \textit{one-step posterior} of the functional $\chi$ given $Z^{(n)}$ is the pushforward measure of
\begin{align*}
\Tilde{\chi}: \mathcal{P} \times M(\mathcal{Z}) &\rightarrow \mathbb{R} \\
 \Tilde{\chi}(P,\Tilde{P}) &= \chi(P) + \Tilde{P}[\dot{\chi}_{P}].
\end{align*}
on the product measure
\begin{equation*}
    (P, \Tilde{P}) \mid Z^{(n)} \sim \Pi(\cdot \mid Z^{(n)}) \times \Pi_{BB}(\cdot \mid Z^{(n)}).
\end{equation*}
\end{definition} 

The structure of the corrected parameter $\Tilde{\chi}$ takes inspiration from the one-step estimator described in the previous subsection. Informally speaking, the initial estimate $\hat{P}$ has been replaced by the parameter $P$ drawn from the initial posterior, and the empirical measure $\mathbb{P}_{n}$ has been replaced by the Bayesian bootstrap parameter $\Tilde{P}$. In our case, we are correcting an entire probability distribution rather than just a point estimator. 

We can sample from the one-step posterior via the simple hierarchical scheme in Algorithm \ref{alg::1step_samp}. First, the user draws samples $\{P^{(1)}, \ldots, P^{(B)}\}$ of $P$ from their initial posterior $\Pi(\cdot \mid Z^{(n)})$ (e.g. via Markov chain Monte Carlo). Each posterior sample $P^{(b)}$ is then paired with an independent vector of uniform Dirichlet weights of length $n$; this arises from the representation
\begin{equation*}
    \Tilde{P} = \sum_{i=1}^{n}W_{i}\delta_{Z_i}
\end{equation*}
of the Bayesian bootstrap \citep[e.g. Section 4.7 of][]{Ghosal17}, where $(W_{1}, \ldots, W_{n}) \sim \text{Dir}(n;1,\ldots, 1)$. Finally, the draw of the estimand $\chi(P^{(b)})$ is summed with the weighted average of the efficient influence function at $P^{(b)}$.

\begin{algorithm}
\caption{One-step posterior sampling scheme}\label{alg::1step_samp}
Input: number of posterior samples $B$\;
Sample $\{P^{(1)}, \ldots, P^{(B)}\}$ from the initial posterior $\Pi(\cdot \mid Z^{(n)})$\;
 \For{$b \gets 1$ to $B$}{
   Sample $(W_{1}^{(b)},\ldots, W_{n}^{(b)})$ from $\textrm{Dir}(n;1,\ldots,1)$ independently of $P^{(b)}$\;
   Compute \begin{equation*}
        \Tilde{\chi}^{(b)} = \chi(P^{(b)}) + \sum_{i=1}^{n}W_{i}^{(b)}\dot{\chi}_{P^{(b)}}(Z_{i});\;
    \end{equation*}
 }
 Return $\{\Tilde{\chi}^{(1)}, \ldots , \Tilde{\chi}^{(B)}\}$.
\end{algorithm}

Algorithm \ref{alg::1step_samp} is attractive for several reasons. First, the same set of posterior samples of $P$ can be retained to use for different functionals. Second, there is no need to modify any existing implementation for sampling from $\Pi(\cdot \mid Z^{(n)})$; the user can just take the output from the initial implementation and add the random correction terms. This extra step is extremely efficient as it only requires drawing independent sets of Dirichlet weights and averaging the efficient influence functions. Moreover, steps 3-6 in Algorithm \ref{alg::1step_samp} can be easily parallelized to save further time. 

To provide some geometric intuition for our method, we adapt the ideas from \citet{Fisher21}. Consider the following mixture distributions constructed from an arbitrary $P \in \mathcal{P}$ (representing a draw from our initial posterior) and the true distribution $P_{0}$:
\begin{equation*}
    P^{\varepsilon} = (1-\varepsilon)P + \varepsilon P_{0}
\end{equation*}
for $0\leq \varepsilon \leq 1$. Within our model space $\mathcal{P}$, we could interpret $\{P^{\varepsilon}\}_{\varepsilon \in [0,1]}$ as a line segment linking $P$ and $P_{0}$ (these line segments are not guaranteed to lie within the model in general, but we will overlook this for the purpose of exposition). This induces a path $\{\chi(P^{\varepsilon})\}_{\varepsilon \in [0,1]}$ in the parameter space that reaches the true parameter value $\chi(P_{0})$ at $\varepsilon = 1$. Starting at $\varepsilon = 0$, we might hope to get closer to the truth by constructing a good approximation to this path.

This idea is illustrated visually in Figure \ref{fig::1step_proj}. The blue solid curve is equal to $\chi(P^{\varepsilon}) - \chi(P_{0})$, and its tangent at $\varepsilon = 0$---given by the blue dotted line---has slope $P_{0}[\dot{\chi}_{P}]$. The one-step estimator---taking $P$ as the ``initial estimate'' of $P_{0}$---estimates the slope with $\mathbb{P}_{n}[\dot{\chi}_{P}]$, and the resulting projection onto $\varepsilon = 1$ is shown by the green dashed line. With our posterior correction, the Bayesian bootstrap term $\Tilde{P}[\dot{\chi}_{P}]$ gives us a distribution on the slope rather than just a point estimate. As a result, we also get a distribution on the intercept at $\varepsilon = 1$ after projecting along the tangent from $\varepsilon = 0$, which is shown by the red histogram on the right. This intercept is the one-step corrected parameter; the histogram represents the uncertainty in $\Tilde{\chi}$ conditional on the ``posterior draw'' $P$. For the one-step posterior, this uncertainty is combined with the posterior uncertainty for $P$, which we recall to be conditionally independent of $\Tilde{P}$ given the data.

The reader may wonder why the empirical measure is insufficient for our posterior correction. We can illustrate the reason by considering linear functionals, which form the simplest class of examples. Let $\chi(P) = P[g]$, where $g:\mathcal{Z}\rightarrow \mathbb{R}$ is a known bounded function. The efficient influence function is $\dot{\chi}_{P} = g-P[g]$ \citep[e.g][]{Castillo15}, such that
\begin{equation*}
    \Tilde{\chi} = P[g] + \Tilde{P}[g - P[g]] = \Tilde{P}[g].
\end{equation*}
In words, the one-step corrected parameter is simply the Bayesian bootstrap expectation of $g$, which we already know to satisfy the semiparametric Bernstein-von Mises theorem. If $\Tilde{P}$ is replaced by $\mathbb{P}_{n}$ in the correction, then the parameter reduces to the one-step estimator $\mathbb{P}_{n}[g]$, which has zero posterior uncertainty. Thus, the Bayesian bootstrap correction term serves not only to debias the initial posterior but also to obtain the (asymptotically) correct posterior shape. This statement is established in generality in Section \ref{sec::theory}.

To summarize, the one-step posterior combines the uncertainty from two unknown quantities. First, we have the uncertainty in the functional $\chi(P)$, which arises from our initial posterior. Given a ``starting point'' $P$, we also have uncertainty in $P_{0}[\dot{\chi}_{P}]$, which is the slope of the tangent to the curve connecting $\chi(P)$ to the truth $\chi(P_{0})$. This latter source of uncertainty is modelled by replacing the unknown $P_{0}$ with the Bayesian bootstrap parameter $\Tilde{P}$.

\begin{figure}[]
\begin{center}
\includegraphics[width=6in]{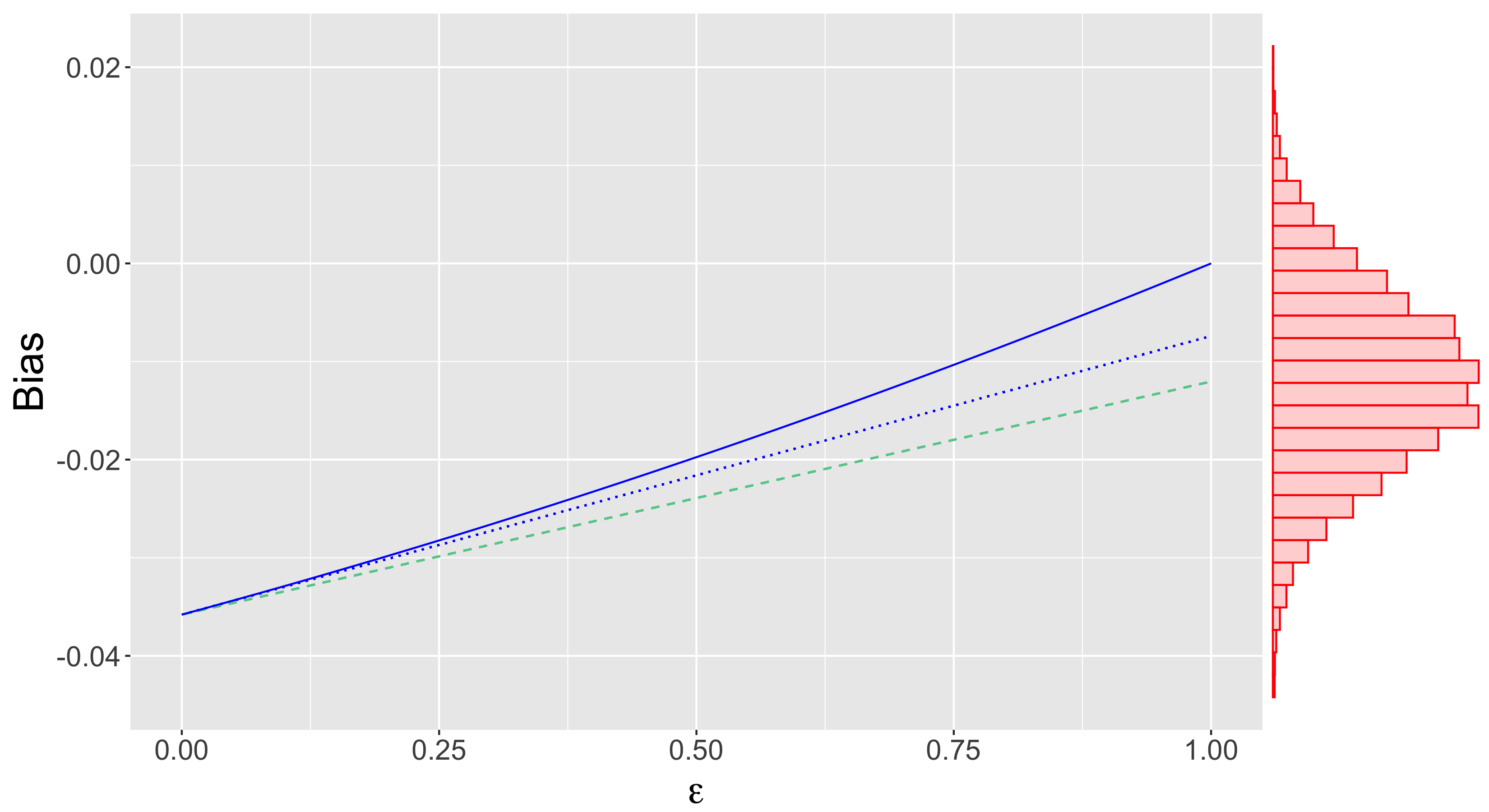}
\end{center}
\caption{One-step corrections as linear projections. The blue solid curve (\full) is $\chi(P^{\varepsilon}) - \chi(P_{0})$, whose tangent at $\varepsilon = 0$ is given by the blue dotted line (\dotted). The green dashed line (\dashed) estimates the tangent with the empirical distribution, and the histogram (\sqr{red}) on the right corresponds to the distribution of the tangent drawn from the Bayesian bootstrap, representing the uncertainty in the estimand conditional on the starting point $P$. \label{fig::1step_proj}}
\end{figure}

\subsection{Theory} \label{sec::theory}

In this subsection, we provide conditions for the one-step posterior to satisfy the semiparametric Bernstein-von Mises theorem; that is,
\begin{equation*}
    d_{BL}\left(\mathcal{L}_{\Pi \times \Pi_{BB}}(\sqrt{n}(\Tilde{\chi} - \hat{\chi}_{n}) \mid Z^{(n)}), \mathcal{N}(0, \Vert \dot{\chi}_{P_{0}} \Vert^{2}_{P_{0}})\right)\xrightarrow[]{P_{0}} 0,
\end{equation*}
where $\mathcal{L}_{\Pi \times \Pi_{BB}}(\sqrt{n}(\Tilde{\chi} - \hat{\chi}_{n}) \mid Z^{(n)})$ denotes the posterior law of $\sqrt{n}(\Tilde{\chi} - \hat{\chi}_{n})$ and $\hat{\chi}_{n}$ is any sequence of estimators satisfying (\ref{eqn::asympeff}). As discussed in Section \ref{sec::setup}, this justifies the use of central credible sets as confidence regions.

\begin{assumption} \label{ass::1step_bvm}
There exists a sequence of measurable subsets $(H_{n})_{n}$ of $\mathcal{P}$ satisfying $\Pi(P \in H_{n} \mid Z^{(n)})  \xrightarrow[]{P_{0}} 1$ with
\begin{itemize}
    \item[(a)]{(No second-order bias)} \begin{equation*}
    \sup_{P \in H_{n}} \left|\sqrt{n}r_{2}(P_{0},P)\right|=\sup_{P \in H_{n}} \left|\sqrt{n}\left\{\chi(P_{0})-\chi(P) - P_{0}[\dot{\chi}_{P}]\right\}\right| \rightarrow 0. 
    \end{equation*}
    \item[(b)]{($L_{2}$-convergence)} $\sup_{P \in H_{n}} \Vert \dot{\chi}_{P} - \dot{\chi}_{P_{0}} \Vert_{P_{0}}  \rightarrow 0$
    \item[(c)]{(Donsker class)} The sets $\{\dot{\chi}_{P}:\,P \in H_{n}\}$ are eventually contained in a fixed $P_{0}$-Donsker class \citep{vanderVaart96}.
\end{itemize}

\end{assumption}

We also provide an alternative to Assumption \ref{ass::1step_bvm}(c) below: \ref{ass::1step_bvm}(c$^{*}$)(i) weakens the Donsker condition at the cost of adding \ref{ass::1step_bvm}(c$^{*}$)(ii), which requires envelope function moment bounds. This alternative is also more appropriate if the sieve $\{H_{n}\}$ is not nested (i.e. we do not have $\ldots \supset H_{k-1} \supset H_{k} \supset H_{k+1} \supset \ldots$), which would likely preclude the possibility of finding a fixed Donsker class that satisfies Assumption \ref{ass::1step_bvm}(c).

\addtocounter{assumption}{-1}
\begin{assumption} \label{ass:onestep:post}
\,
\begin{itemize}
    \item[(c$^{*}$)]
    \begin{itemize}
        \item[(i)]{(Convergence of $\dot{\chi}_{P}$ under the empirical process)}
        \begin{equation*}
    \sup_{P \in H_{n}} |\mathbb{G}_{n}[\dot{\chi}_{P} - \dot{\chi}_{P_{0}}]| \xrightarrow[]{P_{0}} 0.
    \end{equation*}
    This statement is interpreted in terms of outer probability \citep{vanderVaart96} if the supremum is not measurable.
    \item[(ii)]{(Bounding of envelope functions)} The sets $\{\dot{\chi}_{P}:\,P \in H_{n}\}$ have envelope functions $G_{n}$ (i.e. $|\dot{\chi}_{P}(z)| \leq G_{n}(z)$ for all $P \in H_{n}$ and all $z \in \mathcal{Z}$) satisfying
    \begin{equation*}
        \lim_{C \rightarrow \infty} \limsup_{n \rightarrow \infty}P_{0}G_{n}^{2}1_{G_{n}^{2} > C} = 0, \qquad P_{0}G_{n}^{4} = o(n).
    \end{equation*}
    \end{itemize}
\end{itemize}
\end{assumption}

\begin{theorem} \label{theo::1step_bvm}
Under Assumption \ref{ass::1step_bvm}, the one-step posterior satisfies the semiparametric BvM theorem. The result also holds if Assumption \ref{ass::1step_bvm}(c) is replaced by Assumption \ref{ass::1step_bvm}(c$^{*}$).
\end{theorem}

The structure of Assumption \ref{ass::1step_bvm} mirrors the standard assumptions for one-step estimators to attain asymptotic efficiency (see Section \ref{sec::grad_based}). To satisfy Assumptions \ref{ass::1step_bvm}(a) and \ref{ass::1step_bvm}(b), the sets $H_{n}$ are typically included in shrinking neighbourhoods of $P_{0}$. In light of the earlier discussion about the quadratic nature of $r_{2}$, this term can be neglected when these neighbourhoods shrink quickly enough. The empirical process conditions in Assumptions \ref{ass::1step_bvm}(c) and (c*) restrict the complexity of the posterior so that we can obtain uniform convergence results.

To provide some context for Assumption \ref{ass::1step_bvm}, we compare with the conditions for Theorem 2.1 in \citet{Castillo15}, which is a semiparametric BvM result for functionals. As one might expect, the second-order remainder must be negligible in both cases (Assumption \ref{ass::1step_bvm}(a)) so that the functional is well-approximated by its linear expansion. Assumptions \ref{ass::1step_bvm}(b), (c) and (c*) are not explicitly required by \citet{Castillo15}, but similar conditions are generally used to establish the asymptotically normal expansion of the likelihood (see (4.10) of \citet{Castillo15}, Theorem 1 of \citet{Ray20} and p. 373-374 of \citet{Ghosal17}). Analogous conditions can also be found in the classical setting for semiparametric likelihood methods \citep{Murphy00, Cheng08a}.

Crucially, our result avoids the restrictive ``change-of-measure'' condition in (2.14) of \citet{Castillo15} (see also (3.6) in \citet{Ray20}). This depends delicately on the interplay between the model and the prior, which can fail if the prior is oversmooth or adaptive. From a technical viewpoint, the condition is also challenging to verify for complex priors.

As we will see with the examples, the efficient influence function often admits a natural representation of the form
\begin{equation*}
    \dot{\chi}_{P} = \psi_{P} - P[\psi_{P}],
\end{equation*}
where $\psi_{P}: \mathcal{Z} \rightarrow \mathbb{R}$ is a measurable function. Theorem \ref{theo::1step_bvm} still holds if $\dot{\chi}_{P}$ and $\dot{\chi}_{P_{0}}$ are replaced by $\psi_{P}$ and $\psi_{P_{0}}$ respectively in Assumptions \ref{ass::1step_bvm}(b), (c) and (c*). 

For our theoretical analysis, the conditional independence of $P$ and $\Tilde{P}$ given $Z^{(n)}$ is crucial. This is because we wish to utilize the desirable asymptotic properties of the Bayesian bootstrap, and these properties must continue to hold even if we condition on $P$. Our method can be extended to work with proper Dirichlet process posteriors but this would necessitate additional assumptions on the prior base measure, so we will not pursue this further. 

\section{The integrated squared density} \label{sec::dens_squared}

Suppose that we observe independent and identically distributed data $Z^{(n)} =(Z_{1},\ldots,Z_{n})$ from a distribution $P_{0}$ on $\mathbb{R}$ that is absolutely continuous with respect to Lebesgue measure with a square-integrable density. We make slight changes in the notation by identifying distributions with their Lebesgue densities, e.g. $P_{0}$ is now identified by the true density $f_{0}$. Consider the problem of estimating $\chi(f)=\int f(z)^{2}\,dz$. There is a large literature on the estimation of $\chi(f)$ \citep[e.g.][]{Bickel88, Birge95, Laurent96, Newey98, Gine08} and more generally, on power integrals $\int f(z)^{q}dz$ for integers $q \geq 2$. The integrated squared density plays an important role in nonparametric rank statistics \citep[p.266-267 of][]{PrakasaRao83}.

In \cite{Castillo15}, the authors obtained a BvM theorem for this functional with random histogram priors, which are much simpler than common priors used in density estimation such as infinite mixture models. Under regularity assumptions on the true density $f_0$, we prove in this section that the one-step posterior distribution satisfies the semiparametric BvM theorem for Dirichlet process Gaussian location mixtures. The prior on $f$ is defined by 
\begin{equation} \label{eqn::dpmm_model}
f_{G,\sigma}(z) = \int_{\mathbb R} \varphi_{\sigma}(z-\mu) dG(\mu), \quad G \sim DP(M G_0), \quad M>0,
\end{equation}
where $\varphi_{\sigma}(z-\mu)$ is the Gaussian density centred at $\mu$ with standard deviation $\sigma$, and $G_0$ is a distribution on $\mathbb R$ with density $g_0$ that is positive and continuous and satisfies 
$g_0(x) \lesssim e^{-c_1 |x|^\tau}$ for some $\tau>0$. For instance, we can choose $G_0$ to be a Gaussian distribution. As in \citet{Kruijer10,Shen13,Scricciolo14}, we consider an inverse Gamma prior on $\sigma $: $\sigma \sim IG(a,b)$.

The efficient influence function for $\chi$ is $\dot{\chi}_{f}(Z) = 2\left\{f(Z) - \chi(f)\right\}$ \citep{Bickel88}, so the one-step corrected parameter is
\begin{equation*}
    \Tilde{\chi} = 2\Tilde{P}[f(Z)] - \chi(f).
\end{equation*}
It is straightforward to derive the second-order bias:
\begin{equation*}
    r_{2}(f_{0},f) = \Vert f - f_{0} \Vert_{2}^{2},
\end{equation*}
where $\Vert\cdot\Vert_{2}$ denotes the $L_{2}$-norm with respect to Lebesgue measure. We highlight the quadratic nature of this term; accordingly, we require the $L_{2}$ contraction rate for $f$ to be faster than $n^{-1/4}$.

We consider the following smoothness assumption on $f_{0}$ introduced by \citet{Shen13}.
\begin{definition}
     A density $f$ belongs to the class of locally $\alpha$-H\"older functions $C(\alpha, L, \tau_0)$ with polynomial function $L$ and $\tau_{0} \geq 0$ if $p$ has $k_\alpha= \lceil \alpha \rceil -1$ derivatives with 
     $$ |(D^{k_\alpha}f)(z+h)  - (D^{k_\alpha}f)(z)| \leq L(z) e^{\tau_{0}h^{2}}|h|^{\alpha-k_\alpha}, \quad \forall z,h \in \mathbb R.$$
\end{definition}

\begin{theorem} \label{th:f2}
Suppose that the true density $f_{0}$ belongs to a class $C(\alpha, L, \tau_0)$ for $\alpha > 1$. Assume further that there exists $\epsilon>0$ such that for all $k\leq k_\alpha$
     $$ P_{0}\left[\left(\frac{|D^k f_{0}|}{f_{0}}\right)^{(2\alpha+ \epsilon)/k}\right]< \infty , \quad P_{0}\left[ \left(\frac{L}{f_{0}}\right)^{2+ \epsilon/\alpha}\right]< \infty  $$
and $f_0$ is monotone non-increasing (resp. non-decreasing) for $z$ large enough (resp. $-z$ large enough) and there exist $c_1, c_2>0$ such that 
     $$ f_0(z) \leq c_1 e^{-c_2 |z|^{\tau_0}}.$$
If we specify a Dirichlet process Gaussian location mixture prior for $f$ as defined above, then the one-step posterior distribution of the integrated squared density satisfies the semiparametric BvM theorem. 
\end{theorem}

From \citet{Shen13}, under the above conditions the posterior distribution concentrates in Hellinger or $L_1$ distance at the rate $n^{-\alpha/(2\alpha+1)} (\log n)^r$ for some $r>0$, which is minimax adaptive up to a $\log n$ term for all $\alpha>0$  \citep{maugis:michel}. The condition $\alpha>1$ is used to verify part (c*) of Assumption \ref{ass::emp_proc_alt}. It can be weakened to $\alpha >1/2$ by considering the truncated prior: $f \sim \Pi_T(df) \propto \Pi(df)\Ind_{\|f\|_\infty< K}$ for some arbitrarily large but fixed constant $K$. 

Note that there are different definitions of the class $C(\alpha, L, \tau_0)$ in the literature, and all these variants can be used in Theorem \ref{th:f2}. As in \cite{donnet18}, this condition can be replaced by assuming the existence of a finite mixture of Gaussian densities $f_{G^*,\sigma}$ with variance $\sigma^2=o(1)$ with at most $N=O(\sigma^{-1}\log 1/\sigma)$ support points such that the Kullback-Leibler divergence between $f_0$ and $f_{G^*,\sigma}$ is of order $O(\sigma^{2\alpha})$.


\section{The mean of an outcome that is missing-at-random} \label{sec::miss}

\subsection{General methodology}

In this section, we revisit the missing data example in Example \ref{exa::miss}. Recall that the data takes the form $Z = (X,A,AY)$, where $X$ is a vector of covariates, $A$ is a binary missingness indicator, $Y$ is the outcome variable of interest that is unobserved when $A=0$, and the target estimand is 
\begin{equation*}
    \chi(P) =  E_{P}[E_{P}[Y \mid A=1, X]].
\end{equation*}

The density for a single observation $z = (x,a,ay)$ can be written as
\begin{equation*}
    p(z) = \pi(x)^{a}(1-\pi(x))^{1-a}p_{Y \mid X}(y \mid x)^{a}q(x),
\end{equation*}
which factorizes over the parameters. Consequently, if the three parameters are given independent priors, they will remain independent given $Z^{(n)}$ in the posterior. An important implication is that the model for $\pi$ will not feature in the marginal posterior of $\chi$ since the functional is only defined by the other two parameters. This well-known phenomenon is discussed in \citet{RobinsRitov97,Robins15} and \citet{Ray20}.

As stated earlier, the efficient influence function of $\chi$ is
\begin{equation*}
    \dot{\chi}_{P}(Z) = \frac{A}{\pi(X)}\{Y - m(X)\} + m(X) - \chi(P),
\end{equation*}
so that the one-step corrected parameter
\begin{equation*}
    \Tilde{\chi}(P, \Tilde{P}) = \Tilde{P}\left[\frac{A}{\pi(X)}\{Y - m(X)\} + m(X)\right]
\end{equation*}
depends on $P$ only through $\pi$ and $m$. 

In the following, we state assumptions for the one-step posterior to satisfy the semiparametric BvM theorem, specializing the conditions in Assumption \ref{ass::1step_bvm}. Let $\pi_{0}$ and $m_{0}$ respectively denote the true values of $\pi$ and $m$ associated with $P_{0}$. 

\begin{assumption} \label{ass::miss_bvm}
There exists a sequence of measurable subsets $(H_{n})_{n}$ of $\mathcal{P}$ satisfying $\Pi(P \in H_{n} \mid Z^{(n)})  \xrightarrow[]{P_{0}} 1$ such that
\begin{itemize}
    \item[(a)]{($L_{2}$-convergence of $\pi$ and $m$)} There exist numbers $\rho_{n},\varepsilon_{n} \rightarrow 0$ such that
    \begin{align*}
        \sup_{P \in H_{n}} \left\Vert \frac{1}{\pi} - \frac{1}{\pi_{0}} \right\Vert_{P_{0}}  &\leq \rho_{n},  \quad \text{and} \quad 
        \sup_{P \in H_{n}} \left\Vert m - m_{0} \right\Vert_{P_{0}}  \leq \varepsilon_{n}
    \end{align*}
    and $\sqrt{n}\rho_{n}\varepsilon_{n} \rightarrow 0$.
    \item[(b)]{(Uniform bounding)} For all sufficiently large $n$, there exist fixed constants $\delta, C > 0$ such that for all $P \in H_{n}$,
    \begin{align*}
        \left\Vert\frac{1}{\pi}\right\Vert_{\infty}+ \left\Vert\frac{1}{\pi_{0}} \right\Vert_{\infty}&< \frac{1}{\delta} , \quad \text{and} \quad 
        P_{0}(|Y - m(X)|<C) = 1.
    \end{align*}
    \item[(c)]{(Donsker class)} The sequences of sets $\{\pi:\,P \in H_{n}\}$ and $\{m:\,P \in H_{n}\}$ are both eventually contained in fixed $P_{0}$-Donsker classes.
    
\end{itemize}
\end{assumption}

\begin{theorem} \label{theo::1step_miss}
Under Assumption \ref{ass::miss_bvm}, the one-step posterior for the outcome mean satisfies the semiparametric BvM theorem.
\end{theorem}

Assumption \ref{ass::miss_bvm}(a) states that $1/\pi$ and $m$ must both converge uniformly on $H_{n}$ to their respective truths in $L_{2}$, and their combined rate of convergence must be faster than  $1/\sqrt{n}$. To connect this to Assumption \ref{ass::1step_bvm}(a), we note that the second-order bias takes the form
\begin{equation} \label{eqn::miss_2nd}
    r_{2}(P_{0}, P) 
    = P_{0}\left[\pi_{0}(X)\left(\frac{1}{\pi_{0}(X)} - \frac{1}{\pi(X)}\right)\{m_{0}(X) - m(X)\}\right] \\
    \leq \left\Vert \frac{1}{\pi_{0}} - \frac{1}{\pi} \right\Vert_{P_0} \Vert m_{0} - m\Vert_{P_0},
\end{equation}
where the upper bound is due to Cauchy-Schwarz. Thus, Assumption \ref{ass::miss_bvm}(a) immediately implies Assumption \ref{ass::1step_bvm}(a).

The cross-term structure of the second-order bias enables the rate of convergence of one parameter to compensate for the other; this is a property known as ``rate double robustness'' \citep{Rotnitzky21}. For example, if $\pi$ were to be modelled parametrically, then the posterior for $m$ would be permitted to contract around $m_{0}$ arbitrarily slowly. Alternatively, $m$ and $\pi$ could both be modelled using flexible nonparametric methods, and the condition would be satisfied provided that both their posteriors contract faster than $n^{-1/4}$. As mentioned earlier, the ordinary marginal posterior of $\chi$ under independent priors will not depend on the model for $\pi$, so it cannot possess this double robust property \citep{Robins00, Robins15}.

The uniform bounds in Assumption \ref{ass::miss_bvm}(b) help to ensure that the $L_{2}$-convergence of $\dot{\chi}_{P}$ (i.e. Assumption \ref{ass::1step_bvm}(b)) follows from Assumption \ref{ass::miss_bvm}(a). The assumption that the propensity score is bounded away from zero is commonly referred to as ``positivity'' and is routine for any estimation approach that involves inverse probability weighting. As discussed in Section 4.2 of \citet{Kennedy22} in a classical setting, the boundedness of $|Y - m(X)|$ can be relaxed to bounded moment conditions by using H\"older's inequality. 

\subsection{Comparison with prior augmentations in the least favourable direction}

This example was studied by \citet{Ray20} under the constraint that the outcome variables are binary. With similar motivations to our correction method, \cite{Ray20} proposed constructing a preliminary estimator $\hat{\pi}_{n}$ of the propensity score based on a separate dataset that is independent of $Z^{(n)}$ and modelling the outcome regression function by
\begin{equation} \label{eqn::ray_aug}
    m(x) = \Psi\left\{w(x) + \frac{\lambda}{\hat{\pi}_{n}(x)} \right\},
\end{equation}
where $\Psi(t) = 1/(1+e^{-t})$ is the expit function, and $(w,\lambda)$ are parameters.

The prior on $w(x)$ should be specified such that the resulting prior on $\Psi\{w(x)\}$ captures the user's a priori beliefs on $m(x)$. In contrast, the augmentation term $\lambda/\hat{\pi}_{n}(x)$ plays a technical role, with the goal of satisfying the change-of-measure condition for the semiparametric BvM theorem as discussed in Section \ref{sec::1step_cor}. The fluctuation parameter $\lambda$ follows a $\mathcal{N}(0,\sigma_{n}^{2})$ prior independently of $w(x)$, where the variance $\sigma_{n}^{2}$ is allowed to vary with $n$. For the marginal covariate distribution, \citet{Ray20} specified a Dirichlet process prior independent of $(w(x),\lambda)$. We will assume throughout that $\pi_{0}$ and $\hat{\pi}_{n}$ satisfy the following assumption.

\begin{assumption} \label{ass::propensity}
    The true propensity score satisfies $\pi_{0} > \delta$ for some $\delta > 0$, and the preliminary estimators $\hat{\pi}_{n}$ satisfy $\Vert 1/\hat{\pi}_{n}\Vert_{\infty} = O_{P_0}(1)$ and
\begin{equation*}
    \left\Vert\frac{1}{\hat{\pi}_{n}} - \frac{1}{\pi_{0}}\right\Vert_{P_0} = O_{P_0}(\rho_{n}),
\end{equation*}
where $\rho_{n} \rightarrow 0$. 
\end{assumption}

\begin{theorem}[Theorem 2 of \citet{Ray20}] \label{theo::ray_vdV}
Assume that there exist measurable sets $\mathcal{H}_{n}$ of functions satisfying, for every $t \in \mathbb{R}$ and some numbers $u_{n}, \varepsilon_{n} \rightarrow 0$,
\begin{align}
    \Pi(\lambda: |\lambda| \leq u_{n}\sigma_{n}^{2}\sqrt{n} \mid Z^{(n)})&\xrightarrow[]{P_{0}} 1, \label{ass::ray_3.9}\\
    \Pi((w,\lambda): w + (\lambda + tn^{-1/2})\hat{\pi}_{n}^{-1} \in \mathcal{H}_{n} \mid  Z^{(n)})&\xrightarrow[]{P_{0}} 1,\label{ass::ray_3.10}\\
    \sup_{m = \Psi(\eta): \eta \in \mathcal{H}_{n}}\Vert m-m_{0}\Vert_{P_0} &\leq \varepsilon_{n},\label{ass::ray_3.11}\\
    \sup_{m = \Psi(\eta): \eta \in \mathcal{H}_{n}} |\mathbb{G}_{n}[m-m_{0}]|&\xrightarrow[]{P_{0}} 0.  \label{ass::ray_3.12} 
\end{align}
Under Assumption \ref{ass::propensity}, if $n\sigma_{n}^{2} \rightarrow \infty$ and $\sqrt{n}\rho_{n}\varepsilon_{n} \rightarrow 0$, then the marginal posterior of $\chi$ satisfies the semiparametric Bernstein-von Mises theorem.
\end{theorem}

The critical aspect of the above result is that it avoids the change-of-measure condition in equation (3.6) of Theorem 1 of \citet{Ray20} for establishing the semiparametric BvM theorem without the prior augmentation. We now compare this approach with the one-step posterior in this set-up. As discussed by \citet{Ray20}, the preliminary estimator $\hat{\pi}_{n}$ could be interpreted as a degenerate prior on $\pi$. This interpretation immediately suggests a corresponding one-step posterior by defining
\begin{equation} \label{eqn::1step_ray}
    \Tilde{\chi} = \Tilde{P}\left[\frac{A}{\hat{\pi}_{n}(X)}\{Y - m(X)\} + m(X)\right].
\end{equation}
To facilitate comparisons with the method of \citet{Ray20}, we will continue to describe the assumptions in terms of the parameter $w(x)$, where we have $m = \Psi(w)$ in this case.

\begin{theorem} \label{theo::1step_ray}
Assume that there exist measurable sets $\Tilde{\mathcal{H}}_{n}$ of functions satisfying, for some numbers $\varepsilon_{n} \rightarrow 0$,
\begin{align}
    \label{ass::1step_3.10} \Pi(w \in \Tilde{\mathcal{H}}_{n} \mid  Z^{(n)})&\xrightarrow[]{P_{0}} 1,\\
    \label{ass::1step_3.11} \sup_{m = \Psi(w): w \in \Tilde{\mathcal{H}}_{n}}\Vert m-m_{0}\Vert_{P_0} &\leq \varepsilon_{n},\\
    \label{ass::1step_3.12} \sup_{m = \Psi(w): w \in \Tilde{\mathcal{H}}_{n}} |\mathbb{G}_{n}[m-m_{0}]|&\xrightarrow[]{P_{0}} 0.
\end{align}
Under Assumption \ref{ass::propensity}, if $\sqrt{n}\rho_{n}\varepsilon_{n} \rightarrow 0$, then the one-step posterior satisfies the semiparametric Bernstein-von Mises theorem.
\end{theorem}

A direct comparison of the assumptions for Theorems \ref{theo::ray_vdV} and \ref{theo::1step_ray} is complicated slightly by the different specifications for the outcome regression function, which means that the sets $\mathcal{H}_{n}$ and $\Tilde{\mathcal{H}}_{n}$ may not match. Nevertheless, we should keep in mind that the augmentation in (\ref{eqn::ray_aug}) is motivated purely by improving the inference for $\chi$ rather than the recovery of $m$. It is likely, therefore, that a prior on $w$ that satisfies the assumptions (\ref{ass::ray_3.9}) - (\ref{ass::ray_3.12}) (in combination with some sequence of priors for $\lambda$) will also satisfy the assumptions (\ref{ass::1step_3.10}) - (\ref{ass::1step_3.12}) for the unaugmented regression model. (In any case, it will be simpler to check (\ref{ass::1step_3.10}) - (\ref{ass::1step_3.12}) due to the absence of the augmentation term.) We verify this claim below for a Gaussian process prior studied by \citet{Ray20}. 

Following the set-up in \citet{Ray20}, we take the covariate space to be $[0,1]$. The \textit{Riemann-Liouville} process released at zero of regularity $\bar{\beta} > 0$ is defined by
\begin{equation} \label{eqn::riemann}
    w(x) = \sum_{k=0}^{\floor{\bar{\beta}} + 1}g_{k}x^{k} + \int_{0}^{x}(x-s)^{\bar{\beta}-1/2}dB_{s}, \quad x \in [0,1],
\end{equation}
where $\floor{\bar{\beta}} + 1$ is the smallest integer that is strictly larger than $\bar{\beta}$, the $\{g_{k}\}$ are i.i.d. $\mathcal{N}(0,1)$ variables, and $B$ is an independent Brownian motion (see Chapter 11 of \citet{Ghosal17}). We use this to define our prior for $m = \Psi(w)$. For covariate spaces equal to $[0,1]^{d}$ for integer $d \geq 1$, the subsequent corollary can be easily extended to the Gaussian series priors considered by \citet{Ray20}.

\begin{corollary} \label{cor::riemann}
    Suppose that $\pi_{0}^{-1} \in \mathfrak{C}^{\alpha}([0,1])$, $m_{0} \in \mathfrak{C}^{\beta}([0,1])$ for $\alpha,\beta > 0$, where $\mathfrak{C}^{s}([0,1])$ is the H\"older space of order $s$ on $[0,1]$ (see Chapter 4 of \citet{Gine16}). Suppose further that we specify our one-step posterior using (\ref{eqn::1step_ray}) and the Riemann-Liouville prior for $m$ defined by (\ref{eqn::riemann}). Under Assumption \ref{ass::propensity}, if $\beta \wedge \bar{\beta} > 1/2$ and $\sqrt{n}\rho_{n}\varepsilon_{n} \rightarrow 0$ for 
    \begin{equation*}
        \varepsilon_{n} = n^{-\frac{\beta \wedge \bar{\beta}}{2\bar{\beta} + 1}}\log n,
    \end{equation*}
    then the one-step posterior satisfies the semiparametric Bernstein-von Mises theorem.
\end{corollary}

The corresponding result for the prior augmentation method in Corollary 2 of \citet{Ray20} requires the following additional condition:
\begin{equation*}
    \sigma_{n} = O(1), \quad \text{and} \quad \varepsilon_{n} \ll \sigma_{n}.
\end{equation*}
For finite-sample inference, the need to tune $\sigma_{n}$ is a substantial practical drawback of the prior augmentation approach relative to the one-step posterior. This is because the prior for $\lambda$ will likely have a significant impact on the posterior for the target estimand, yet the augmentation term only plays a technical role, so there will be no subjective or substantive guidance to select $\sigma_{n}$.

From a computational perspective, the prior augmentation is also more difficult to implement than the one-step posterior correction---it is generally the case that existing software for posterior sampling will have to be rewritten in order to accommodate the new linear term. Furthermore, as \citet{Ray20} wrote in their discussion, the augmented prior ``will not perform better for any other functional''. Therefore, if multiple functionals are of interest, the user would have the time-consuming task of implementing distinct models for each functional. As we have mentioned previously, our approach handles this situation very efficiently because we can retain a single set of posterior samples from the initial model, with which we can compute all requisite corrected posteriors. 

We recall that the theory provided by \citet{Ray20} for this prior augmentation relies on an independent, auxiliary dataset to construct the propensity score estimate $\hat{\pi}_{n}$. A similar prior augmentation was studied by \citet{RaySzabo19} for continuous outcomes, and they evaluated the performance of their method empirically when $\hat{\pi}_{n}$ is trained on the same data as the posterior. Using a point estimate in this fashion is similar to classical correction approaches like \textit{targeted maximum likelihood estimation} \citep{vanderLaan06, vanderLaan11}. As we will discuss further in Section \ref{sec::acic}, we have found that it is beneficial to instead incorporate the Bayesian uncertainty in the nuisance parameters (in this case the propensity score) by propagating the variation from the initial posterior to the one-step posterior.

\section{The average treatment effect on the treated} \label{sec::att}

Consider the problem of inferring the causal effect of a binary treatment $A$ on an outcome $Y$ using observational data. Our data takes the form $Z = (X,A,Y)$, where $X$ is a vector of pre-treatment covariates. Causal frameworks often posit the existence of counterfactual variables and define causal effects as contrasts between functionals of counterfactual distributions \citep[e.g.][]{Pearl09, Hernan20}. In this case, we can define the pair of variables $(Y^{1}, Y^{0})$ to be the potential outcomes associated with treatment values $1$ (``treated'') and $0$ (``controls'') respectively. Under the \textit{consistency} assumption \citep{Hernan20}
\begin{equation*}
    Y = AY^{1} + (1-A)Y^{0},
\end{equation*}
we observe exactly one of $(Y^{1},Y^{0})$ for each unit in the dataset.

Our target estimand is the \textit{average treatment effect on the treated (ATT)} \citep{Rubin77, Heckman85}
\begin{equation*}
    \chi(P) = E_{P}[Y^{1}-Y^{0} \mid A = 1],
\end{equation*}
which quantifies the average causal effect of withholding treatment on the population currently being treated. This can also be formulated without counterfactuals \citep[e.g.][]{Geneletti11, Dawid21, Yiu22}.

An attractive aspect of the ATT is that it is identified under assumptions that are strictly weaker than those commonly specified for the more well-known average treatment effect $E_{P}[Y^{1}-Y^{0}]$. More specifically, the identification formula
\begin{equation} \label{eqn::att_iden}
    \chi(P) = E_{P}[Y-E_{P}[Y \mid A=0, X] \mid A=1]
\end{equation}
can be obtained under the assumptions of \textit{unconfoundedness for controls}
\begin{equation*}
    Y^{0} \independent A \mid X
\end{equation*}
and \textit{weak overlap}:
\begin{equation*}
    P(A=1 \mid X) < 1
\end{equation*}
with $P$-probability 1 \citep{Imbens04}. Under violations of consistency or unconfoundedness for controls, our results still apply for estimation of the functional defined by (\ref{eqn::att_iden}) provided that weak overlap holds.

Similar to Section \ref{sec::miss}, it is natural to parameterize $P$ in terms of the propensity score $\pi(x) = P(A=1 \mid X=x)$, the conditional outcome density $p_{Y \mid A,X}(y \mid a,x)$, and the marginal covariate density $q$. For a single observation $z=(x,a,y)$, the density
\begin{equation*}
    p(z) = \pi(x)^{a}(1-\pi(x))^{1-a}p_{Y \mid A,X}(y \mid a,x)q(x)
\end{equation*}
factorizes into the three parameters, so it is particularly convenient and intuitive to specify independent priors.

The functional in (\ref{eqn::att_iden}) has efficient influence function
\begin{equation} \label{eqn::att_dp}
    \dot{\chi}_{P} = \frac{A-\pi(X)}{P[A][1-\pi(X)]}\{Y - \mu^{(0)}(X)\} - \frac{A}{P[A]}\chi(P),
\end{equation}
where $\mu^{(0)}(x) = E_{P}[Y \mid A=0,X=x]$ is the outcome regression function for the controls \citep[Chapter 8 of][]{vanderLaan11}. We highlight the presence of the marginal probability of treatment assignment $P[A]$, which we can write as $P[A] = Q[\pi(X)]$ with respect to the natural parameterization in terms of $\pi$ and the marginal covariate distribution $Q$. 

The behaviour of a posterior on $Q[\pi(X)]$ will depend on the interplay between the models for $\pi$ and $Q$. In the construction of our one-step posterior, we have found it to be beneficial theoretically to replace $Q[\pi(X)]$ with $\Tilde{P}[A]$, which is partially due to the fact that the posterior for $\Tilde{P}[A]$ is guaranteed to contract to $P_{0}[A]$ at the rate of $n^{-1/2}$ without requiring additional assumptions on $\pi$ and $Q$. With a slight abuse of notation and terminology, we define the one-step corrected parameter to be
\begin{equation} \label{eqn::att_1step}
    \Tilde{\chi}(P,\Tilde{P}) = \Tilde{P}\left[\frac{A-\pi(X)}{\Tilde{P}[A][1-\pi(X)]}\{Y - \mu^{(0)}(X)\}\right].
\end{equation}
We clarify that the $\Tilde{P}$ in the denominator of the expression is the same as the one that appears outside of the large square brackets. From a computational perspective, this means that the user does not need to draw an additional set of uniform Dirichlet weights for each posterior sample of $P$. Moreover, (\ref{eqn::att_1step}) only depends on $P$ through $\pi$ and $\mu^{(0)}$, so the user can undertake an initial posterior analysis that is conditional on the covariates and avoid computing the posterior for $Q$.

This example helps to illustrate that our framework offers the flexibility to attain more attractive theoretical properties by making modifications to the correction. Similar modifications are used for classical one-step estimation \citep[e.g.][]{Rotnitzky21} with the empirical measure and the Bayesian bootstrap once again playing analogous roles. The corrected parameter defined in (\ref{eqn::att_1step}) also admits an alternative motivation by using the efficient influence function in (\ref{eqn::att_dp}) to define a Bayesian-bootstrap-weighted estimating equation
\begin{equation*}
    \Tilde{P}\left[\frac{A-\pi(X)}{P[A][1-\pi(X)]}\{Y - \mu^{(0)}(X)\} - \frac{A}{P[A]}\chi\right] = 0,
\end{equation*}
which is solved by $\chi = \Tilde{\chi}(P,\Tilde{P})$. The relationship between one-step and estimating equation estimators is discussed in \citet{Hines22} from a classical perspective.

\begin{assumption} \label{ass::att_bvm}
Let $\pi_{0}$ and $\mu_{0}^{(0)}$ denote the true parameters associated with $P_{0}$. There exists a sequence of measurable subsets $(H_{n})_{n}$ of $\mathcal{P}$ satisfying $\Pi(P \in H_{n} \mid Z^{(n)})  \xrightarrow[]{P_{0}} 1$ such that
\begin{itemize}
    \item[(a)]{($L_{2}$-convergence of $\pi$ and $\mu^{(0)}$)} There exist numbers $\rho_{n},\varepsilon_{n} \rightarrow 0$ such that
    \begin{align*}
        \sup_{P \in H_{n}} \left\Vert \pi - \pi_{0} \right\Vert_{P_{0}}  &\leq \rho_{n} \\
        \sup_{P \in H_{n}} \left\Vert \mu^{(0)} - \mu_{0}^{(0)} \right\Vert_{P_{0}}  &\leq \varepsilon_{n}
    \end{align*}
    and $\sqrt{n}\rho_{n}\varepsilon_{n} \rightarrow 0$.
    \item[(b)]{(Uniform bounding)} For all sufficiently large $n$, there exist fixed constants $\delta, C > 0$ such that for all $P \in H_{n}$,
    \begin{align*}
        \Vert\pi\Vert_{\infty}, \Vert\pi_{0}\Vert_{\infty} &< 1-\delta \\
        P_{0}(|Y - \mu^{(0)}(X)| < C) &= 1.
    \end{align*}
    \item[(c)]{(Donsker class)} The sequences of sets $\{\pi:\,P \in H_{n}\}$ and $\{\mu^{(0)}:\,P \in H_{n}\}$ are both eventually contained in fixed $P_{0}$-Donsker classes.
\end{itemize}
\end{assumption}

\begin{theorem} \label{theo::1step_att}
Under Assumption \ref{ass::att_bvm}, the one-step posterior for the ATT functional (\ref{eqn::att_iden}) satisfies the semiparametric BvM theorem.
\end{theorem}

Assumption \ref{ass::att_bvm} has a similar form to Assumption \ref{ass::miss_bvm} for the example in Section \ref{sec::miss}. In particular, the second-order bias in this case possesses a cross-term structure in terms of $\pi$ and $\mu^{(0)}$, which induces the rate double robustness property in Assumption \ref{ass::att_bvm}(a). 

Given the conspicuous absence of $\mu^{(1)}(x) = E_{P}[Y \mid A = 1,X=x]$ from the definitions and assumptions, it is natural to ask whether one should model the parameters $p_{Y\mid A,X}(y \mid 1,x)$ and $p_{Y\mid A,X}(y \mid 0,x)$ with independent priors, such that the model for $p_{Y\mid A,X}(y \mid 1,x)$ drops out of the analysis for $\Tilde{\chi}$, and the posterior for $\mu^{(0)}$ is derived solely from the data on the controls. In the heterogeneous causal effects literature, this separation of the regression model into the two treatment arms is called the \textit{T-learner} approach (``two'' learners; \citet{Kunzel19}).

In general, we would recommend against doing this and instead advocate for specifying a single regression model that treats $A$ as just another covariate, which is called the \textit{S-learner} approach (``single'' learner; \citet{Kunzel19}). Provided that Assumption \ref{ass::att_bvm} is satisfied, there is no difference in terms of asymptotic performance, but we would expect to obtain finite-sample efficiency gains from the pooling of information across the entire dataset.

An advantage of using the T-learner approach, however, would be to prevent potential model misspecification for $p_{Y\mid A,X}(y \mid 1,x)$ from contaminating the posterior for $p_{Y\mid A,X}(y \mid 0,x)$, thereby reducing the risk of violating Assumption \ref{ass::att_bvm}(a). There could also be a finite sample impact if $\mu_{0}^{(0)}$ and $\mu_{0}^{(1)}$ are very different (e.g. $\mu_{0}^{(0)}$ is much smoother than $\mu_{0}^{(1)}$ due to high heterogeneity in the conditional treatment effect given $X$), because prior regularization may shrink a joint posterior for $p_{Y \mid A,X}(y \mid a,x)$ towards a compromise that is detrimental to precision for $\mu_{0}^{(0)}$. Ultimately, we would encourage the user to think carefully about their modelling specifications, both in terms of incorporating substantive knowledge and the implications concerning the theory we have provided.

\section{Simulation studies} \label{sec::sim}

\subsection{Estimating the integrated squared density under a Laplace distribution} \label{sec::laplace}

Suppose that the data $Z^{(n)}$ are generated i.i.d. from a $\textrm{Laplace}(0,1)$ distribution, which has density
\begin{equation*}
    f_{0}(z) = \frac{1}{2}\exp(-|z|).
\end{equation*}
The Laplace distribution was chosen because it is non-differentiable at $z=0$, and we would expect this to induce oversmoothing in the Dirichlet process Gaussian mixture model from Section \ref{sec::dens_squared}. We specified our (initial) Bayesian model according to (\ref{eqn::dpmm_model}) with the standard normal distribution for the base measure $G_{0}$ and the prior parameter values $\{M=1, a=1, b=1\}$.

Our goal is to estimate the integrated squared density as described in Section \ref{sec::dens_squared}. To evaluate the performance of the one-step posterior correction, we ran 1000 independent Monte Carlo trials for each of the two sample sizes $n=1000$ and $n=2000$. For both the uncorrected and corrected posteriors, point estimation was evaluated using the posterior mean, and intervals were constructed using the 95\% central credible region.

The numerical results can be found in Table \ref{tab::dpmm}. First, we can see that the one-step posterior improved on the initial posterior across all point estimation metrics. This was achieved through substantial bias reduction, which is illustrated visually by the histogram/density plots of the posterior mean in Figure \ref{fig::dpmm}. 

The improvements in coverage are even more striking. For the larger sample size of $n=2000$, the one-step posterior was able to attain nominal coverage despite the fact that the initial posterior only achieved 67.4\%. The poor performance of the initial posterior can be partially attributed to the aforementioned bias, which is of a similar order of magnitude to the average interval length. However, it also appears that the interval lengths are too short. This could have been caused by an overcurved likelihood or residual prior regularization, or both. In any case, it seems unlikely that the initial posterior satisfies the semiparametric BvM theorem for this Laplace example.

\begin{table}[h!]
\begin{center}
\resizebox{.8\textwidth}{!}{\begin{tabular}{|c|l |ccccc|} \hline
Sample size & Method & Bias & MAE & RMSE & Cov & Int. len. \\ \hline 
$n=1000$ & DPMM   & -0.0098 & 0.0108 & 0.0185 & 64.9\% & 0.0290  \\
& DPMM+1step & -0.0040 & 0.0072 & 0.0151 & 93.7\% & 0.0434 \\\hline
$n=2000$&DPMM   & -0.0068 & 0.0073 & 0.0138 & 67.4\% & 0.0210  \\
& DPMM+1step & -0.0025 & 0.0048 & 0.0116 & 95.8\% & 0.0324 \\\hline
\end{tabular}}
\end{center}
\caption{Monte Carlo numerical results comparing the integrated squared density posteriors with and without correction. MAE, median absolute error; RMSE, root mean squared error; Cov, coverage; Int. len., average 95\% central credible interval length.\label{tab::dpmm}}
\end{table}


\begin{figure}[]
\begin{center}
\includegraphics[width=6in]{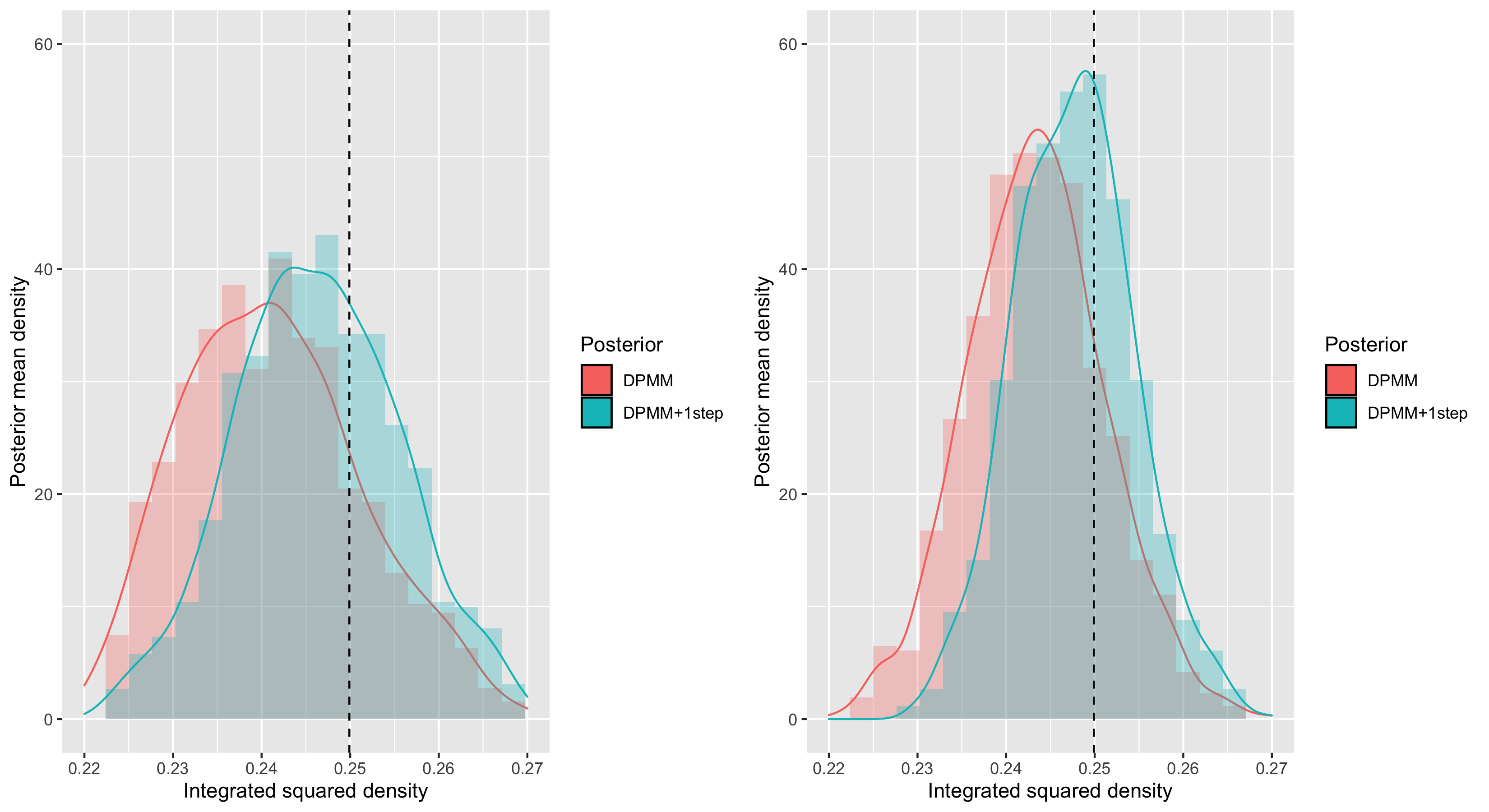}
\end{center}
\caption{Comparison of posterior mean density plots for the integrated squared density: $n=1000$ (left), and $n=2000$ (right). The true value of the functional is shown by the black vertical dashed line. \label{fig::dpmm}}
\end{figure}

For further context, we have included a ``before and after'' scatter plot in Figure \ref{fig::dpmm_scat} that compares the samples from both posteriors for a single iteration of the experiment ($n=2000$). This plot illustrates that, despite the randomness of the correction terms, the draws from the initial posterior tend to be pulled in the direction of the true functional value.

\begin{figure}[]
\begin{center}
\includegraphics[width=6in]{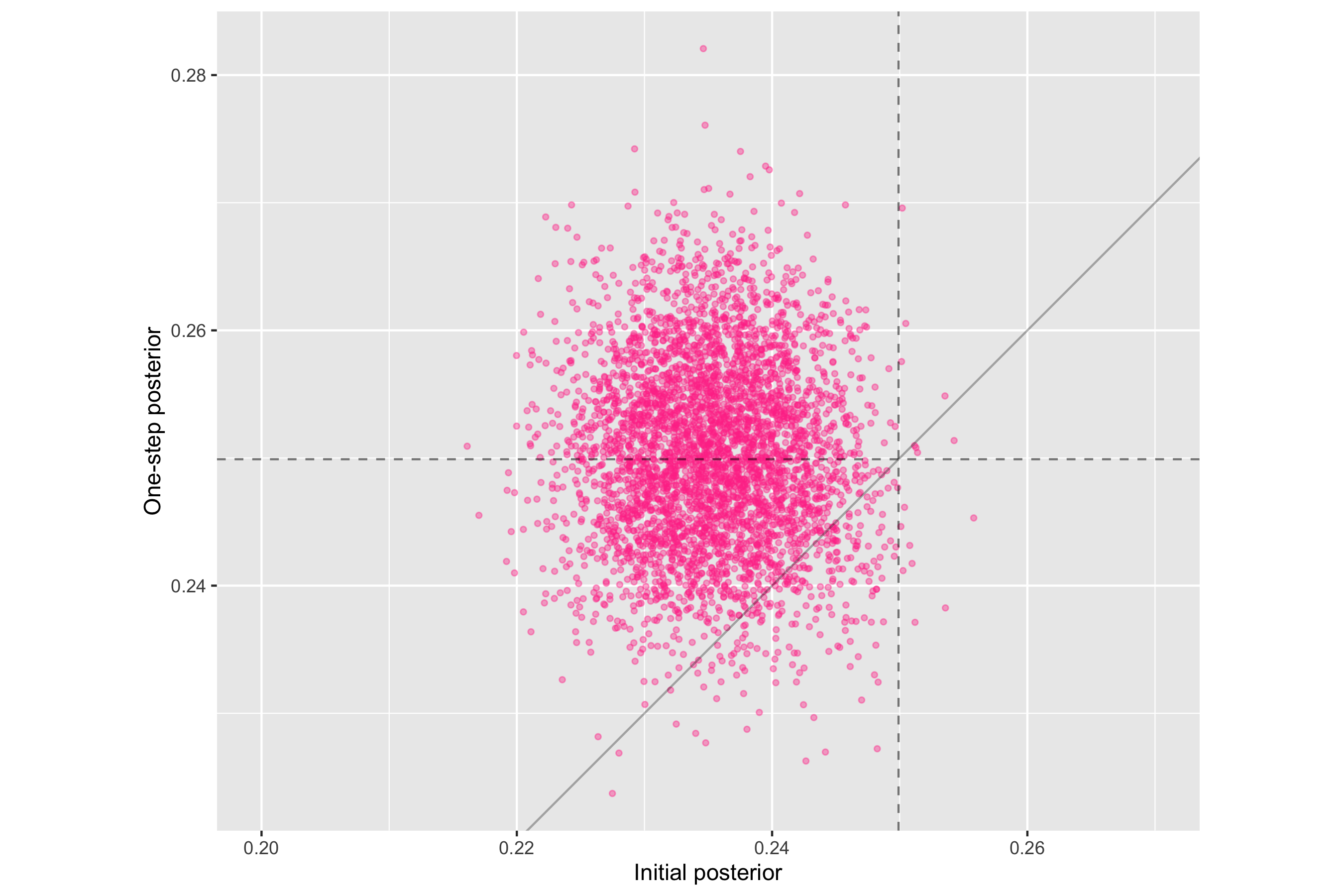}
\end{center}
\caption{Scatter plot comparing the initial posterior samples with their one-step corrected counterparts. We have also displayed the identity line (\fullgrey) and the true value of the estimand (\dashedgrey). \label{fig::dpmm_scat}}
\end{figure}

\subsection{The ACIC 2016 competition} \label{sec::acic}

The causal data analysis competition ``Is Your SATT Where It's At?'' was held at the 2016 Atlantic Causal Inference Conference with the goal of providing a comprehensive, objective comparison of the available methods for estimating causal effects. A total of 24 submissions were compared in the ``black box'' section of the competition, which assessed the automated performance of the methods across 77 different simulation settings. Detailed descriptions of the competition design, motivations and findings can be found in \citet{Dorie19}.

Each dataset $Z^{(n)} = (Z_{1},\ldots,Z_{n})$ had a sample size of $n=4802$ with the same data structure as Section \ref{sec::att}; that is, $Z_{i} = (X_{i},A_{i},Y_{i})$ for $i=1,\ldots,n$. The set of covariates $(X_{1},\ldots,X_{n})$ was fixed across datasets---they were taken from a real study that examined the impact of birth weight on a child's IQ and contained a mixture of categorical, binary, continuous and count variables (58 variables in total). Conditional on the covariates, the potential outcomes $(Y^{1}_{i}, Y^{0}_{i})$ and treatment assignments were generated i.i.d. according to the configuration of simulation parameters, which included the degree of nonlinearity of the regression functions, the treated percentage, and the treatment effect heterogeneity. One hundred independent datasets were generated from each of the 77 settings, resulting in a total of 7700 datasets.

The target estimand is the \textit{sample average treatment effect on the treated (SATT)}
\begin{equation} \label{eqn::satt}
    \beta = \frac{\mathbb{P}_{n}[A(Y^{1}-Y^{0})]}{\mathbb{P}_{n}[A]},
\end{equation}
which quantifies the average treatment effect among those treated in the dataset. But the SATT cannot be directly inferred in a semiparametric framework because we never observe $AY^{0}$ when $A=1$. An obvious alternative is to use the estimates for the (population) ATT as described in Section \ref{sec::att}, but this would likely produce substantially conservative intervals for the SATT \citep{Imbens04}. As a partial remedy, we will also consider estimates derived from the one-step posterior targeting a new estimand that we call the \textit{average covariate-conditional treatment effect on the treated (ACTT)}
\begin{equation*}
    \theta(P) = \frac{\mathbb{Q}_{n}[\pi(X)\{\mu^{(1)}(X)-\mu^{(0)}(X)\}]}{\mathbb{Q}_{n}[\pi(X)]},
\end{equation*}
which replaces the unknown population covariate distribution $Q$ in the ATT identification formula (\ref{eqn::att_iden}) with the empirical covariate distribution
\begin{equation*}
    \mathbb{Q}_{n} = \frac{1}{n}\sum_{i=1}^{n}\delta_{X_{i}}.
\end{equation*}
We refer the reader to Section \ref{sec::samp_treat} for further information about the ACTT, including details on implementing the one-step posterior correction.

For our initial posterior, we modelled both the conditional outcome distribution and propensity score with \textit{Bayesian additive regression trees (BART)} \citep{Chipman10}. As discussed at the end of Section \ref{sec::att}, we fitted the outcome regression model on both treatment arms together to pool information. In accordance with the ``black box'' nature of the competition, we used only the default hyperparameters in the R \texttt{BART} package implementations with no cross-validation (the propensity score model used the \texttt{pbart} procedure, which specifies a probit link). For each dataset, we obtained 4000 posterior samples aggregated from 4 independent chains after discarding 2000 burn-in samples. The posterior mean and the central 95\% credible interval were used for point and interval estimation respectively.

We highlight two top performers from the competition (descriptions of the remaining submissions can be found on p. 62 of \citet{Dorie19}). The ``BART on Pscore'' method \citep{Hahn20} fitted a BART outcome regression model that included an estimate of the propensity score (from a cross-validated probit BART model) as a splitting variable. A posterior for $\beta$ was then derived by replacing $(Y^{1}_{i}, Y^{0}_{i})$ with $(\mu^{(1)}(X_{i}), \mu^{(0)}(X_{i}))$ in (\ref{eqn::satt}) \citep{Hill11}. The argument put forth by \citet{Hahn20} was that their augmentation helps to mitigate regularization bias when estimating the outcome regression function because the propensity score is likely to be correlated with $\mu_{0}^{(0)}$ in practice, e.g. a physician may assign treatment more readily to patients with worse risk factors. Accordingly, the propensity score should be an informative transformation of the covariates for predicting the outcome, which makes it easier for the BART model to approximate the truth with finite samples.

The ``BART+TMLE'' method employed a correction procedure similar to that of the one-step estimator. \textit{Targeted maximum likelihood estimation (TMLE)} \citep{vanderLaan06, vanderLaan11} makes adjustments in the model space rather than the parameter space; an initial estimate $\hat{P}$ is mapped to a targeted estimate $\check{P}$ such that $\mathbb{P}_{n}[\dot{\chi}_{\check{P}}] \approx 0$ for the population ATT. Consequently, the TMLE plug-in $\chi(\check{P})$ is approximately equal to its own one-step estimator and its variance can be estimated by $\frac{1}{n}\mathbb{P}_{n}[(\dot{\chi}_{\check{P}}- \mathbb{P}_{n}[\dot{\chi}_{\check{P}}])^{2}]$ provided that $\check{P}$ satisfies the requisite conditions for asymptotic efficiency (see Section \ref{sec::grad_based}). In this case, the initial estimates of the outcome regression function and the propensity score were both derived from BART.

Figure \ref{fig::acic_rmse} compares the standardized point estimation errors across all settings. The top performer from the competition submissions was ``BART on Pscore''. In particular, it improved on ``BART MChains'', which was the same procedure aside from the propensity score augmentation (``MChains'' stands for multiple chains). We emphasize again, however, that the ``BART on Pscore'' augmentation does not target any particular low-dimensional functional, and the Bayesian update will instead use the estimated propensity score to help find a favourable bias-variance trade-off for the whole outcome regression function. In contrast, our one-step posteriors used the propensity score model to debias specifically in the direction of the ATT, which is perhaps why they were able to present further improvements in reducing the point estimation errors. However, as \citet{Hahn20} pointed out, there is the option of combining ``BART on Pscore'' with a targeted approach like TMLE, and by extension, we can also combine it with the one-step posterior correction to reap the benefits of both approaches. For additional context, we have included the results for the BART ACTT posterior without the one-step correction (``BART plug-in (ACTT)'') and it is clear that the one-step ACTT posterior performs substantially better.


\begin{figure}[]
\begin{center}
\includegraphics[width=5.5in]{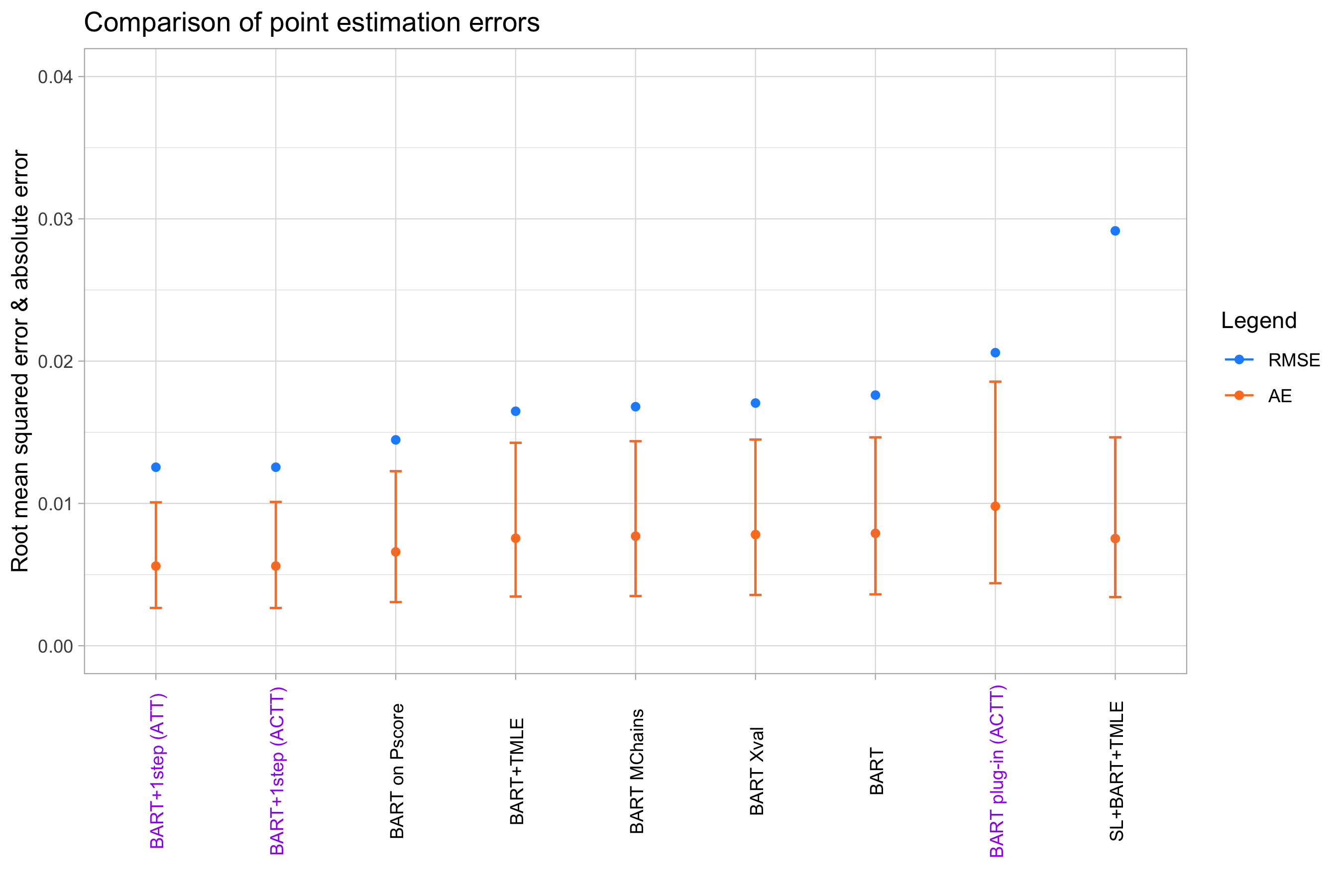}
\end{center}
\caption{Comparison of standardized point estimation errors for ACIC 2016. The competition submissions are shown in black. Root mean squared errors are shown by the blue dots; the median, lower quartile and upper quartile absolute errors across the 7700 datasets are shown in orange.\label{fig::acic_rmse}}
\end{figure}

The results for interval estimation can be found in Figure \ref{fig::acic_cov}. None of the competition submissions were able to attain the nominal coverage level of 95\% across all 7700 datasets, with ``BART+TMLE'' coming the closest. We recall that the intervals for ``BART+TMLE'' were constructed for the population ATT rather than the SATT, so we would expect them to be conservative. As the direct comparator, the one-step ATT posterior was able to attain a conservative level of coverage, albeit with slightly wider intervals. As expected, the one-step ACTT posterior had shorter interval lengths than both methods, but it was still able to attain nominal coverage. Taking both coverage and interval lengths into account, the one-step ACTT posterior was the clear winner in this category. As before, we have included the results for the uncorrected ACTT posterior, which performed slightly worse than the other plug-in BART methods like ``BART on Pscore'' and ``BART MChains''. 

\begin{figure}[]
\begin{center}
\includegraphics[width=5.5in]{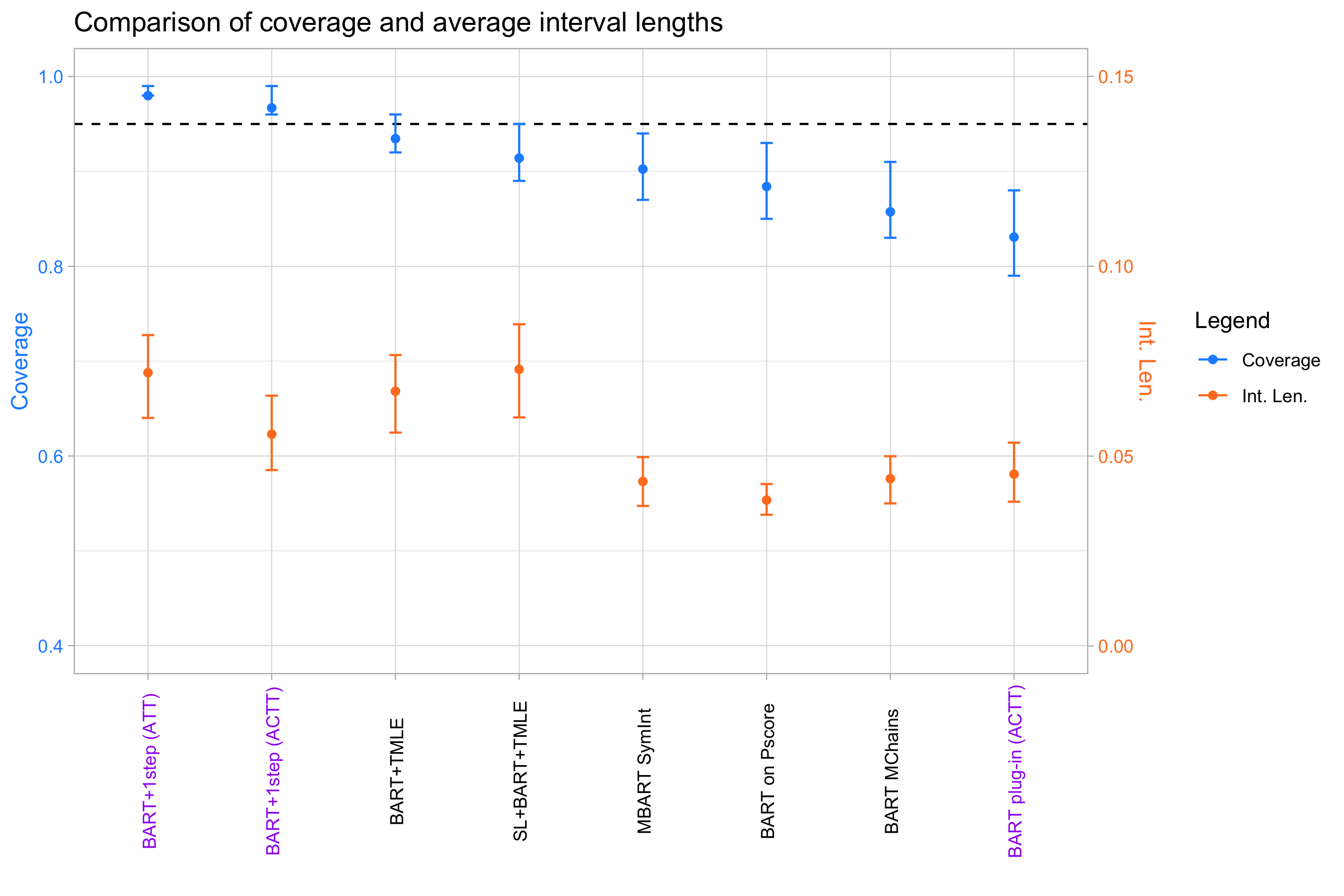}
\end{center}
\caption{Comparison of coverage and average standardized interval lengths for ACIC 2016. The competition submissions are shown in black. The dots are averages across all 7700 datasets, while the bars show the upper and lower quartiles for the 77 settings. The dashed horizontal line indicates the 95\% coverage level.\label{fig::acic_cov}}
\end{figure}

Earlier, we described how the variance of the TMLE estimator could be estimated by $\frac{1}{n}\mathbb{P}_{n}[(\dot{\chi}_{\check{P}}- \mathbb{P}_{n}[\dot{\chi}_{\check{P}}])^{2}]$, which is the sample variance of the estimated efficient influence function divided by the sample size. Other influence-function-based methods like one-step estimation take the same approach but with $\check{P}$ replaced by the original estimate $\hat{P}$ \citep{Hines22, Kennedy22}. This variance estimate is based on a first-order asymptotic approximation that assumes the second-order remainder terms to be negligible; in other words, the estimation uncertainty of $\check{P}$ (or $\hat{P}$) is effectively treated as being zero. Perhaps unsurprisingly, the intervals derived from these variance estimates tend to undercover in practice \citep{Gruber19}.

We can derive some further insight by formulating this in our one-step posterior correction framework. The classical one-step approach roughly corresponds to specifying a point-mass initial posterior on $\hat{P}$: $\Pi(\cdot \mid Z^{(n)}) = \delta_{\hat{P}}(\cdot)$, where the normal approximation used for constructing intervals is replaced by ``Bayesian bootstrapping'' the sampling distribution quantiles of $\mathbb{P}_{n}[\dot{\chi}_{\hat{P}}]$ akin to the \textit{percentile bootstrap} \citep{Efron16}. In the context of using a Bayesian procedure like BART, this makes little sense. Instead, our one-step posterior propagates the uncertainty from the whole BART posterior, which has the effect of extracting more information from the BART model. For example, is the posterior tightly concentrated around its mean or is it widely spread out? This gives an indication regarding the amount of information that is contained in the data, which serves to regularize our one-step posterior for better inference.

In conclusion, the one-step posterior is a much more natural and effective way to leverage Bayesian procedures for semiparametric inference than classical approaches like one-step estimation and TMLE. We expect there to be the potential for further performance gains for one-step BART posteriors if the user employs cross-validation to select hyperparameters or augments the outcome regression model with an estimated propensity score augmentation like ``BART on Pscore''. 

\section{Discussion}

The holistic nature of Bayesian inference---automatic uncertainty quantification, marginalizing nuisance parameters, inference for all estimands in one stroke---is often cited as one of its key strengths. This wisdom is subverted in the semiparametric setting. A Bayesian user has the unenviable choice of either tailoring their model towards a particular estimand at the expense of all others or run the risk of obtaining biased, uncalibrated inference. Our one-step posterior correction allows us to circumvent this problem by targeting estimands after the initial model has been fitted. 

The empirical evaluations provided suggest that our methodology can in fact outperform classical correction methods when using Bayesian algorithms. We have argued informally that this was achieved by drawing out more information from the initial posterior for regularization. It might be possible to offer more formal arguments by exploring higher-order properties of the one-step posterior; there are some precedents for this type of theory for semiparametric Bayes-like methods \citep{Cheng08a, Cheng08b}. 

A potential obstacle for our approach is the requirement that the efficient influence function can be computed across the whole model. This could be particularly challenging for separated semiparametric models, because the efficient influence function is defined in terms of orthogonal projections \citep{vanderVaart98} and may not admit a tractable expression. However, with recent developments in numerical procedures for approximating efficient influence functions \citep{Carone19, Jordan22}, we anticipate that it will be possible to fully automate the one-step posterior correction in future work. 

In this article, we have restricted our attention to conventional Bayesian models for deriving our initial posterior distribution. It is in fact possible to use other types of posteriors as long as they satisfy the requisite theoretical conditions, e.g. generalized posteriors \citep{Bissiri16}, martingale posteriors \citep{Fong21}, and posteriors arising from approximate computational algorithms like variational Bayes \citep{Jordan99}. These frameworks still abide by the plug-in principle and may therefore have the potential to benefit from our methodology.

\section{Acknowledgments}
The authors thank Edward Kennedy and Kolyan Ray for their helpful input. \textbf{AY} receives funding from Novo Nordisk. \textbf{CH} is supported by The Alan Turing Institute, the Li Ka Shing Foundation, the Medical Research Council, the EPSRC through the Bayes4Health
grant EP/R018561/1, AI for Science and Government UKRI, and the U.K. Engineering and Physical Sciences Research Council. \textbf{JR} has received funding from the European Research Council
(ERC) under the European Union’s Horizon 2020 research and innovation programme (grant agreement No 834175).

\bibliographystyle{abbrvnat}
\bibliography{reference}

\appendix

\section{An overview of semiparametric theory} \label{supp::overview}

\subsection{Differentiable functionals and asymptotic efficiency} \label{supp::semi_eff}

Semiparametric efficiency theory is concerned with estimators that are insensitive to small local changes in the data-generating distribution in any direction. The notion of ``direction'' is formalized by considering ``smooth'', one-dimensional paths contained in the model $\mathcal{P}$ that pass through a distribution $P \in \mathcal{P}$, and the set of permitted directions is called the \textit{tangent space} $\dot{\mathcal{P}}_{P} \subseteq L_{2}(P)$ of the model $\mathcal{P}$ at $P$. The precise notion of ``smoothness'' is defined as follows. 
\begin{definition}
A \textit{parametric submodel} $\{P_{t,g}:  t \in [0, \varepsilon), \varepsilon > 0\} \subset \mathcal{P}$ with $P_{t,g}\rvert_{t=0} = P$ is \textit{differentiable in quadratic mean} with \textit{score function} $g \in \dot{\mathcal{P}}_{P}$ if
\begin{equation} \label{eqn::dqm}
    \lim_{t\downarrow 0}\int \left[\frac{dP^{1/2}_{t,g}-dP^{1/2}}{t}-\frac{1}{2}g\, dP^{1/2}\right]^{2} = 0.
\end{equation}
In words, a parametric submodel is a one-dimensional model parameterized by $t$ that is contained in $\mathcal{P}$ and passes through $P$ at $t=0$ with score function $g$, which parameterizes the direction in which the submodel approaches $P$ as $t \downarrow 0$. The resulting high-level interpretation of the tangent space is the set of directions in which we can move an infinitesimal distance away from $P$ and still remain in the model $\mathcal{P}$.
\end{definition}

The definition of the score function is closely related to the more familiar notion of
\begin{equation} \label{eqn::scoredef}
    g(z) = \frac{\partial}{\partial t} \log dP_{t,g}(z)\rvert_{t=0},
\end{equation}
although neither definition is strictly stronger than the other. The quadratic mean version is more appropriate because it exactly ensures a type of \textit{local asymptotic normality} (see Lemma 1.9 in \citet{vanderVaart02}). Despite that, explicit constructions of parametric submodels will usually make use of (\ref{eqn::scoredef}) instead (see the examples in Lecture 1 of \citet{vanderVaart02}).

A fundamental requirement for semiparametric efficiency theory is (pathwise) differentiability of the target functional. 

\begin{definition}
A functional $\chi$ \textit{(pathwise) differentiable} at $P$ with respect to $\dot{\mathcal{P}}_{P}$ if:
\begin{itemize}
    \item[(a)] the mapping $t \mapsto \chi(P_{t,g})$ is differentiable, and
    \item[(b)] there exists a fixed, measurable function $\psi_{P}: \mathcal{Z} \rightarrow \mathbb{R}$ such that
\begin{equation} \label{eqn::graddef}
    \frac{\partial\chi(P_{t,g})}{\partial t}\rvert_{t=0} = P[\psi_{P} g]
\end{equation}
for every $g \in \dot{\mathcal{P}}_{P}$ and any parametric submodel $\{P_{t,g}\}$ with score function $g$. We call $\psi_{P}$ a \textit{gradient} of $\chi$ at $P$.
\end{itemize}
\end{definition}

The definition of pathwise differentiability above can be motivated as follows. Suppose we try to form a distributional Taylor expansion
\begin{equation} \label{eqn::vonmises}
    \chi(P_{t,g})-\chi(P) = (P_{t,g}-P)[\psi_{P}] + r_{2}(P_{t,g},P),
\end{equation}
where $r_{2}(P_{t,g},P)$ is simply defined as the left-hand side minus the first term on the right. Now suppose that we divided both sides by $t$ and took the limit as $t \rightarrow 0$. It is clear that the left-hand side converges to $\partial\chi(P_{t,g})/\partial t\rvert_{t=0}$. Assuming that (\ref{eqn::scoredef}) holds and that the order of differentiation and integration can be exchanged, the first term on the right converges to
\begin{equation*} 
    \frac{\partial}{\partial t} P_{t,g}[\psi_{P}] \rvert_{t=0} = \int \psi_{P} \frac{\partial}{\partial t} dP_{t,g}\rvert_{t=0} = \int \psi_{P} \left(\frac{\partial}{\partial t} \log dP_{t,g}\right)dP_{t,g}\rvert_{t=0}= P[\psi_{P} g].
\end{equation*}
By the definition of $\psi_{P}$ in (\ref{eqn::graddef}), this must mean that $\lim_{t \rightarrow 0}r_{2}(P_{t,g},P)/t = 0$.

This suggests that a gradient can be viewed as a type of first-order distributional derivative, which takes as input the direction of change in the data-generating distribution and outputs the corresponding slope in $\chi$. Furthermore, the remainder term $r_{2}(P_{t,g},P)$ must depend on the difference between $P_{t,g}$ and $P$ in some higher-order way that allows it to vanish at a faster-than-linear rate. This is a crucial aspect that enables the use of flexible algorithms to estimate the data-generating distribution while retaining $\sqrt{n}$-consistent estimation of the target estimand. The ``second-order'' nature of the remainder is also a key ingredient of the so-called \textit{doubly robust} property enjoyed by some estimators, which will be described later.

Gradients are only unique up to summation with elements that are orthogonal to the tangent space. Among all gradients with mean zero, the one with the smallest variance is called the \textit{canonical gradient} or \textit{efficient influence function}, which we denote by $\dot{
\chi}_{P}$. This is the unique projection of any mean-zero gradient onto the closure of the linear span of the tangent space. As stated in the main paper, we say that an estimator $\hat{\chi}_{n} = \hat{\chi}_{n}(Z^{(n)})$ is \textit{asymptotically efficient} at $P_{0}$ if it satisfies \eqref{eqn::asympeff}:
\begin{equation*} 
    \sqrt{n}(\hat{\chi}_{n}-\chi(P_{0})) = \sqrt{n}\mathbb{P}_{n}[\dot{\chi}_{P_{0}}] + o_{P_{0}}(1).
\end{equation*}

\subsection{Gradient-based estimation} \label{sec::grad_based}

In this subsection, we give an account of an influence-function-based approach for obtaining semiparametric efficient estimators that has now become standard in the literature and widely applied \citep[e.g.][]{vanderLaan06, vanderVaart14, Kennedy22, Hines22}.

The most intuitive approach for estimating $\chi$ proceeds by constructing an estimate $\hat{P}$ of the entire data-generating distribution and ``plugging'' this into the functional mapping to obtain $\chi(\hat{P})$. For example, if $\mathcal{P}$ is a parametric model and $\hat{P}$ is the maximum likelihood estimate (MLE), then $\chi(\hat{P})$ is the MLE of $\chi$ and therefore enjoys the familiar optimal asymptotic properties from classical theory.

Since we are interested in nonparametric models, it appears natural to consider plugging in the nonparametric MLE: the empirical distribution $\mathbb{P}_{n}$. This works for certain estimands---e.g. means and quantiles---in the sense that the \textit{functional delta method} \citep[e.g. Chapter 20 of][]{vanderVaart98} will guarantee $\sqrt{n}$-consistency, asymptotic normality and asymptotic efficiency. The discrete nature of the empirical distribution, however, precludes estimands that depend on densities or regression functions on continuous variables. Plug-in estimation for such functionals must employ a smoothed estimate of the data-generating distribution. If we go beyond parametric models, we have to be careful that the smoothing and increased flexibility does not induce a non-negligible bias that prevents $\sqrt{n}$-consistency.

A crude heuristic for interpreting this bias can be derived from expanding $\chi(P_{0})$ around $\chi(\hat{P})$ in a similar fashion to (\ref{eqn::vonmises}):
\begin{equation} \label{eqn::plugexp}
    \chi(P_{0})-\chi(\hat{P}) = (P_{0}-\hat{P})[\dot{\chi}_{\hat{P}}] + r_{2}(P_{0},\hat{P}).
\end{equation}
In general, we would expect the linear term on the right to vanish on the order of $d(P_{0}, \hat{P})$ for some appropriate distance $d$. Furthermore, as discussed previously, we expect the remainder $r_{2}(P_{0},\hat{P})$ to vanish at a higher-order rate (perhaps quadratic in $d(P_{0}, \hat{P})$). Thus, the asymptotic behaviour of our plug-in estimator $\chi(\hat{P})$ will tend to be dominated by the linear ``first-order bias'' term $(P_{0}-\hat{P})[\dot{\chi}_{\hat{P}}]$. 

If $\hat{P}$ is $\sqrt{n}$-consistent, this heuristic suggests that the plug-in estimator should also be $\sqrt{n}$-consistent. But if we choose to use flexible nonparametric methods, there is a risk that the slower rate of convergence from estimating $P$ will bleed into the estimation of $\chi$. In order to achieve $\sqrt{n}$-consistency and asymptotic efficiency, it may be necessary to ``undersmooth'' the nonparametric estimator/prior to reduce bias \citep{Goldstein92,  Newey98}, which is difficult to implement in practice.

The expansion in (\ref{eqn::plugexp}) suggests that it might be possible to improve on the plug-in estimator by estimating the first-order bias and removing it from the plug-in estimator. The canonical gradient at any $P \in \mathcal{P}$ has mean zero under $P$ by definition, so the first term on the right of (\ref{eqn::plugexp}) equals $P_{0}[\dot{\chi}_{\hat{P}}]$. Since we do not know $P_{0}$, the obvious approach is replace it with the empirical measure $\mathbb{P}_{n}$. Taking $\mathbb{P}_{n}[\dot{\chi}_{\hat{P}}]$ from both sides of (\ref{eqn::plugexp}) leads to a new expansion:
\begin{equation} \label{eqn::1stepexp}
\begin{split}
    \chi(P_{0})-\chi(\hat{P}) - \mathbb{P}_{n}[\dot{\chi}_{\hat{P}}] &= -(\mathbb{P}_{n} -P_{0})[\dot{\chi}_{\hat{P}}] + R_{2}(P_{0},\hat{P})\\
    &= \underbrace{-(\mathbb{P}_{n}-P_{0})[\dot{\chi}_{P_{0}}]}_{\circled{1}}+\underbrace{(\mathbb{P}_{n}-P_{0})[\dot{\chi}_{P_{0}} - \dot{\chi}_{\hat{P}}]}_{\circled{2}}+ \underbrace{r_{2}(P_{0},\hat{P})}_{\circled{3}},
\end{split}
\end{equation}
where the second equality follows simply from adding and subtracting $(\mathbb{P}_{n}-P_{0})[\dot{\chi}_{P_{0}}]$. Thus, under the following assumptions to control terms $\circled{2}$ and $\circled{3}$ \citep[e.g.][]{vanderVaart14}:
\begin{assumption} \label{ass::emp_proc}
$\mathbb{G}_{n}[\dot{\chi}_{\hat{P}} - \dot{\chi}_{P_{0}}] = o_{P_0}(1)$ 
\end{assumption}
\begin{assumption} \label{ass::sec_order}
$\sqrt{n}r_{2}(P_{0},\hat{P}) = o_{P_0}(1)$
\end{assumption}
\noindent we have
\begin{equation*}
    \sqrt{n}(\chi(\hat{P})+\mathbb{P}_{n}[\dot{\chi}_{\hat{P}}]-\chi(P_{0})) = \sqrt{n}\mathbb{P}_{n}[\dot{\chi}_{P_{0}}]+o_{P_{0}}(1),
\end{equation*}
which is precisely the condition displayed in (\ref{eqn::asympeff}) for asymptotic efficiency. This motivates replacing the plug-in estimator with the new \textit{one-step estimator}
\begin{equation*}
    \hat{\chi}_{\text{1-step}} = \chi(\hat{P})+\mathbb{P}_{n}[\dot{\chi}_{\hat{P}}].
\end{equation*}

Assumption \ref{ass::emp_proc} is implied by the following pair of conditions (Lemma 19.24 of \citet{vanderVaart98}):
\begin{assumption} \label{ass::emp_proc_alt} \leavevmode 
\begin{itemize}
    \item[(a)] $\Vert \dot{\chi}_{\hat{P}} - \dot{\chi}_{P_{0}} \Vert_{P_0} = o_{P_0}(1)$
    \item[(b)] $\dot{\chi}_{\hat{P}}$ is contained in a fixed $P_{0}$-Donsker class with probability tending to 1,
\end{itemize}
\end{assumption}
\noindent where the Donsker class assumption \citep{vanderVaart96} restricts the complexity of the estimators. Assumption \ref{ass::emp_proc_alt}(b) can be removed by \textit{sample-splitting} \citep{Chernozhukov18}.

The way to control term $\circled{3}$ depends on the problem, but as stated earlier, we would expect it to be quadratic in $d(P_{0}, \hat{P})$ for some appropriate distance $d$. If this is indeed true, it would be sufficient for the distance to shrink like $o_{P_0}(n^{-1/4})$, which is a vast improvement on the $\sqrt{n}$-convergence that is potentially required for a na\"ive plug-in estimator. For problems in causal inference, this second-order remainder often takes the beneficial form of a product of errors (a ``cross-term'') of different components of the data-generating distribution, i.e.
\begin{equation*}
    r_{2}(P_{0}, \hat{P}) = O_{P_0}\left(\Vert \eta(P_{0})-\eta(\hat{P}) \Vert_{P_{0}}\Vert \xi(P_{0})-\xi(\hat{P}) \Vert_{P_{0}}\right),
\end{equation*}
where $\eta$ and $\xi$ are nuisance components of the data-generating distribution. In this case, one of the errors can be allowed to decay slower than $o_{P_0}(n^{-1/4})$ as long as it is compensated by the other error. This is referred to as \textit{(rate) double-robustness} \citep{Rotnitzky21}.


\section{Weak convergence and the bounded Lipschitz distance} \label{supp::weak_conv}

We collect some definitions and results about weak convergence and the bounded Lipschitz metric. For a real-valued function $f:\mathbb{R} \rightarrow \mathbb{R}$, the \textit{bounded Lipschitz norm} is defined as
\begin{equation*}
    \lVert f \lVert_{BL} = \lVert f \lVert_{\infty} + \lVert f \lVert_{L},
\end{equation*}
where
\begin{align*}
    \lVert f \lVert_{\infty} &= \sup_{x} |f(x)|,\text{ and} \\
    \lVert f \lVert_{L} &= \sup_{x \neq y} \frac{|f(x)-f(y)|}{|x-y|}.
\end{align*}
For two laws $P$ and $Q$ on $\mathbb{R}$, the \textit{bounded Lipschitz metric} is defined as
\begin{equation*}
    d_{BL}(P,Q) = \sup\left\{\left|\int f(dP-dQ)\right|:\,\lVert f \lVert_{BL} \leq 1 \right\}.
\end{equation*}
\begin{theorem}[Theorem 11.3.3, p.395 of \citet{Dudley02}]
For a sequence of laws $(P_{n})_{n}$ and $P_{0}$ defined on $\mathbb{R}$, the following are equivalent:
\begin{itemize}
    \item[(a)] $P_{n}$ converges weakly to $P_{0}$.
    \item[(b)] $d_{BL}(P_{n},P) \rightarrow 0$.
\end{itemize}
\end{theorem}

The following lemma adapts Slutsky's lemma to the setting with random sequences of probability measures. This will be used extensively in the proofs of our main results.
\begin{lemma} \label{lem::slutsky}
Let $Z^{(n)} = (Z_{1},\ldots,Z_{n})$ be i.i.d. variables from a distribution $P_{0}$ on a Polish sample space $(\mathcal{Z},\mathcal{A})$. Suppose that $(P_{n})_{n}$ is a sequence of random probability measures on $(\mathbb{R}^{2}, \mathcal{B}(\mathbb{R}^{2}))$ such that $P_{n}$ is $\sigma(Z^{(n)})$-measurable for each $n$. Let $(X_{n},Y_{n})$ be variables each taking values in $\mathbb{R}$ with $(X_{n},Y_{n})\mid P_{n} \sim P_{n}$ and denote the marginals by $P^{X}_{n}$ and $P^{Y}_{n}$ for $X_{n}$ and $Y_{n}$ respectively. Suppose that
\begin{align*}
    d_{BL}(P^{X}_{n},P^{X}) \xrightarrow[]{P_{0}} 0\\
    d_{BL}(P^{Y}_{n},\delta_{\{c\}}) \xrightarrow[]{P_{0}} 0,
\end{align*}
where $P^{X}$ is a fixed probability measure on $(\mathbb{R},\mathcal{B}(\mathbb{R}))$, and $c$ is a fixed constant in $\mathbb{R}$. Then, 
\begin{align*}
    d_{BL}(\mathcal{L}(X_{n}+Y_{n} \mid P_{n}),\mathcal{L}(X_{n}+c \mid P_{n})) \xrightarrow[]{P_{0}} 0 \\
    d_{BL}(\mathcal{L}(X_{n}Y_{n} \mid P_{n}),\mathcal{L}(cX_{n} \mid P_{n})) \xrightarrow[]{P_{0}} 0.
\end{align*}
\end{lemma}

\begin{proof}
We denote $\mathbb E_{P_n}[\cdot ]$ for the expectation under the conditional distribution of $(X_n, Y_n)$ given $P_n$.
All equalities and inequalities between random quantities are to be understood in the almost sure sense. For any $\varepsilon > 0$, and any $h$ such that $\lVert h \lVert_{BL} \leq 1$, 
\begin{align*}
   & \left|\mathbb{E}_{P_n}\left\{h(X_{n}+Y_{n})-h(X_{n}+c)\right\}\right| \leq  \mathbb{E}_{P_n}\left|h(X_{n}+Y_{n})-h(X_{n}+c)\right|\\
    &\leq  \mathbb{E}_{P_n}\left\{\left|h(X_{n}+Y_{n})-h(X_{n}+c)\right|1\left(|Y_{n}-c| \leq \frac{\varepsilon}{2}\right)\right\} +  
    \mathbb{E}_{P_n}\left\{\left|h(X_{n}+Y_{n})-h(X_{n}+c)\right|1\left(|Y_{n}-c| > \frac{\varepsilon}{2}\right)\right\} \\
    &\leq  \,\, \mathbb{E}_{P_n}\left\{|Y_{n}-c|1\left(|Y_{n}-c| \leq \frac{\varepsilon}{2}\right)\right\} + 2 P_n\left(|Y_{n}-c| > \frac{\varepsilon}{2}\right) \\
    &\leq  \,\, \frac{\varepsilon}{2} +  2 P_n\left(|Y_{n}-c| > \frac{\varepsilon}{2}\right).
\end{align*}
The third inequality follows from applying $\lVert h \lVert_{L} \leq 1$ and $\lVert h \lVert_{\infty} \leq 1$ to the first and second terms respectively.

We now show that 
\begin{equation*}
    P_n\left(|Y_{n}-c| > \frac{\varepsilon}{2}\right) \xrightarrow[]{P_{0}} 0,
\end{equation*}
which would establish the first claim of Lemma \ref{lem::slutsky}, since 
\begin{equation*}
    \mathbb{P}\left(\sup_{\lVert h \lVert_{BL} \leq 1}\left|\mathbb{E}_{P_n}\left\{h(X_{n}+Y_{n})-h(X_{n}+c)\right\}\right| > \varepsilon \right) \leq \mathbb{P}\left(2P_n\left(|Y_{n}-c| > \frac{\varepsilon}{2}\right) > \frac{\varepsilon}{2}\right) \rightarrow 0.
\end{equation*}
Consider the function
\begin{equation*}
    g_{\gamma}(y) = \left[1- \gamma\,\min_{|y^{\prime}-c| \geq \frac{\varepsilon}{2}}|y-y^{\prime}|\right]\vee 0,
\end{equation*}
where $c > 0$. This function dominates $1(|y-c| > \varepsilon/2)$ and satisfies $\lVert g_{\gamma} \lVert_{L} = \gamma$. If we set $\gamma > 2/\varepsilon$, then $g_{\gamma}(c) = 0$. For sufficiently large $M(\gamma) > 0$, we have $\lVert g_{\gamma}/M(\gamma) \lVert_{BL} \leq 1$, so 
\begin{equation*}
    P_n\left(|Y_{n}-c| > \frac{\varepsilon}{2}\right) \leq \mathbb{E}_{P_n}[g_{\gamma}(Y_{n})]\xrightarrow[]{P_{0}} g_{\gamma}(c) = 0.
\end{equation*}

For the second claim of Lemma \ref{lem::slutsky}, let $\varepsilon,R > 0$ . Consider
\begin{align*}
    P_n[|X_{n}(Y_{n}-c)| > \varepsilon] =&\,\, P_n( \{|X_{n}(Y_{n}-c)| > \varepsilon\} \cap \{|X_{n}| > R\}) +  P_n( \{|X_{n}(Y_{n}-c)| > \varepsilon\} \cap \{|X_{n}| \leq R\}) \\
    \leq &\,\, P_n[|X_{n}| > R] + P_n[R|Y_{n}-c| > \varepsilon].
\end{align*}
We already showed earlier that the second term on the last line tends to zero in probability for fixed $R$. Consider the function
\begin{align*}
    g_{\gamma}(x) &= \left[1- \gamma\,\min_{|x^{\prime}| \geq R}|x-x^{\prime}|\right]\vee 0
\end{align*}
for $\gamma > 0$. For any fixed $\gamma$, we can find $M(\gamma) > 0$ such that $\lVert g_{\gamma}/M(\gamma) \lVert_{BL} \leq 1$. Thus, we have that $\mathbb{E}_{P_n}[g_{\gamma}(X_{n})]$ converges in probability to $\mathbb{E}_{P^X}[g_{\gamma}(X)]$. Also, $g_{\gamma}$ converges pointwise to $1(|x| \geq R)$ as $\gamma \rightarrow \infty$ and is dominated by 1, so $\mathbb{E}_{P^X}[g_{\gamma}(X)]$ converges from above to $P^{X}(|X| \geq R)$ by the dominated convergence theorem. We have the upper bound:
\begin{equation*}
   P_n(|X_{n}| > R) \leq \mathbb{E}_{P_n}[g_{\gamma}(X_{n})].
\end{equation*}
Let $\delta > 0$ and fix $R,\gamma$ large enough such that
\begin{align*}
    P^{X}(|X| \geq R) &< \frac{\delta}{8} , \quad 
    \mathbb{E}_{P^{X}}[g_{\gamma}(X)] - P^{X}(|X| \geq R) < \frac{\delta}{8}
\end{align*}
so that $\mathbb{E}_{P^{X}}[g_{\gamma}(X)] < \delta/4$. Then
\begin{align*}
    \mathbb{P}\left(\mathbb{E}_{P_n}[g_{\gamma}(X_{n})] > \frac{\delta}{2}\right) &\leq \mathbb{P}\left(\left\{\left|\mathbb{E}_{P_n}[g_{\gamma}(X_{n})]-\mathbb{E}_{P^X}[g_{\gamma}(X)]\right| > \frac{\delta}{4}\right\}\right)\rightarrow 0.
\end{align*}
Therefore,
\begin{align*}
    \mathbb{P}[P_n(|X_{n}(Y_{n}-c)| > \varepsilon)> \delta] &\leq \mathbb{P}\left(\left\{\mathbb{E}_{P_n}[g_{\gamma}(X_{n})] > \frac{\delta}{2}\right\}\cup \left\{P_n(R|Y_{n}-c| > \varepsilon)] > \frac{\delta}{2}\right\} \right) \\
    & \leq \mathbb{P}\left(\mathbb{E}_{P_n}[g_{\gamma}(X_{n})] > \frac{\delta}{2}\right) + \mathbb{P}\left( P_n(R|Y_{n}-c| > \varepsilon) > \frac{\delta}{2} \right)  \rightarrow 0.
\end{align*}
Thus, $P_n[|X_{n}(Y_{n}-c)| > \varepsilon]$ converges in probability to zero.

Finally, we have, using the same arguments as for $X_n+Y_n$; for all $h$ such that $\lVert h \lVert_{BL} \leq 1$
\begin{align*}
    \left|\mathbb{E}_{P_n}\left\{h(X_{n}Y_{n})-h(cX_{n})\right\}\right| \leq & E_{P_n}\left|h(X_{n}Y_{n})-h(cX_{n})\right|\\
    \leq & \,\, \mathbb{E}_{P_n}\left\{|X_{n}(Y_{n}-c)|1\left(|X_{n}(Y_{n}-c)| \leq \varepsilon\right)\right\} + 2 P_n\left(|X_{n}(Y_{n}-c)| > \varepsilon\right) \\
    \leq & \,\, \varepsilon + 2 P_n\left(|X_{n}(Y_{n}-c)| > \varepsilon\right),
\end{align*}
which shows that the left-hand side converges to zero in probability since $\varepsilon$ was arbitrary.
\end{proof}

\section{The Bernstein-von Mises theorem for the Bayesian bootstrap} \label{supp::bvm_bb}

Under Assumption \ref{ass::1step_bvm}(c$^{*}$), we can prove the semiparametric Bernstein-von Mises theorem for the one-step corrected posterior without assuming a Donsker-type class condition. In this section, we adapt Lemma 2 from \citet{Ray21} to provide the crucial ingredient: a nonparametric Bernstein-von Mises theorem for the Bayesian bootstrap posterior with only a Glivenko-Cantelli-type condition. According to personal communications with Kolyan Ray, the conditions stated in Remark 2 of \citet{Ray21} are in fact insufficient to establish their Lemmas 2 and 5 when we consider a sequence of classes rather than a single fixed class. We rectify this in the case of proving convergence in probability by adding conditions on the moments of the envelope functions.

\begin{lemma} \label{lem::bvm_bb}
We work in the same set-up as Section \ref{sec::setup} in the main paper. Suppose $(\mathcal{G}_{n})_{n}$ is a sequence of separable classes of measurable functions $g:\mathcal{Z}\rightarrow \mathbb{R}$ with envelope functions $(G_{n})_{n}$ satisfying $\lim_{C \rightarrow \infty} \limsup_{n \rightarrow \infty}P_{0}G_{n}^{2}1_{G_{n}^{2} > C} = 0$, $P_{0}G_{n}^{4} = o(n)$, and 
\begin{equation*}
    \sup_{g \in \mathcal{G}_{n}} |(\mathbb{P}_{n}-P_{0})[g]| \xrightarrow[]{P_{0}} 0.
\end{equation*}
Then
\begin{equation*}
    \sup_{g \in \mathcal{G}_{n}}d_{BL}\left(\mathcal{L}(\sqrt{n}(\Tilde{P}-\mathbb{P}_{n})[g] \mid Z^{(n)}),\mathcal{N}(0,P_{0}[(g-P_{0}[g])^{2}])\right)\xrightarrow[]{P_{0}} 0.
\end{equation*}
\end{lemma}

\begin{proof}
Let $W_{1},W_{2},\ldots$ be i.i.d. exponential variables with mean 1, and let $\bar{W}_{n} = n^{-1}\sum_{i=1}^{n}W_{i}$, which converges almost surely to 1 by the strong law of large numbers. We can write
\begin{equation*}
    \sqrt{n}(\Tilde{P}-\mathbb{P}_{n})[g] = \frac{1}{\bar{W}_{n}}\frac{1}{\sqrt{n}}\sum_{i=1}^{n}(W_{i}-1)(g(Z_{i})-\mathbb{P}_{n}[g]).
\end{equation*}
Let $\sigma_{n,g}^{2} = \mathbb{P}_{n}[(g-\mathbb{P}_{n}[g])^{2}]$ and $\sigma_{g}^{2} = P_{0}[(g-P_{0}[g])^{2}]$. We apply Lemma 11 of \citet{Ray20} with the classes of functions $\mathcal{H}_{n,1} = \mathcal{G}_{n}$ and continuous map $\phi(h_{1}) = h_{1}^{2}$. By assumption, $\mathcal{G}_{n}$ is separable and the envelope functions $(G_{n}^{2})_{n}$ of $\phi(\mathcal{H}_{n,1})$ are asymptotically uniformly integrable. Thus,
\begin{equation*}
    \sup_{g \in \mathcal{G}_{n}} |(\mathbb{P}_{n}-P_{0})[g^{2}]| \xrightarrow[]{P_{0}} 0,
\end{equation*}
which further implies that
\begin{equation*}
    \sup_{g \in \mathcal{G}_{n}} |\sigma_{n,g}^{2}-\sigma_{g}^{2}| \xrightarrow[]{P_{0}} 0.
\end{equation*}
Let $H(u) = u + u^{1/3}$ and $H_{0}(u) = u^{1/4}(1+|\log u |^{1/2})$. By Lemma 1 of \citet{Ray21}, we have
\begin{align} \label{eqn::dbl_bvm1}
    &d_{BL}\left(\mathcal{L}\left(\frac{1}{\sqrt{n}}\sum_{i=1}^{n}(W_{i}-1)(g(Z_{i})-\mathbb{P}_{n}[g]) \mid Z^{(n)}\right), \mathcal{N}(0,\sigma_{n,g}^{2})\right)\\ \label{eqn::dbl_bvm2}
    &\lesssim \varepsilon + \frac{1}{\varepsilon^{2}}H\left(\sup_{g \in \mathcal{G}_{n}}\sigma_{n,g}^{2}E(W_{1}-1)^{2}1_{|W_{1}-1|2 \max_{1 \leq i \leq n}G_{n}(Z_{i}) > \varepsilon \sqrt{n}}\right)+H_{0}\left(\varepsilon \sup_{g \in \mathcal{G}_{n}}\sigma_{n,g}^{2}\right)
\end{align}
for all $0 < \varepsilon < 1$. Since $P_{0}G_{n}^{4} = o(n)$ by assumption, Chebyshev's inequality implies that 
\begin{equation*}
    |\mathbb{P}_{n}G_{n}^{2} - P_{0}G_{n}^{2}| \xrightarrow[]{P_{0}} 0.
\end{equation*}
With $\sup_{g \in \mathcal{G}_{n}}\sigma_{n,g}^{2} \leq \mathbb{P}_{n}G_{n}^{2}$ and $P_{0}G_{n}^{2} = O(1)$ by assumption, we have $\sup_{g \in \mathcal{G}_{n}}\sigma_{n,g}^{2} = O_{P_0}(1)$. For any $C > 0$,
\begin{align*}
    \frac{\left[\max_{1 \leq i \leq n} G_{n}(Z_{i})\right]^{2}}{n} &= \max_{1 \leq i \leq n}\frac{G_{n}(Z_{i})^{2}}{n}\\
    &\leq \frac{C^{2}}{n} + \frac{1}{n}\sum_{i=1}^{n}G_{n}(Z_{i})^{2}1_{\{G_{n}(Z_{i})>C\}} \\
    &\leq \frac{C^{2}}{n} + P_{0}G_{n}^{2}1_{\{G_{n}>C\}} + |\mathbb{P}_{n}G_{n}^{2}1_{\{G_{n}>C\}} - P_{0}G_{n}^{2}1_{\{G_{n}>C\}}|.
\end{align*}

Let $\eta > 0$. By the assumption of asymptotic uniform integrability of $(G_{n}^{2})_{n}$, we can find some large enough $C > 0$ such that 
\begin{equation*}
    \frac{C^{2}}{n} + P_{0}G_{n}^{2}1_{\{G_{n}>C\}} < \eta /2
\end{equation*}
for all sufficiently large $n$. Using Chebyshev's inequality again, we also have
\begin{equation*}
    |\mathbb{P}_{n}G_{n}^{2}1_{\{G_{n}>C\}} - P_{0}G_{n}^{2}1_{\{G_{n}>C\}}| \xrightarrow[]{P_{0}} 0,
\end{equation*}
so for all sufficiently large $n$,
\begin{equation*}
    P_{0}\left(\max_{1 \leq i \leq n}\frac{G_{n}(Z_{i})^{2}}{n} > \eta \right) \leq P_{0} \left(|\mathbb{P}_{n}G_{n}^{2}1_{\{G_{n}>C\}} - P_{0}G_{n}^{2}1_{\{G_{n}>C\}}| > \eta/2\right) \rightarrow 0.
\end{equation*}
We deduce that $\max_{1 \leq i \leq n} G_{n}(Z_{i}) = o_{P_{0}}(\sqrt{n})$ and (\ref{eqn::dbl_bvm2}) converges to zero in probability. The identity $d_{BL}(\mathcal{N}(0, \sigma^{2}), \mathcal{N}(0, \tau^{2})) \leq \sqrt{|\sigma^{2}-\tau^{2}|}$ implies that we can replace $\sigma_{n,g}^{2}$ with $\sigma_{g}^{2}$ in (\ref{eqn::dbl_bvm1}). Finally, we apply the multiplicative part of Lemma \ref{lem::slutsky} with $X_{n} = n^{-1/2} \sum_{i=1}^{n}(W_{i}-1)(g(Z_{i}) - \mathbb{P}_{n}[g])$ and $Y_{n} = 1/\bar{W}_{n}$ and $c=1$ to complete the proof.
\end{proof}

\section{Conditional treatment effects} \label{sec::samp_treat}

In Section \ref{sec::att}, we studied the average treatment effect on the treated (ATT), which is identified by the formula
\begin{equation*} 
    \chi(P) = E_{P}[Y-E_{P}[Y \mid A=0, X] \mid A=1] = \frac{Q[\pi(X)\{\mu^{(1)}(X)-\mu^{(0)}(X)\}]}{Q[\pi(X)]}
\end{equation*}
under the unconfoundedness for controls and weak overlap assumptions. For the purpose of inferring the sample average treatment effect on the treated (SATT) in Section \ref{sec::acic}, we remarked that we could use the ATT as a proxy for the SATT; that is, we use the point and interval estimates from the one-step ATT posterior to infer the SATT. However, this would likely give us substantially conservative interval estimation. As a partial remedy, we will also consider estimates from the one-step posterior for the \textit{average covariate-conditional treatment effect on the treated (ACTT)}
\begin{equation*}
    \theta(P) = \frac{\mathbb{Q}_{n}[\pi(X)\{\mu^{(1)}(X)-\mu^{(0)}(X)\}]}{\mathbb{Q}_{n}[\pi(X)]},
\end{equation*}
which replaces the unknown population covariate distribution $Q$ with the empirical covariate distribution
\begin{equation*}
    \mathbb{Q}_{n} = \frac{1}{n}\sum_{i=1}^{n}\delta_{X_{i}}.
\end{equation*}

Following Chapter 8 of \citet{vanderLaan11}, the efficient influence function for the ATT can be decomposed into
\begin{equation*}
    \dot{\chi}_{P} = \dot{\chi}^{Y}_{P} + \dot{\chi}^{A}_{P} + \dot{\chi}^{X}_{P},
\end{equation*}
where 
\begin{align*}
    \begin{split}
        \dot{\chi}^{Y}_{P} &= \left(\frac{A}{P[A]} - \frac{(1-A)\pi(X)}{P[A]\{1-\pi(X)\}}\right)[Y-\mu^{(A)}(X)] \\
        \dot{\chi}^{A}_{P} &= \frac{A- \pi(X)}{P[A]}[\mu^{(1)}(X) - \mu^{(0)}(X) - \chi(P)] \\
        \dot{\chi}^{X}_{P} &= \frac{ \pi(X)}{P[A]}[\mu^{(1)}(X) - \mu^{(0)}(X) - \chi(P)].
    \end{split}
\end{align*}
Roughly speaking, each component in the decomposition governs a different source of uncertainty in $P$: $\dot{\chi}^{Y}_{P}$---the conditional distribution of $Y$ given $(A,X)$; $\dot{\chi}^{A}_{P}$---the conditional distribution of $A$ given $X$; and $\dot{\chi}^{X}_{P}$---the marginal distribution of $X$.

For the ACTT, the marginal covariate distribution is fixed at $\mathbb{Q}_{n}$, so the idea is to ``switch off'' the Bayesian bootstrap fluctuations in $\dot{\chi}^{X}_{P}$, which will give us a less variable (i.e. tighter) posterior. The one-step corrected parameter for the ACTT is defined by
\begin{equation*}
    \Tilde{\theta} = \theta(P) + \Tilde{P}\left[\left(\frac{A}{\Tilde{P}[A]} - \frac{(1-A)\pi(X)}{\Tilde{P}[A]\{1-\pi(X)\}}\right)\{Y-\mu^{(A)}(X)\}\right]
    + \Tilde{P}\left[\left(\frac{A-\pi(X)}{\Tilde{P}[A]}\right)\{\mu^{(1)}(X) - \mu^{(0)}(X) - \theta(P)\}\right].
\end{equation*}
This is almost the same as $\theta(P) + \Tilde{P}[\dot{\chi}^{Y}_{P} +\dot{\chi}^{A}_{P}]$, except that we have replaced all instances of $P[A]$ with $\Tilde{P}[A]$, just like we did for our ATT methodology, and $\chi(P)$ has been replaced by $\theta(P)$.

For the semiparametric Bernstein-von Mises theorem, we will use the same centring estimator as the ATT:
\begin{equation*}
    \hat{\chi}_{n} = \chi(P_{0}) + \mathbb{P}_{n}[\dot{\chi}_{P_{0}}].
\end{equation*}
We need to first show that this is a valid estimator for the ACTT with the desired asymptotic variance.

\begin{proposition} \label{prop::ACTTlim}
As $n \rightarrow \infty$,
\begin{equation*}
    \sqrt{n}(\hat{\chi}_{n} - \theta(P_{0})) \rightsquigarrow \mathcal{N}(0, \Vert \dot{\chi}^{Y}_{P_0} +\dot{\chi}^{A}_{P_0}\Vert_{P_0}^{2}).
\end{equation*}
\end{proposition}

The proof is postponed to the end of this Section. 

We now give the condition under which the corrected posterior satisfies a Bernstein-von Mises theorem for the ACTT. 

\begin{assumption} \label{ass::ACTT_bvm}
There exists a sequence of measurable subsets $(H_{n})_{n}$ of $\mathcal{P}$ satisfying $\Pi(P \in H_{n} \mid Z^{(n)})  \xrightarrow[]{P_{0}} 1$ such that
\begin{itemize}
    \item[(a)]{($L_{2}$-convergence of $\pi$, $\mu^{(0)}$ and $\mu^{(1)}$)} There exist numbers $\rho_{n},\varepsilon_{n} \rightarrow 0$ and $\sqrt{n}\rho_{n}\varepsilon_{n} \rightarrow 0$ such that
    \begin{align*}
        \sup_{P \in H_{n}} \left\Vert \pi - \pi_{0} \right\Vert_{P_{0}}  \leq \rho_{n}, \quad 
        \sup_{P \in H_{n}} \left\Vert \mu^{(0)} - \mu_{0}^{(0)} \right\Vert_{P_{0}}  \leq \varepsilon_{n}, \quad
        \sup_{P \in H_{n}} \left\Vert \mu^{(1)} - \mu_{0}^{(1)} \right\Vert_{P_{0}}  = o(1).
    \end{align*}
    \item[(b)]{(Donsker class)} The sequences of sets $\{\pi:\,P \in H_{n}\}$ and $\{\mu^{(0)}:\,P \in H_{n}\}$ are both eventually contained in fixed $P_{0}$-Donsker classes.
    \item[(c)]{(Uniform bounding)} For all sufficiently large $n$, there exist fixed constants $\delta, C > 0$ such that for all $P \in H_{n}$,
    \begin{align*}
        \delta < \pi, \pi_{0} &< 1-\delta \\
        |Y - \mu^{(a)}(X)| &< C \qquad \text{for $a=0,1$}
    \end{align*}
    with $P_{0}$-probability 1.
    \item[(d)]{(Glivenko-Cantelli class condition for $\mu^{(1)}$)} 
    \begin{equation*}
        \sup_{P \in H_{n}} |(\mathbb{P}_{n} - P_{0})[\mu^{(1)}(X)]| \xrightarrow[]{P_{0}^{*}} 0.
    \end{equation*}
\end{itemize}
\end{assumption}

\begin{theorem} \label{theo::1step_ACTT}
Under Assumption \ref{ass::ACTT_bvm}, the one-step corrected posterior for the ACTT functional satisfies 
\begin{equation*}
    d_{BL}\left(\mathcal{L}_{\Pi \times \Pi_{BB}}(\sqrt{n}(\Tilde{\theta} - \hat{\chi}_{n}) \mid Z^{(n)}), \mathcal{N}(0, \Vert \dot{\chi}^{Y}_{P_0} +\dot{\chi}^{A}_{P_0} \Vert^{2}_{P_{0}})\right)\xrightarrow[]{P_{0}} 0.
\end{equation*}
\end{theorem}

We highlight the matching asymptotic variances in Proposition \ref{prop::ACTTlim} and Theorem \ref{theo::1step_ACTT}. Thus, we deduce that central credible intervals for $\Tilde{\theta}$ will be approximate confidence intervals for the ACTT as desired. 

The proof of Theorem \ref{theo::1step_ACTT} is deferred to Section \ref{sec::proof_ACTT} because it builds on some components from the proofs of other results. By comparing with Assumption \ref{ass::att_bvm} for the ATT, we see that Assumption \ref{ass::ACTT_bvm} adds some conditions on $\mu^{(1)}$. This is unsurprising, given that the corrected ACTT parameter $\Tilde{\theta}$ involves $\mu^{(1)}$, whereas the corrected ATT parameter does not. Fortunately, the conditions on $\mu^{(1)}$ are relatively mild; the $L_{2}$-convergence in Assumption \ref{ass::ACTT_bvm}(a) does not require a rate, and we only need a Glivenko-Cantelli type condition in Assumption \ref{ass::ACTT_bvm}(d) rather than a Donsker class. In Assumption \ref{ass::ACTT_bvm}(c), we have also added a lower bound for $\pi$ and $\pi_{0}$ (similar to the analogous condition in Section \ref{sec::miss}) for technical purposes in the proof. 

A similar methodology can be established for the \textit{average treatment effect (ATE)}
\begin{equation*}
    \chi(P) = E_{P}[Y^{1}-Y^{0}],
\end{equation*}
which is identified by
\begin{equation} \label{eqn::ate_iden}
    \chi(P) = E_{P}[\mu^{1} - \mu^{0}] = Q[\mu^{1} - \mu^{0}]
\end{equation}
under consistency, \textit{unconfoundedness}
\begin{equation*}
    Y^{a} \independent A \mid X \qquad \text{for $a=0,1$}
\end{equation*}
and \textit{overlap/positivity}:
\begin{equation*}
    0 < P(A=1 \mid X) < 1
\end{equation*}
with $P$-probability 1 \citep{Hernan20}. With the formula in (\ref{eqn::ate_iden}), estimation of the ATE is often approached as two copies of the problem in Section \ref{sec::miss} \citep{Ray20}.

The efficient influence function of the ATE can be decomposed into
\begin{equation*}
    \dot{\chi}_{P} = \dot{\chi}^{Y}_{P} + \dot{\chi}^{A}_{P} + \dot{\chi}^{X}_{P},
\end{equation*}
where 
\begin{align*}
    \begin{split}
        \dot{\chi}^{Y}_{P} &= \left(\frac{A}{\pi(X)} - \frac{(1-A)}{1-\pi(X)}\right)[Y-\mu^{(A)}(X)] \\
        \dot{\chi}^{A}_{P} &= 0 \\
        \dot{\chi}^{X}_{P} &= \mu^{(1)}(X) - \mu^{(0)}(X) - \chi(P).
    \end{split}
\end{align*}
By again replacing $Q$ with $\mathbb{Q}_{n}$, we obtain the \textit{conditional average treatment effect (CATE)} \citep{Imbens04}
\begin{equation*}
    \theta(P) = \mathbb{Q}_{n}[\mu^{1} - \mu^{0}]
\end{equation*}
and the one-step corrected functional for the CATE is
\begin{equation*}
    \Tilde{\theta} = \theta(P) + \Tilde{P}[\dot{\chi}^{Y}_{P}].
\end{equation*}
As before, the CATE is estimated more efficiently than the ATE since we have removed the uncertainty in the marginal covariate distribution, reducing the asymptotic posterior variance from $\Vert \dot{\chi}_{P} \Vert_{P_0}^{2}$ to $\Vert \dot{\chi}^{Y}_{P} \Vert_{P_0}^{2}$. The analysis for the CATE is relatively straightforward so it is omitted; the details can be easily deduced from Theorems \ref{theo::1step_miss} and \ref{theo::1step_ACTT}.

\subsection{Proof of Proposition \ref{prop::ACTTlim}}
It is sufficient to show that 
\begin{equation*}
    \sqrt{n}(\chi(P_{0}) + \mathbb{P}_{n}[\dot{\chi}^{X}_{P_0}] - \theta(P_{0})) = o_{P_0}(1)
\end{equation*}
since this implies that
\begin{align*}
    \sqrt{n}(\hat{\chi}_{n} - \theta(P_{0})) &= \sqrt{n}(\chi(P_{0}) + \mathbb{P}_{n}[\dot{\chi}^{Y}_{P_0} +\dot{\chi}^{A}_{P_0}] + \mathbb{P}_{n}[\dot{\chi}^{X}_{P_0}] - \theta(P_{0}))\\
    &= \sqrt{n}\mathbb{P}_{n}[\dot{\chi}^{Y}_{P_0} +\dot{\chi}^{A}_{P_0}] + o_{P_0}(1)
\end{align*}
and $\sqrt{n}\mathbb{P}_{n}[\dot{\chi}^{Y}_{P_0} +\dot{\chi}^{A}_{P_0}]$ converges to the requisite normal distribution by the central limit theorem.

Some straightforward algebra establishes
\begin{equation*}
    \sqrt{n}(\chi(P_{0}) + \mathbb{P}_{n}[\dot{\chi}^{X}_{P_0}] - \theta(P_{0})) = \left[\frac{\sqrt{n}(P_{0}[A]-\mathbb{Q}_{n}[\pi_{0}(X)])}{P_{0}[A]}\right]\left\{\chi(P_{0}) - \frac{\mathbb{Q}_{n}[\pi_{0}(X)\{\mu_{0}^{(1)}(X)-\mu_{0}^{(0)}(X)\}]}{\mathbb{Q}_{n}[\pi_{0}(X)]}\right\}.
\end{equation*}
Noting that $P_{0}[A] = P_{0}[\pi_{0}(X)]$, the central limit theorem implies that $\sqrt{n}(P_{0}[A]-\mathbb{Q}_{n}[\pi_{0}(X)])$ converges in distribution to a normal distribution with finite variance. By the weak law of large numbers,
\begin{align*}
    \mathbb{Q}_{n}[\pi_{0}(X)\{\mu_{0}^{(1)}(X)-\mu_{0}^{(0)}(X)\}] & \xrightarrow[]{P_{0}} P_{0}[\pi_{0}(X)\{\mu_{0}^{(1)}(X)-\mu_{0}^{(0)}(X)\}] \\
    \mathbb{Q}_{n}[\pi_{0}(X)] &\xrightarrow[]{P_{0}} P_{0}[\pi_{0}(X)].
\end{align*}
Combining these with Slutsky's lemma, we deduce that the term in curly brackets is $o_{P_0}(1)$. Using Slutsky's lemma again, we have that $\sqrt{n}\{\chi(P_{0}) + \mathbb{P}_{n}[\dot{\chi}^{X}_{P_0}] - \theta(P_{0})\}$ converges in distribution to zero. Since zero is a constant, we also have convergence in probability. \qed

\section{Proofs of main results} \label{sec::proofs}

\subsection{Proof of Theorem \ref{theo::1step_bvm}}

By standard arguments, (e.g. p. 142 of \citet{vanderVaart98} or \citet{Castillo15}) it is sufficient to prove the result for the sequence of posterior laws conditioned on $P \in H_{n}$, i.e.
\begin{equation*}
    d_{BL}(\mathcal{L}_{\Pi \times \Pi_{BB}}(\sqrt{n}(\Tilde{\chi} - \hat{\chi}_{n}) \mid Z^{(n)}, P \in H_{n}), \mathcal{N}(0, \Vert \dot{\chi}_{P_{0}} \Vert^{2}_{P_0})) \xrightarrow[]{P_{0}} 0.
\end{equation*} 

We decompose $\{\chi(P) - \chi(P_{0})\}$ as follows:
\begin{equation*}
    \chi(P) - \chi(P_{0}) = -P_{0}[\dot{\chi}_{P}] - r_{2}(P_{0}, P),
\end{equation*}
so we have
\begin{equation*}
    \begin{split}
        \sqrt{n}(\Tilde{\chi} - \hat{\chi}_{n}) &= \sqrt{n}\{\chi(P)-\chi(P_{0}) + \Tilde{P}[\dot{\chi}_{P}] - \mathbb{P}_{n}[\dot{\chi}_{P_{0}}]\} \\
        &= \sqrt{n}\{-P_{0}[\dot{\chi}_{P}] - r_{2}(P_{0}, P)+ \Tilde{P}[\dot{\chi}_{P}] - \mathbb{P}_{n}[\dot{\chi}_{P_{0}}]\}.
    \end{split}
\end{equation*}
Next, we add and subtract $\sqrt{n}\mathbb{P}_{n}[\dot{\chi}_{P}]$:
\begin{equation*}
    \sqrt{n}(\Tilde{\chi} - \hat{\chi}_{n}) = \sqrt{n}\{(\Tilde{P} - \mathbb{P}_{n})[\dot{\chi}_{P}] + (\mathbb{P}_{n} - P_{0})[\dot{\chi}_{P} - \dot{\chi}_{P_{0}}] - r_{2}(P_{0}, P)\}.
\end{equation*}

Let $R_{n}(P) = \mathbb{G}_{n}[\dot{\chi}_{P} - \dot{\chi}_{P_{0}}] - \sqrt{n}r_{2}(P_{0},P)$. Then the decomposition above implies
\begin{align} \label{eqn::decom_ineq}
\begin{split}
    &d_{BL}(\mathcal{L}_{\Pi \times \Pi_{BB}}(\sqrt{n}(\Tilde{\chi} - \hat{\chi}_{n}) \mid Z^{(n)}, P \in H_{n}), \mathcal{N}(0, \Vert \dot{\chi}_{P_{0}} \Vert^{2}_{P_0})) \\
    &= d_{BL}(\mathcal{L}_{\Pi \times \Pi_{BB}}(\sqrt{n}(\Tilde{P} - \mathbb{P}_{n})[\dot{\chi}_{P}] +R_{n}(P) \mid Z^{(n)}, P \in H_{n}), \mathcal{N}(0, \Vert \dot{\chi}_{P_{0}} \Vert^{2}_{P_0})) \\
    &\leq \frac{\sup_{\lVert f \lVert_{BL} \leq 1} \left|\int 1_{P \in H_{n}}\left\{\int f(\sqrt{n}[\Tilde{P}-\mathbb{P}_{n}][\dot{\chi}_{P}])\,d\Pi_{BB}(\Tilde{P} \mid Z^{(n)})-\int f(w)d\mathcal{N}_{(0,\Vert \dot{\chi}_{P_{0}}\Vert_{P_0}^{2})}(w)\right\}d\Pi(P \mid Z^{(n)})\right|}{\int 1_{P \in H_{n}} d\Pi(P \mid Z^{(n)})} \\
    &+ \frac{\sup_{\lVert f \lVert_{BL} \leq 1} \left|\int \int 1_{P \in H_{n}}\left\{f(\sqrt{n}[\Tilde{P}-\mathbb{P}_{n}][\dot{\chi}_{P}] + R_{n}(P))-f(\sqrt{n}[\Tilde{P}-\mathbb{P}_{n}][\dot{\chi}_{P}])\right\}d\Pi_{BB}(\Tilde{P} \mid Z^{(n)})d\Pi(P \mid Z^{(n)})\right|}{\int 1_{P \in H_{n}} d\Pi(P \mid Z^{(n)})} \\
    &\leq d_{BL}(\mathcal{L}_{\Pi \times \Pi_{BB}}(\sqrt{n}(\Tilde{P} - \mathbb{P}_{n})[\dot{\chi}_{P}] \mid Z^{(n)}, P \in H_{n}), \mathcal{N}(0, \Vert \dot{\chi}_{P_{0}} \Vert^{2}_{P_0})) \\
    &+ \frac{ \int 1_{P \in H_{n}}\left|R_{n}(P) \right|\Pi(P \mid Z^{(n)})}{\int 1_{P \in H_{n}} d\Pi(P \mid Z^{(n)})} \\
    &\leq  d_{BL}(\mathcal{L}_{\Pi \times \Pi_{BB}}(\sqrt{n}(\Tilde{P} - \mathbb{P}_{n})[\dot{\chi}_{P}] \mid Z^{(n)}, P \in H_{n}), \mathcal{N}(0, \Vert \dot{\chi}_{P_{0}} \Vert^{2}_{P_0})) \\
    &+ \sup_{P \in H_{n}}|\mathbb{G}_{n}[\dot{\chi}_{P} - \dot{\chi}_{P_{0}}]| + \sup_{P \in H_{n}} |\sqrt{n}r_{2}(P_{0},P)|.
    \end{split}
\end{align}

First consider the empirical process term on the last line of (\ref{eqn::decom_ineq}), which is immediately seen to be $o_{P_0}(1)$ under Assumption \ref{ass::1step_bvm}(c$^{*}$)(i). With Assumption \ref{ass::1step_bvm}(c) instead, let $\mathcal{F}$ be a fixed $P_{0}$-Donsker class such that $\{\dot{\chi}_{P}:\,P\in H_{n}\}$ is contained in $\mathcal{F}$ for all sufficiently large $n$. Let $\rho_{P_0}(f) = \sqrt{P_{0}(f-P_{0}f)^{2}}$ and
\begin{equation*}
    \mathcal{F}_{\delta} = \{f-g:\,f,g\in\mathcal{F},\rho_{P_{0}}(f-g) < \delta\}
\end{equation*}
for any $\delta > 0$. Assumption \ref{ass::1step_bvm}(b) implies that
\begin{align*}
    \delta_{n}=\sup_{P \in H_{n}}\rho_{P_0}(\dot{\chi}_{P} - \dot{\chi}_{P_{0}}) \leq \sup_{P \in H_{n}} \Vert \dot{\chi}_{P} - \dot{\chi}_{P_{0}} \Vert_{P_{0}}  \rightarrow 0.
\end{align*}
Then for all sufficiently large $n$, $\{\dot{\chi}_{P}:\,P\in H_{n}\}$ is contained in $\mathcal{F}_{2\delta_{n}}$, so Corollary 2.3.12 of \citet{vanderVaart96} implies that
\begin{equation*}
    \sup_{P \in H_{n}}|\mathbb{G}_{n}[\dot{\chi}_{P} - \dot{\chi}_{P_{0}}]| \leq \sup_{f \in \mathcal{F}_{2\delta_{n}}}|\mathbb{G}_{n}[f]| \xrightarrow[]{P_{0}^{*}} 0,
\end{equation*}
where $P_{0}^{*}$ denotes convergence in outer probability.

The third term on the last line of (\ref{eqn::decom_ineq})---involving $r_{2}$---tends to zero in probability by Assumption \ref{ass::1step_bvm}(a). This leaves the first term:
\begin{align}
\begin{split}\label{eqn::BB_ineq}
    &d_{BL}(\mathcal{L}_{\Pi \times \Pi_{BB}}(\sqrt{n}(\Tilde{P} - \mathbb{P}_{n})[\dot{\chi}_{P}] \mid Z^{(n)}, P \in H_{n}), \mathcal{N}(0, \Vert \dot{\chi}_{P_{0}} \Vert^{2}_{P_0})) \\
    &= \frac{\sup_{\lVert f \lVert_{BL} \leq 1} \left|\int 1_{P \in H_{n}}\left\{\int f(\sqrt{n}[\Tilde{P}-\mathbb{P}_{n}][\dot{\chi}_{P}])\,d\Pi_{BB}(\Tilde{P} \mid Z^{(n)})-\int f(w)d\mathcal{N}_{(0,\Vert \dot{\chi}_{P_{0}}\Vert_{P_0}^{2})}(w)\right\}d\Pi(P \mid Z^{(n)})\right|}{\int 1_{P \in H_{n}} d\Pi(P \mid Z^{(n)})}\\
    &\leq \frac{\sup_{\lVert f \lVert_{BL} \leq 1} \left|\int 1_{P \in H_{n}}\left\{\int f(\sqrt{n}[\Tilde{P}-\mathbb{P}_{n}][\dot{\chi}_{P}])\,d\Pi_{BB}(\Tilde{P} \mid Z^{(n)})-\int f(v)d\mathcal{N}_{(0,\Vert \dot{\chi}_{P}-P_{0}[\dot{\chi}_{P}]\Vert_{P_0}^{2})}(v)\right\}d\Pi(P \mid Z^{(n)})\right|}{\int 1_{P \in H_{n}} d\Pi(P \mid Z^{(n)})}\\
    &+ \frac{\sup_{\lVert f \lVert_{BL} \leq 1} \left|\int 1_{P \in H_{n}}\left\{\int f(v)d\mathcal{N}_{(0,\Vert \dot{\chi}_{P}-P_{0}[\dot{\chi}_{P}]\Vert_{P_0}^{2})}(v)-\int f(w)d\mathcal{N}_{(0,\Vert \dot{\chi}_{P_{0}}\Vert_{P_0}^{2})}(w)\right\}d\Pi(P \mid Z^{(n)})\right|}{\int 1_{P \in H_{n}} d\Pi(P \mid Z^{(n)})}\\
    &\leq \sup_{\lVert f \lVert_{BL} \leq 1} \sup_{P \in H_{n}} \left| \int f(\sqrt{n}[\Tilde{P}-\mathbb{P}_{n}][\dot{\chi}_{P}])\,d\Pi_{BB}(\Tilde{P} \mid Z^{(n)})-\int f(v)d\mathcal{N}_{(0,\Vert \dot{\chi}_{P}-P_{0}[\dot{\chi}_{P}]\Vert_{P_0}^{2})}(v)\right| \\
    &+ \sup_{\lVert f \lVert_{BL} \leq 1} \sup_{P \in H_{n}} \left| \int f(v)d\mathcal{N}_{(0,\Vert \dot{\chi}_{P}-P_{0}[\dot{\chi}_{P}]\Vert_{P_0}^{2})}(v) - \int f(w)d\mathcal{N}_{(0,\Vert \dot{\chi}_{P_{0}}\Vert_{P_0}^{2})}(w)\right| \\
    &=  \sup_{P \in H_{n}} d_{BL}(\mathcal{L}_{\Pi_{BB}}(\sqrt{n}(\Tilde{P} - \mathbb{P}_{n})[\dot{\chi}_{P}] \mid Z^{(n)}), \mathcal{N}(0, \Vert \dot{\chi}_{P} - P_{0}[\dot{\chi}_{P}] \Vert^{2}_{P_0})) \\
    &+ \sup_{P \in H_{n}} d_{BL}(\mathcal{N}(0, \Vert \dot{\chi}_{P} - P_{0}[\dot{\chi}_{P}] \Vert^{2}_{P_0}),\mathcal{N}(0, \Vert \dot{\chi}_{P_{0}} \Vert^{2}_{P_0})).
    \end{split}
\end{align}
By Assumption \ref{ass::1step_bvm}(c), we can apply Theorem 3.6.13 of \citet{vanderVaart96} to show that the first term on the last line of (\ref{eqn::BB_ineq}) tends to zero in outer probability. Otherwise, we apply Lemma \ref{lem::bvm_bb} under Assumption \ref{ass::1step_bvm}(c$^{*}$). The second term on the last line of (\ref{eqn::BB_ineq}) tends to zero via the identity $d_{BL}(\mathcal{N}(0, \sigma^{2}), \mathcal{N}(0, \tau^{2})) \leq \sqrt{|\sigma^{2}-\tau^{2}|}$ and the assumed uniform $L_{2}$-convergence in Assumption \ref{ass::1step_bvm}(b), which completes the proof.

If the efficient influence function takes the form
\begin{equation*}
    \dot{\chi}_{P} = \psi_{P} - P[\psi_{P}],
\end{equation*}
we point out that
\begin{align*}
    \sqrt{n}(\Tilde{P} - \mathbb{P}_{n})[\dot{\chi}_{P}] &=  \sqrt{n}(\Tilde{P} - \mathbb{P}_{n})[\psi_{P}]\\
    \mathbb{G}_{n}[\dot{\chi}_{P} - \dot{\chi}_{P_{0}}] &= \mathbb{G}_{n}[\psi_{P} - \psi_{P_{0}}].
\end{align*}
The proof can then be easily adapted when $\dot{\chi}_{P}$ and $\dot{\chi}_{P_{0}}$ are replaced by $\psi_{P}$ and $\psi_{P_{0}}$ respectively in Assumptions \ref{ass::1step_bvm}(b), (c) and (c*). \qed

\subsection{Proof of Theorem \ref{th:f2}} \label{proof::dens_squared}

From \citet{Shen13}, let $\tilde \epsilon_n = M_n n^{-\alpha/(2\alpha+1)}$ with $M_n \geq (\log n)^r$ and 
$$\mathcal F_n = \left\{f (x)= \int \varphi_\sigma(x-\mu) dQ(\mu), Q \in \mathcal Q_n, \sigma \in (\sigma_n, S_n)\right\}\cap \{d_H(f_0, f) \leq \tilde \epsilon_n \}, \quad $$
for some $r>0$ with
\begin{equation*}
    \mathcal Q_n = \left\{ \sum_j p_j \delta_{(\mu_j)}, \sum_{j>J_n}p_j \leq \sigma_n n^{-2}, \max_{j\leq J_n}|\mu_j |\leq n^B\right\}, \quad J_n = J_0n \tilde{\epsilon_n}^2, \quad B,J_0>0
\end{equation*}
and $\sigma_n = \sigma_0 \tilde \epsilon_n^{1/\alpha}$ and $S_n = e^{s_0 n\tilde \epsilon_n^2} $ for some $s_0>0$. Then from \citet{Shen13} Proposition 2 and Theorem 1, 
$$ \Pi( \mathcal Q_n^c |Y^n) = o_{P_0}(1), \quad \Pi( d_H( f_0, f) > \tilde \epsilon_n |Y^n ) = o_{P_0}(1).$$
Thus if $H_n = \mathcal F_n \cap \{ \|f-f_0\|_2 > M_0\tilde \epsilon_{n}\}$, $ \Pi( H_n^c |Y^n) = o_{P_0}(1)$.
We write $\mathcal G_n =\{g= f_{J_n}-f_0, f \in H_n \}$,  where we denote for $f (x)= \sum_{j =1}^{\infty} p_j \varphi_\sigma(x-\mu_j)$, $f_{J_n} (x)= \sum_{j =1}^{J_n} p_j \varphi_\sigma(x-\mu_j)$.  To prove assumption \ref{ass::1step_bvm} we first prove that on $H_n$ the densities are bounded and that $\|g\|_2 = O(\tilde \epsilon_n)$. 
First we show that 
\begin{equation}\label{supnorm:1}
\Pi ( \|f\|_\infty > 4 \|f_0\|_\infty | Y^n) = o_{p_0}(1)
\end{equation}
First, on the Hellinger ball $d_H(f, f_0) \leq \tilde \epsilon_n$, we have 
$$ \int \mathbf 1_{f> 4 \|f_0\|_\infty } f(x)dx \leq 4 d_H(f_0, f)^2 \leq 4 \tilde \epsilon_n^2,$$
which in terms implies that $$\text{Leb}( \{ f\geq 4 \|f_0\|_\infty \} ) \leq  \frac{  \tilde \epsilon_n^2 }{\|f_0\|_\infty} .$$
Let $x_1\in A_n =\{ x; f(x)> 8 \|f_0\|_\infty\}$, noting that for all $f \in \mathcal F_n $, 
$ |f^{'}(y)| \lesssim \sigma^{-2} \lesssim \sigma_n^{-2},$
we have that 
$$ f(x) \geq f(x_1)- |x-x_1| C/\sigma_n^{2}\geq 4\|f_0\|_\infty$$
as soon as $|x-x_1|\leq 2\sigma_n^{2}/C \lesssim \tilde \epsilon_n^{2/\alpha} >>\tilde \epsilon_n^2$ when $\alpha>1$,  which contradicts $\text{Leb}(A_n)\lesssim \tilde{\epsilon_n}^2$. 

Interestingly \eqref{supnorm:1} allows us to refine the $L_2$ rate obtained by \cite{Scricciolo14} since
$$ \|f-f_0\|_2^2 \leq 2d_H(f,f_0)^2 (\|f\|_\infty +\|f_0\|_\infty) \lesssim \tilde \epsilon_n^2.$$
So that we obtain 
$$ \Pi( \|f-f_0\|_2 > M_0\tilde \epsilon_n |Y^n ) = o_{P_0}(1). $$
while the rate obtained in \cite{Scricciolo14} was $\tilde \epsilon_n(n \tilde \epsilon_n^2)^{1/4}$. 

Hereafter we write $\dot{\chi}_f$ in place of $\dot{\chi}_P$. 
 Then,  when $\alpha>1/2$ and uniformly over $H_n$, 
\begin{equation*}
\sqrt{n} r_2(f_0, f) = \chi(f_0) - \chi(f) - P_0(\dot{\chi}_f) = \sqrt{n}\|f-f_0\|_2^2 = o(1), 
    \end{equation*}
and 
\begin{equation*}
\|\dot{\chi}_f -\dot{\chi}_{f_0}\|_{f_0}= 4\|f-f_0\|_2^2 = o(1).
    \end{equation*} 
 We now prove Assumption \ref{ass::1step_bvm} $(c^*)$, using  Lemma 19.36 of \cite{vanderVaart98}. For all $g \in \mathcal G_n$
 $$F_0(g^2) \leq 2\|f_0\|_\infty \tilde\epsilon_{n}^2 + 2F_0((f-f_{J_n})^2)\leq 2\|f_0\|_\infty\tilde\epsilon_{n}^2 + 2 n^{-4}/\sigma_n^2\leq 3\|f_0\|_\infty\tilde\epsilon_{n}^2 := \delta_n^2,$$
 for $n$ large enough 
 and $\|g\|_\infty \leq 1/\sigma_n$. We thus only need to upper bound 
 $ J_{[],n} := J_{[]}(\delta_n, \mathcal G_n, L_2(f_0)) = o(  n^{1/4}\tilde \epsilon_n \wedge 1 ) $ so that using $\|g\|_\infty \leq 2\|f_0\|_\infty$ 
 $$\mathbb E_0^* \sup_{g \in \mathcal G_n} \mathbb G_n(g) \lesssim  J_{[],n} + \frac{ J_{[],n}^2}{ \delta_n^2 \sqrt{n} } =o(1)$$
 
 In Lemma \ref{lem:J} we show that 
 $$J_{[],n} \lesssim \sqrt{J_n}\tilde \epsilon_{n} \log n \lesssim  \sqrt{n}\tilde \epsilon_n^2 \log n \asymp n^{-(\alpha -1/2)/(2 \alpha+1)} (\log n)^{r_1} $$
 for some $r_1$, so that as soon as $\alpha >1/2$, $J_{[],n} = =o(1)$ and $J_{[],n}^2/\tilde \epsilon_n^2 \lesssim n \tilde \epsilon_n^2 \log n = o(\sqrt{n}) $ for $\alpha>1/2$. Therefore, 
 $$ \sup_{g \in \mathcal G_n} |\mathbb G_n(g)| = o_{p_0}(1).$$
 Finally on $\mathcal F_n$ $\|f\|_\infty \leq 4 \|f_0\|_\infty$, therefore an envelop function of $\mathcal F_n $
 is the constant $4\|f_0\|_\infty$ and Assumption  \ref{ass:onestep:post} (ii) is verified which terminates the proof of  Theorem \ref{th:f2}. \qed

\begin{lemma} \label{lem:J}
 Using the definitions above 
 $$J_{[],n} =O(\sqrt{J_n}\tilde \epsilon_{n}\log n ).$$
 \end{lemma}

 \begin{proof}[Proof of Lemma \ref{lem:J}]
     Recall that, writing $\delta_n =\|f_0\|_\infty \epsilon_{n,2} $ 
     $$ J_{[],n} = \int_0^{\delta_n} \sqrt{ \log N_{[]}(u, \mathcal G_n, L_2(f_0))}du$$
     where $N_{[]}(u, \mathcal G_n, L_2(f_0))$ is the bracketing covering number . 
Let $\zeta>0$,  let $\mathcal S_\mu = (\hat \mu_\ell, \ell \leq L_1(u))$ be a net of radius $\zeta u \sigma_n^2$ of $[-a_n,a_n]$ , $\mathcal S_{\sigma} = \hat \sigma_\ell = \sigma_n ( 1 + \zeta u^2)^\ell $, $ \ell \leq L_2(u) \lesssim \frac{ \log (S_n/\sigma_n) }{ u^2 }$ and $\mathcal S_p= (\hat p^{(\ell)}, \ell \leq L_3(u)) $  a net of $[0,1]$ of radius $\sigma_n u /J_n$. 
Consider the set of functions 
$$ \hat g^U (x) = \sum_{j=1}^{J_n} (\hat p_j+ \sigma_n u^2/J_n)(1 + u^2)\varphi_{\hat \sigma/(1 + u^2)} (x - \hat \mu_j)^2 e^{  |x-\hat \mu_j|\zeta u} -f_0$$
$$ \hat g^L (x)= \sum_{j=1}^{J_n} (\hat p_j- \sigma_n u^2/J_n)_+(1 + u^2)\varphi_{\hat \sigma/(1-u^2)} (x - \hat \mu_j)^2 e^{ - |x-\hat \mu_j|\zeta u} -f_0$$
Let $g\in \mathcal G_n$, we show that there exists a pair $\hat g^U, \hat g^L$ as above such that $\hat g^U\geq g \geq \hat g^L$. Then since the number of such pairs is bounded by 
$ L_1(u)^{J_n}L_2(u)L_3(u)^{J_n}$, 
\begin{align*}
     \log N_{[]}(u, \mathcal G_n, L_2(f_0)) &\leq J_n [\log L_1(u) + \log L_3(u)] + \log L_2(u) \\
     &\lesssim J_n ( \log (a_n/\sigma_n) +\log (J_n/\sigma_n) + \log (1/u)) + \log \log (S_n/\sigma_n) + \log (1/u).
\end{align*} 
By definition $\log (a_n/\sigma_n) +\log (J_n/\sigma_n) +\log \log (S_n/\sigma_n) \lesssim \log n$ so that 
$$ J_{[],n} =\sqrt{J_n} \sqrt{ \log n}  \int_0^{\delta_n} \sqrt{ \log(1/ u)}du \lesssim \sqrt{J_n} \delta_n \log n.$$
 \end{proof}

\subsection{Proof of Theorem \ref{theo::1step_miss}}

Let
\begin{equation*}
    \psi_{P} = \chi(P) + \dot{\chi}_{P} = \frac{A}{\pi(X)}\{Y - m(X)\} + m(X),
\end{equation*}
such that $\dot{\chi}_{P} = \psi_{P} - P[\psi_{P}]$. As we have remarked in the main paper, it is sufficient to verify Assumption \ref{ass::1step_bvm} with $\dot{\chi}_{P}$ and $\dot{\chi}_{P_{0}}$ replaced by $\psi_{P}$ and $\psi_{P_{0}}$ respectively for (b) and (c). Assumption \ref{ass::1step_bvm}(a) is verified by the inequality in (\ref{eqn::miss_2nd}) and Assumption \ref{ass::miss_bvm}(a).

Note that
\begin{equation*}
    \psi_{P} - \psi_{P_{0}} = \left(1 - \frac{A}{\pi(X)}\right)[m(X)-m_{0}(X)] +A(Y-m(X))\left[\frac{1}{\pi(X)} - \frac{1}{\pi_{0}(X)}\right],
\end{equation*}
so by Assumption \ref{ass::miss_bvm}(b), we have
\begin{equation*}
    \Vert \psi_{P} - \psi_{P_{0}} \Vert_{P_{0}} \leq \left(1+ \frac{1}{\delta}\right)\Vert m - m_{0} \Vert_{P_0} + C\left\Vert \frac{1}{\pi} - \frac{1}{\pi_{0}} \right\Vert_{P_0}.
\end{equation*}
We deduce by Assumption \ref{ass::miss_bvm}(a) that
\begin{equation} \label{eqn::miss_l2}
    \sup_{P \in H_{n}} \Vert \psi_{P} - \psi_{P_{0}} \Vert_{P_{0}}  \rightarrow 0.
\end{equation}
With the uniform bounding in Assumption \ref{ass::miss_bvm}(b) and the permanence of the Donsker property under Lipschitz transformations (see examples 2.10.7, 2.10.8 and 2.10.9 in \citet{vanderVaart96}), the classes $\{\psi_{P}: P \in H_{n}\}$ eventually lie in a fixed $P_{0}$-Donsker class. \qed

\subsection{Proof of Theorem \ref{theo::1step_ray}}

To aid comparisons with \citet{Ray20}, we use the following notation:
\begin{align*}
    a &= \frac{1}{\pi}, \quad a_{0} = \frac{1}{\pi_{0}}, \quad \hat{a}_{n} = \frac{1}{\hat{\pi}_{n}}\\
    b  &= m, \quad b_{0} = m_{0}.
\end{align*}

Similar to the argument in the proof of Theorem 2 of \citet{Ray20}, it suffices to fix $\hat{a}_{n}$ to be a deterministic sequence $a_{n} \in \mathcal{A}_{n}$, where $\mathcal{A}_{n} = \{a:\,\Vert a \Vert_{\infty} \leq M, \Vert a-a_{0} \Vert_{P_0} \leq M\rho_{n} \}$ for any $M > 0$. This is due to the assumed stochastic independence of the data and $\hat{a}_{n}$, as well as the assumptions that $\Vert \hat{a}_{n}\Vert_{\infty} = O_{P_0}(1)$ and $\Vert \hat{a}_{n} - a_{0}\Vert_{P_{0}} = O_{P_0}(\rho_{n})$. Like the proof of Theorem \ref{theo::1step_miss}, we set our sequence of efficient estimators to be 
\begin{equation*}
    \hat{\chi}_{n} = \mathbb{P}_{n}[Aa_{0}(X)(Y-b_{0}(X))+b_{0}(X)].
\end{equation*}
Then we have
\begin{align*}
    \sqrt{n}(\Tilde{\chi}-\hat{\chi}_{n})=&\,\sqrt{n}\Tilde{P}[Aa_{n}(X)(Y-b(X))+b(X)]-\sqrt{n}\mathbb{P}_{n}[Aa_{0}(X)(Y-b_{0}(X))+b_{0}(X)]\\ =&\underbrace{\sqrt{n}(\Tilde{P}-\mathbb{P}_{n})[Aa_{n}(X)(Y-b(X))+b(X)]}_{\circled{1}}\\
    &+\underbrace{\mathbb{G}_{n}[Aa_{n}(X)(Y-b(X))+b(X) - Aa_{0}(X)(Y-b_{0}(X))-b_{0}(X)]}_{\circled{2}}\\
    &+\underbrace{\sqrt{n}P_{0}[Aa_{n}(X)(Y-b(X))+b(X) - Aa_{0}(X)(Y-b_{0}(X))-b_{0}(X)]}_{\circled{3}}
\end{align*}

Term \circled{3} is just the second-order bias
\begin{align*}
    \circled{3} &= \sqrt{n}P_{0}\left[\frac{1}{a_{0}(X)}\left(a_{0}(X) - a_{n}(X)\right)\{b_{0}(X) - b(X)\}\right] \\
    &\leq \sqrt{n}\left\Vert\frac{1}{a_{0}}\right\Vert_{\infty}\left\Vert a_{0} - a_{n} \right\Vert_{P_0} \Vert b_{0} - b\Vert_{P_0},
\end{align*}
and we can upper bound this by $\sqrt{n}M\rho_{n}\varepsilon_{n}\rightarrow 0$ for $b$ in $\{b = \Psi(w): w \in \Tilde{\mathcal{H}}_{n}\}$ by (\ref{ass::1step_3.11}).
For term \circled{2}, we have the inequality
\begin{align*}
    |\circled{2}| \leq &\, |\mathbb{G}_{n}[Aa_{n}(X)(Y-b(X)) - Aa_{0}(X)(Y-b_{0}(X))]| + |\mathbb{G}_{n}[b(X)-b_{0}(X)]|.
\end{align*}
The second term on the right converges uniformly in probability to zero on $\{b = \Psi(w): w \in \Tilde{\mathcal{H}}_{n}\}$ by (\ref{ass::1step_3.12}). We can decompose the first term on the right as follows:
\begin{align*}
    |\mathbb{G}_{n}[Aa_{n}(X)(Y-b(X)) - Aa_{0}(X)(Y-b_{0}(X))]| \leq &\, \underbrace{|\mathbb{G}_{n}[Aa_{n}(X)\{b_{0}(X)-b(X)\}]|}_{\circled{A}} \\
    &+ \underbrace{|\mathbb{G}_{n}[A\{a_{n}(X)-a_{0}(X)\}\{Y-b_{0}(X)\}]|}_{\circled{B}}
\end{align*}
We need to show that each term converges uniformly in probability to zero on $\{b = \Psi(w): w \in \Tilde{\mathcal{H}}_{n}\}$. For term \circled{A}, this is achieved by applying Lemmas 10 and 9 of \citet{Ray20} with $\varphi = Aa_{n}$ and $h = b-b_{0}$ with the assumptions of (\ref{ass::1step_3.11}) and (\ref{ass::1step_3.12}), along with the uniform bounding for $a_{n}$. Since $a_{n}$ is a deterministic sequence that converges in $L_{2}(P_{0})$ to $a_{0}$, term \circled{B} converges to zero in probability.

It remains to establish the convergence of term \circled{1} to the requisite normal distribution. Having shown that term \circled{2} converges uniformly in probability to zero on $\{b = \Psi(w): w \in \Tilde{\mathcal{H}}_{n}\}$, we already have the requisite Glivenko-Cantelli condition to apply Lemma \ref{lem::bvm_bb} (the bounding on the envelope functions is immediately implied by $\Vert a_{n} \Vert_{\infty} \leq M$). Thus, we just need to show that $[Aa_{n}(X)(Y-b(X))+b(X)]$ converges uniformly in $L_{2}$ to $[Aa_{0}(X)(Y-b_{0}(X))+b_{0}(X)]$. Following the steps in the proof of Theorem \ref{theo::1step_miss}, we have
\begin{equation*}
    \Vert Aa_{n}(X)(Y-b(X))+b(X) - Aa_{0}(X)(Y-b_{0}(X))-b_{0}(X)\Vert_{P_0} \leq \left(1+ M\right)\Vert b - b_{0} \Vert_{P_0} + \left\Vert a_{n} - a_{0} \right\Vert_{P_0},
\end{equation*}
so we are done after applying (\ref{ass::1step_3.11}). \qed

\subsection{Proof of Corollary \ref{cor::riemann}}
\label{sec::proof_cor_riemann}

We prove this corollary by applying several of the results in \citet{Ray20} to check the conditions in Theorem \ref{theo::1step_ray}. First, by applying Lemma 5 of \citet{Ray20} with $t=0$, there exist measurable subsets $\mathcal{H}_{n} \subset C([0,1])$ (the set of continuous functions on $[0,1]$) with
\begin{align*}
    \Pi(w \in \mathcal{H}_{n} \mid  Z^{(n)})&\xrightarrow[]{P_{0}} 1\\
    \sup_{m = \Psi(w): w \in \mathcal{H}_{n}}
    |\mathbb{G}_{n}[m-m_{0}]|&\xrightarrow[]{P_{0}} 0.
\end{align*}

Next, Lemma 16 and the proof of Corollary 1 for the Riemann-Liouville prior from \citet{Ray20} implies that $m$ contracts to $m_{0}$ at the rate
\begin{equation*}
    \varepsilon_{n} = n^{-\frac{\beta \wedge \bar{\beta}}{2\bar{\beta} + 1}}(\log n)^{\kappa},
\end{equation*}
where $\kappa$ either takes the value 0 or 1 depending on the values of $(\beta, \bar{\beta})$ (see Theorem 4 of \citet{Castillo08} or the discussion following Lemma 11.34 in \citet{Ghosal17}). In other words,
\begin{equation*}
    \Pi(\Vert m - m_{0} \Vert_{P_0} \leq M\varepsilon_{n} \mid  Z^{(n)})\xrightarrow[]{P_{0}} 1
\end{equation*}
for some sufficiently large constant $M > 0$ (see also: Theorem 11.22 of \citet{Ghosal17}). Setting $\kappa$ to be identically equal to one also suffices. (This was done in the statement of the corollary for the sake of simplicity.) So the sequence of subsets 
\begin{equation*}
    \Tilde{\mathcal{H}}_{n} = \mathcal{H}_{n} \cap \{w \in C([0,1]): \Vert \Psi(w) - m_{0} \Vert_{P_0} \leq M\varepsilon_{n}\}
\end{equation*}
satisfies the conditions for Theorem \ref{theo::1step_ray}. \qed

\subsection{Proof of Theorem \ref{theo::1step_att}} \label{sec::proof_att}

Our sequence of efficient estimators is defined by
\begin{equation*}
    \hat{\chi}_{n} = \frac{\mathbb{P}_{n}\left[\left(\frac{A-\pi_{0}(X)}{1-\pi_{0}(X)}\right)\{Y-\mu_{0}^{(0)}(X)\}\right]}{\mathbb{P}_{n}[A]}.
\end{equation*}
We begin by showing that this sequence does indeed achieve the requisite limiting distribution.
\begin{align*}
    \sqrt{n}(\hat{\chi}_{n} - \chi(P_{0})) &= \frac{\sqrt{n}\mathbb{P}_{n}\left[\left(\frac{A-\pi_{0}(X)}{1-\pi_{0}(X)}\right)\{Y-\mu_{0}^{(0)}(X)\} - A\chi(P_{0})\right]}{\mathbb{P}_{n}[A]} \\
    &= \sqrt{n}\mathbb{P}_{n}[\dot{\chi}_{P_{0}}]\frac{P_{0}[A]}{\mathbb{P}_{n}[A]}\\
    &\rightsquigarrow \mathcal{N}(0, \Vert \dot{\chi}_{P_{0}} \Vert^{2}_{P_{0}}),
\end{align*}
since the ratio $P_{0}[A]/\mathbb{P}_{n}[A]$ converges in probability to one by the weak law of large numbers and Slutsky's lemma. 

Now define
\begin{align*}
    h(Z) &= \left(\frac{A-\pi(X)}{1-\pi(X)}\right)\{Y-\mu^{(0)}(X)\}-A\chi(P_{0}) \\
    h_{0}(Z) &= \left(\frac{A-\pi_{0}(X)}{1-\pi_{0}(X)}\right)\{Y-\mu_{0}^{(0)}(X)\}-A\chi(P_{0}) = P_{0}[A]\dot{\chi}_{P_{0}}(Z).
\end{align*}
With the uniform bounding in Assumption \ref{ass::att_bvm}(b) and the permanence of the Donsker property under Lipschitz transformations (see examples 2.10.7, 2.10.8 and 2.10.9 in \citet{vanderVaart96}), the classes $\{h: P \in H_{n}\}$ eventually lie in a fixed $P_{0}$-Donsker class. We have
\begin{align*}
    \sqrt{n}(\Tilde{\chi}-\hat{\chi}_{n}) &=\sqrt{n}\{\Tilde{\chi}-\chi(P_{0}) -\hat{\chi}_{n}+ \chi(P_{0})\}\\
    &=\frac{\sqrt{n}\Tilde{P}\left[\left(\frac{A-\pi(X)}{1-\pi(X)}\right)\{Y-\mu^{(0)}(X)\}-A\chi(P_{0})\right]}{\Tilde{P}[A]} - \frac{\sqrt{n}\mathbb{P}_{n}\left[\left(\frac{A-\pi_{0}(X)}{1-\pi_{0}(X)}\right)\{Y-\mu_{0}^{(0)}(X)\}-A\chi(P_{0})\right]}{\mathbb{P}_{n}[A]}\\
    &= \frac{\sqrt{n}(\Tilde{P} - \mathbb{P}_{n})[h(Z)]+\sqrt{n}\mathbb{P}_{n}[h(Z)]}{\Tilde{P}[A]}-\frac{\sqrt{n}\mathbb{P}_{n}[h_{0}(Z)]}{\mathbb{P}_{n}[A]}\\
    &= \underbrace{\frac{\sqrt{n}(\Tilde{P}-\mathbb{P}_{n})[h(Z)]}{P_{0}[A]} + \frac{\sqrt{n}(\mathbb{P}_{n}-P_{0})[h(Z)-h_{0}(Z)]}{\mathbb{P}_{n}[A]} + \frac{\sqrt{n}P_{0}[h(Z)]}{\mathbb{P}_{n}[A]}}_{\circled{1}}\\
    & + \underbrace{\left(\frac{1}{\Tilde{P}[A]} - \frac{1}{P_{0}[A]}\right)\sqrt{n}(\Tilde{P} - \mathbb{P}_{n})[h(Z)] + \left(\frac{1}{\Tilde{P}[A]} - \frac{1}{\mathbb{P}_{n}[A]}\right)\sqrt{n}\mathbb{P}_{n}[h(Z)]. }_{\circled{2}}
\end{align*}
We will show that term $\circled{1}$ converges conditionally to the requisite $\mathcal{N}(0, \Vert \dot{\chi}_{P_{0}} \Vert^{2}_{P_{0}})$ distribution, while term $\circled{2}$ converges conditionally to zero. Combining these two results with Lemma \ref{lem::slutsky} completes the proof.

As before, it is sufficient to prove the result conditioned on $P \in H_{n}$. For term $\circled{1}$, we have
\begin{align*}
    d_{BL}\left(\mathcal{L}\left(\circled{1} \mid P \in H_{n}, Z^{(n)}\right), \mathcal{N}(0, \Vert \dot{\chi}_{P_{0}} \Vert^{2}_{P_{0}})\right) &\leq d_{BL}\left(\mathcal{L}\left(\frac{\sqrt{n}(\Tilde{P}-\mathbb{P}_{n})[h(Z)]}{P_{0}[A]} \mid P \in H_{n}, Z^{(n)}\right), \mathcal{N}(0, \Vert \dot{\chi}_{P_{0}} \Vert^{2}_{P_{0}})\right)\\
    &+ \sup_{P \in H_{n}}\left|\frac{\mathbb{G}_{n}[h(Z)-h_{0}(Z)]}{\mathbb{P}_{n}[A]}\right| + \sup_{P \in H_{n}}\left|\frac{\sqrt{n}P_{0}[h(Z)]}{\mathbb{P}_{n}[A]}\right|.
\end{align*}
Provided that 
\begin{equation*}
    \sup_{P \in H_{n}}\Vert h - h_{0} \Vert_{P_0} \rightarrow 0,
\end{equation*}
the first two terms on the right can be handled in the same way as Theorem \ref{theo::1step_bvm}. Namely, the first term uses Theorem 3.6.13 of \citet{vanderVaart96}, and the second term uses Corollary 2.3.12 of \citet{vanderVaart96}. To establish the uniform $L_{2}$-convergence of $h$, consider
\begin{align*}
    h(Z) - h_{0}(Z) =& \left[\frac{A-\pi(X)}{1-\pi(X)}\right]\{Y-\mu^{(0)}(X)\} - \left[\frac{A-\pi_{0}(X)}{1-\pi_{0}(X)}\right]\{Y-\mu_{0}^{(0)}(X)\}\\
    =& \left(1-\frac{1-A}{1-\pi(X)}\right)\{\mu_{0}^{(0)}(X)-\mu^{(0)}(X)\} + (1-A)\left[\frac{1}{1-\pi_{0}(X)}-\frac{1}{1-\pi(X)}\right]\{Y - \mu_{0}^{(0)}(X)\} \\
    =& \left(1-\frac{1-A}{1-\pi(X)}\right)\{\mu_{0}^{(0)}(X)-\mu^{(0)}(X)\} + (1-A)\left[\frac{\pi_{0}(X)-\pi(X)}{(1-\pi_{0}(X))(1-\pi(X))}\right]\{Y - \mu_{0}^{(0)}(X)\}
\end{align*}
so
\begin{equation*}
    \Vert h - h_{0} \Vert_{P_0} \leq \left(1 + \frac{1}{\delta}\right)\Vert \mu_{0}^{(0)} - \mu^{(0)}\Vert_{P_0} + \left(\frac{C}{\delta^{2}}\right)\Vert \pi - \pi_{0} \Vert_{P_0},
\end{equation*}
which establishes the uniform $L_{2}$-convergence of $h$ under our assumptions.

For the final term in our upper bound for $\circled{1}$, we have
\begin{align*}
    P_{0}[h(Z)] =& P_{0}\left[\left(\frac{A-\pi(X)}{1-\pi(X)}\right)\{Y-\mu^{(0)}(X)\}-\pi_{0}(X)\{\mu_{0}^{(1)}(X)-\mu_{0}^{(0)}(X)\}\right]\\
    =& P_{0}\left[\frac{\pi_{0}(X)\mu_{0}^{(1)}(X) - \pi_{0}(X)\mu^{(0)}(X) - \pi(X)[\pi_{0}(X)\mu_{0}^{(1)}(X) + (1-\pi_{0}(X))\mu_{0}^{(0)}(X)] + \pi(X)\mu^{(0)}(X)}{1-\pi(X)}\right]\\
    &-P_{0}\left[\pi_{0}(X)\{\mu_{0}^{(1)}(X)-\mu_{0}^{(0)}(X)\}\right]\\
    =& P_{0}\left[\frac{ - \pi_{0}(X)\mu^{(0)}(X) - \pi(X) (1-\pi_{0}(X))\mu_{0}^{(0)}(X) + \pi(X)\mu^{(0)}(X) + \pi_{0}(X)\mu_{0}^{(0)}(X)(1-\pi(X))}{1-\pi(X)}\right]\\
    =& P_{0}\left[\frac{ - \pi_{0}(X)\mu^{(0)}(X) - \pi(X)\mu_{0}^{(0)}(X) + \pi(X)\mu^{(0)}(X) + \pi_{0}(X)\mu_{0}^{(0)}(X))}{1-\pi(X)}\right]\\
    =& P_{0}\left[\frac{(\pi(X)-\pi_{0}(X))(\mu^{(0)}(X)-\mu_{0}^{(0)}(X))}{1-\pi(X)}\right]\\
    \leq & \frac{1}{\delta}\Vert \pi - \pi_{0} \Vert_{P_0} \Vert \mu_{0}^{(0)} - \mu^{(0)}\Vert_{P_0}
\end{align*}
where the inequality on the final line is due to Cauchy-Schwarz and the assumed bounding on $\pi$. By assumption, the product of the errors converge uniformly faster than $1/\sqrt{n}$.

Now we proceed to study term $\circled{2}$. First we note that $\Tilde{P}[A]$ converges conditionally in distribution to $P_{0}[A]$ almost surely under the infinite product probability measure $P_{0}^{(\infty)}$. For an outcome in the infinite product sample space such that $\Tilde{P}[A] \rightsquigarrow P_{0}[A]$, we can apply Slutsky's lemma to deduce that
\begin{equation*}
    \left(\frac{P_{0}[A]}{\Tilde{P}[A]}-1\right) \rightsquigarrow 0.
\end{equation*}
Therefore, this also holds conditionally almost surely $[P_{0}^{(\infty)}]$. We already know that 
\begin{equation*}
    d_{BL}\left(\mathcal{L}\left(\frac{\sqrt{n}(\Tilde{P}-\mathbb{P}_{n})[h(Z)]}{P_{0}[A]} \mid P \in H_{n}, Z^{(n)}\right), \mathcal{N}(0, \Vert \dot{\chi}_{P_{0}} \Vert^{2}_{P_{0}})\right)\xrightarrow[]{P_{0}} 0
\end{equation*}
so by Lemma \ref{lem::slutsky}, we deduce that
\begin{equation*}
    d_{BL}\left(\mathcal{L}\left(\left(\frac{P_{0}[A]}{\Tilde{P}[A]}-1\right)\frac{\sqrt{n}(\Tilde{P}-\mathbb{P}_{n})[h(Z)]}{P_{0}[A]} \mid P \in H_{n}, Z^{(n)}\right), \delta_{\{0\}}\right) \xrightarrow[]{P_{0}} 0.
\end{equation*}
Finally, we have the second term in $\circled{2}$. Using similar arguments to the above, we have that
\begin{equation*}
    \left(\frac{1}{\Tilde{P}[A]} - \frac{1}{\mathbb{P}_{n}[A]}\right)
\end{equation*}
converges conditionally in distribution to zero almost surely $[P_{0}^{(\infty)}]$. We also have
\begin{equation*}
    \sqrt{n}\mathbb{P}_{n}[h(Z)] = \sqrt{n}(\mathbb{P}_{n}-P_{0})[h(Z) - h_{0}(Z)] + \sqrt{n}P_{0}[h(Z)] + \sqrt{n}\mathbb{P}_{n}[h_{0}(Z)].
\end{equation*}
We have already established that the first two terms on the right converge uniformly to zero over $P \in H_{n}$ (the first one is in probability). The final term converges weakly to a normal distribution by the central limit theorem. Thus, by Lemma \ref{lem::slutsky}, we have
\begin{equation*}
    d_{BL}\left(\mathcal{L}\left(\left(\frac{1}{\Tilde{P}[A]} - \frac{1}{\mathbb{P}_{n}[A]}\right)\sqrt{n}\mathbb{P}_{n}[h(Z)]\mid P \in H_{n}, Z^{(n)}\right), \delta_{\{0\}}\right) \xrightarrow[]{P_{0}} 0.
\end{equation*}
We use Lemma \ref{lem::slutsky} once more to combine the results for the two terms in $\circled{2}$. \qed

\subsection{Proof of Theorem \ref{theo::1step_ACTT}} \label{sec::proof_ACTT}

Define
\begin{align*}
    g(Z) &= \left(A - \frac{(1-A)\pi(X)}{1-\pi(X)}\right)[Y-\mu^{(A)}(X)] + [A-\pi(X)]\{\mu^{(1)}(X) -\mu^{(0)}(X) - \chi(P_{0})\}\\
    g_{0}(Z) &= \left(A - \frac{(1-A)\pi_{0}(X)}{1-\pi_{0}(X)}\right)[Y-\mu_{0}^{(A)}(X)] + [A-\pi_{0}(X)]\{\mu_{0}^{(1)}(X) - \mu_{0}^{(0)}(X) - \chi(P_{0})\}.
\end{align*}
We begin by establishing the uniform $L_{2}$-convergence of $g$ on $H_{n}$:
\begin{equation*}
    \sup_{P \in H_{n}}\Vert g - g_{0} \Vert_{P_0} \rightarrow 0.
\end{equation*}
Using the definitions of $h$ and $h_{0}$ from the proof of Theorem \ref{theo::1step_att} in Section \ref{sec::proof_att}, $g$ and $g_{0}$ can be rewritten as
\begin{align*}
    g(Z) &= h(Z) -\pi(X)\{\mu^{(1)}(X) -\mu^{(0)}(X) - \chi(P_{0})\}\\
    g_{0}(Z) &= h_{0}(Z) -\pi_{0}(X)\{\mu_{0}^{(1)}(X) - \mu_{0}^{(0)}(X) - \chi(P_{0})\}.
\end{align*}
Under the assumptions, we already showed that $h$ converges uniformly in $L_{2}$ to $h_{0}$ on $H_{n}$. Thus, it suffices to show that
\begin{equation*}
    \sup_{P \in H_{n}}\Vert \pi\{\mu^{(1)} -\mu^{(0)}- \chi(P_{0})\} - \pi_{0}\{\mu_{0}^{(1)} - \mu_{0}^{(0)} - \chi(P_{0})\} \Vert_{P_0} \rightarrow 0.
\end{equation*}
The expression inside the norm can be decomposed into
\begin{equation*}
    \pi\{\mu^{(1)} -\mu^{(0)} - (\mu_{0}^{(1)} -\mu_{0}^{(0)})\} + (\pi - \pi_{0})\{\mu_{0}^{(1)} -Y + Y-\mu_{0}^{(0)} - \chi(P_{0})\}
\end{equation*}
and this has $L_{2}$-norm upper-bounded by
\begin{equation*}
    \Vert \mu^{(1)} -\mu_{0}^{(1)} \Vert_{P_0} + \Vert \mu^{(0)} -\mu_{0}^{(0)} \Vert_{P_0} + (2C+|\chi(P_{0})|)\Vert \pi -\pi_{0} \Vert_{P_0},
\end{equation*}
which converges uniformly to zero on $H_{0}$ by Assumptions \ref{ass::ACTT_bvm}(a) and (c).

To reduce verbosity and repetition, we will say hereafter that an expression is ``controlled'' if it converges conditionally to zero in probability on $H_{n}$. The ACTT one-step corrected parameter can be decomposed as follows:
\begin{align*}
    \Tilde{\theta} &= \theta(P) + \Tilde{P}\left(\frac{A}{\Tilde{P}[A]} - \frac{(1-A)\pi(X)}{\Tilde{P}[A][1-\pi(X)]}\right)[Y-\mu^{(A)}(X)]
    + \Tilde{P}\left[\left(\frac{A-\pi(X)}{\Tilde{P}[A]}\right)\{\mu^{(1)}(X) - \mu^{(0)}(X) - \theta(P)\}\right]\\
    &= \frac{\Tilde{P}[g(Z)]}{\Tilde{P}[A]}+\chi(P_{0})+\frac{\Tilde{P}[\pi(X)]\{\theta(P)-\chi(P_{0})\}}{\Tilde{P}[A]}
\end{align*}
We can further decompose $\Tilde{P}[g(Z)]$ into
\begin{align*}
    \Tilde{P}[g(Z)] &= [\Tilde{P}-\mathbb{P}_{n}][g(Z)]+ \mathbb{P}_{n}[h(Z)] - \mathbb{P}_{n}[\pi(X)\{\mu^{(1)}(X)-\mu^{(0)}(X)-\chi(P_{0})\}]\\
    &= [\Tilde{P}-\mathbb{P}_{n}][g(Z)]+ \mathbb{P}_{n}[h(Z)]-\mathbb{P}_{n}[\pi(X)]\{\theta(P)-\chi(P_{0})\},
\end{align*}
where the second equality follows from the definition of $\theta(P)$, which satisfies
\begin{equation*}
    \mathbb{P}_{n}[\pi(X)\{\mu^{(1)}(X)-\mu^{(0)}(X)-\theta(P)\}] = 0.
\end{equation*}
Then
\begin{align*}
    \sqrt{n}(\Tilde{\theta} - \hat{\chi}_{n}) &= \underbrace{\frac{\sqrt{n}[\Tilde{P}-\mathbb{P}_{n}][g(Z)]}{\Tilde{P}[A]}}_{\circled{1}}+\underbrace{\frac{\sqrt{n}\mathbb{P}_{n}[h(Z)]}{\Tilde{P}[A]}-\frac{\sqrt{n}\mathbb{P}_{n}[h_{0}(Z)]}{P_{0}[A]}}_{\circled{2}}-\underbrace{\frac{\sqrt{n}(\Tilde{P}-\mathbb{P}_{n})[\pi(X)]\{\theta(P)-\chi(P_{0})\}}{\Tilde{P}[A]}}_{\circled{3}}.
\end{align*}
We will show that term $\circled{1}$ converges conditionally in distribution to $\mathcal{N}(0, \Vert \dot{\chi}^{Y}_{P_0} +\dot{\chi}^{A}_{P_0} \Vert^{2}_{P_{0}})$ as required, while the other two terms are controlled. 

Term $\circled{1}$ can be decomposed into
\begin{equation*}
    \frac{\sqrt{n}[\Tilde{P}-\mathbb{P}_{n}][g(Z)]}{\Tilde{P}[A]} = \frac{\sqrt{n}[\Tilde{P}-\mathbb{P}_{n}][g(Z)]}{P_{0}[A]} + 
    \left(\frac{P_{0}[A]}{\Tilde{P}[A]}-1\right)\frac{\sqrt{n}[\Tilde{P}-\mathbb{P}_{n}][g(Z)]}{P_{0}[A]}.
\end{equation*}
For the first term on the right, we already showed that $g$ converges uniformly in $L_{2}$ to $g_{0}$ on $H_{n}$, and moreover,
\begin{equation*}
    \frac{g_{0}}{P_{0}[A]} = \dot{\chi}^{Y}_{P_0} +\dot{\chi}^{A}_{P_0}.
\end{equation*}
As argued in Theorem \ref{ass::1step_bvm}, the Donsker class assumption establishes the necessary Bayesian bootstrap Bernstein-von Mises theorem, which implies the convergence to $\mathcal{N}(0, \Vert \dot{\chi}^{Y}_{P_0} +\dot{\chi}^{A}_{P_0} \Vert^{2}_{P_{0}})$. For the second term, we already showed in the proof of Theorem \ref{theo::1step_att} that the factor $(P_{0}[A]/\Tilde{P}[A] - 1)$ converges conditionally in distribution to zero almost surely. Thus, by Lemma \ref{lem::slutsky}, this second term is controlled.

For term $\circled{2}$, we have
\begin{align*}
    \frac{\sqrt{n}\mathbb{P}_{n}[h(Z)]}{\Tilde{P}[A]}-\frac{\sqrt{n}\mathbb{P}_{n}[h_{0}(Z)]}{P_{0}[A]} &= \underbrace{\frac{\sqrt{n}(\mathbb{P}_{n}-P_{0})[h(Z)-h_{0}(Z)]}{P_{0}[A]}}_{\circled{A}}+\underbrace{\frac{\sqrt{n}P_{0}[h(Z)]}{P_{0}[A]}}_{\circled{B}}\\ &+\underbrace{\left(\frac{P_{0}[A]}{\Tilde{P}[A]}-1\right)\left(\frac{\sqrt{n}\mathbb{P}_{n}[h_{0}(Z)]}{P_{0}[A]}+\frac{\sqrt{n}(\mathbb{P}_{n}-P_{0})[h(Z)-h_{0}(Z)]}{P_{0}[A]} + \frac{\sqrt{n}P_{0}[h(Z)]}{P_{0}[A]} \right).}_{\circled{C}}
\end{align*}
Terms $\circled{A}$ and $\circled{B}$ were already shown to be controlled in the proof of Theorem \ref{theo::1step_att}. As before, the factor $(P_{0}[A]/\Tilde{P}[A] - 1)$ converges conditionally in distribution to zero almost surely. By the central limit theorem, $\sqrt{n}\mathbb{P}_{n}[h_{0}(Z)]$ converges weakly to a normal distribution with finite variance. The second and third terms in the second bracket of term $\circled{C}$ are terms $\circled{A}$ and $\circled{B}$ respectively. We deduce that term $\circled{C}$ is controlled as well.

Finally, for term $\circled{3}$, we have
\begin{equation*}
    \circled{3} = \frac{\sqrt{n}(\Tilde{P}-\mathbb{P}_{n})[\pi(X)]\{\theta(P)-\chi(P_{0})\}}{P_{0}[A]} + \left(\frac{P_{0}[A]}{\Tilde{P}[A]}-1\right)\frac{\sqrt{n}(\Tilde{P}-\mathbb{P}_{n})[\pi(X)]\{\theta(P)-\chi(P_{0})\}}{P_{0}[A]}.
\end{equation*}
Therefore, it suffices to show that $\sqrt{n}(\Tilde{P}-\mathbb{P}_{n})[\pi(X)]\{\theta(P)-\chi(P_{0})\}$ is controlled. Since $\pi$ is eventually contained in a fixed $P_{0}$-Donsker class and converges uniformly in $L_{2}$ to $\pi_{0}$, we have that $\sqrt{n}(\Tilde{P}-\mathbb{P}_{n})[\pi(X)]$ converges conditionally in distribution to $\mathcal{N}(0, P_{0}[(\pi_{0}-P_{0}[\pi_{0}])^{2}])$ in probability. 

It remains to study the convergence of $\theta(P)$ to $\chi(P_{0})$:
\begin{align*}
    \theta(P) - \chi(P_{0}) &= \frac{\mathbb{Q}_{n}[\pi\{\mu^{(1)} - \mu^{(0)}\}]}{\mathbb{Q}_{n}[\pi]} - \frac{P_{0}[\pi_{0}\{\mu_{0}^{(1)} - \mu_{0}^{(0)}\}]}{P_{0}[\pi_{0}]}\\
    &= \underbrace{\left(\frac{P_{0}[\pi_{0}]}{\mathbb{Q}_{n}[\pi]}-1\right)\frac{\mathbb{Q}_{n}[\pi\{\mu^{(1)} - \mu^{(0)}\}]}{P_{0}[\pi_{0}]}}_{\circled{I}} + \underbrace{\frac{\mathbb{Q}_{n}[\pi\{\mu^{(1)} - \mu^{(0)}\}]-P_{0}[\pi_{0}\{\mu_{0}^{(1)} - \mu_{0}^{(0)}\}]}{P_{0}[\pi_{0}]}}_{\circled{II}}.
\end{align*}
First consider term \circled{II}. The numerator can be decomposed into
\begin{align*}
    P_{0}[\pi_{0}]\circled{II} &= (\mathbb{Q}_{n}-P_{0})[\pi\{\mu^{(1)} - \mu^{(0)}\}] + P_{0}[\pi\{\mu^{(1)} - \mu^{(0)}\} - \pi_{0}\{\mu_{0}^{(1)} - \mu_{0}^{(0)}\}].
\end{align*}
Define the classes $\mathcal{H}_{n,1} = \{\pi: P \in H_{n}\}$, $\mathcal{H}_{n,2} = \{\mu^{(1)}-y: P \in H_{n}\}$ and $\mathcal{H}_{n,3} = \{\mu^{(0)}-y: P \in H_{n}\}$. These have envelope functions $1$, $C$ and $C$ respectively. The first and third classes are each eventually contained in fixed $P_{0}$-Glivenko-Cantelli classes by assumption. We apply Lemma 11 of \citet{Ray20} with the continuous function $\phi(h_{1},h_{2},h_{3}) = h_{1}(h_{2}-h_{3})$ to deduce that 
\begin{equation*}
    \sup_{P \in H_{n}}|(\mathbb{Q}_{n}-P_{0})[\pi\{\mu^{(1)} - \mu^{(0)}\}]|\xrightarrow[]{P_{0}^{*}} 0.
\end{equation*}
For the remaining term in the numerator of \circled{II}, consider the decomposition
\begin{align*}
    |P_{0}[\pi\{\mu^{(1)} - \mu^{(0)}\} - \pi_{0}\{\mu_{0}^{(1)} - \mu_{0}^{(0)}\}]| &= |P_{0}[\pi\{\mu^{(1)} - \mu^{(0)}-\mu_{0}^{(1)}+ \mu_{0}^{(0)}\}] - P_{0}[(\pi_{0}-\pi)\{\mu_{0}^{(1)}- Y + Y-\mu_{0}^{(0)}\}]|\\
    &\leq P_{0}|\mu^{(1)} -\mu_{0}^{(1)}| + P_{0}|\mu^{0)} -\mu_{0}^{(0)}| + 2C P_{0}|\pi -\pi_{0}| \\
    &\leq \Vert\mu^{(1)} -\mu_{0}^{(1)}\Vert_{P_0} + \Vert\mu^{0)} -\mu_{0}^{(0)}\Vert_{P_0} + 2C \Vert\pi -\pi_{0}\Vert_{P_0}.
\end{align*}
The final inequality follows from Cauchy-Schwarz. Thus, term \circled{II} is controlled. 

For term \circled{I}, the first bracket is upper-bounded by
\begin{align*}
    \left|\frac{P_{0}[\pi_{0}]}{\mathbb{Q}_{n}[\pi]}-1\right| &= \left|\frac{P_{0}[\pi_{0}] - \mathbb{Q}_{n}[\pi]}{\mathbb{Q}_{n}[\pi]}\right| \\
    &\leq \frac{1}{\delta}\left|P_{0}[\pi_{0}] - \mathbb{Q}_{n}[\pi]\right|\\
    &\leq \frac{1}{\delta}P_{0}\left|\pi_{0}-\pi\right| + \frac{1}{\delta}\left|(\mathbb{Q}_{n}-P_{0})[\pi]\right|\\
    &\leq \frac{1}{\delta}\Vert \pi_{0}-\pi\Vert_{P_0} + \frac{1}{\delta}\left|(\mathbb{Q}_{n}-P_{0})[\pi]\right|.
\end{align*}
The upper bound is controlled since $\pi$ converges uniformly to $\pi_{0}$ in $L_{2}$ on $H_{n}$, and $\pi$ is eventually contained in a $P_{0}$-Glivenko-Cantelli class. We can write term \circled{I} as
\begin{equation*}
    \circled{I} = \left(\frac{P_{0}[\pi_{0}]}{\mathbb{Q}_{n}[\pi]}-1\right) \left(\circled{II} +\frac{P_{0}[\pi_{0}\{\mu_{0}^{(1)} - \mu_{0}^{(0)}\}]}{P_{0}[\pi_{0}]}\right).
\end{equation*}
Since we have already showed that \circled{II} is controlled, we deduce that term \circled{I} is also controlled. This completes the proof. \qed

\end{document}